\documentclass[alpha-refs]{wiley-article}

\usepackage{siunitx}

\papertype{Original Article}
\paperfield{Journal Section}

\usepackage{fixltx2e,amsmath}
\usepackage{amssymb}
\usepackage{bm}
\usepackage{graphics,color} 
\usepackage{subcaption}
\usepackage{algorithm}
\usepackage{epstopdf}
\usepackage{color}

\MakeRobust{\eqref}

\def\d{\mbox{d}}
\def\half{\hbox{$1\over2$}}
\def\qrts{\hbox{$1\over4$}}

\title{On a Class of Objective Priors from Scoring Rules}

\abbrevs{ABC, a black cat; DEF, doesn't ever fret; GHI, goes home immediately.}

\author[1\authfn{1}]{Fabrizio Leisen}
\author[1\authfn{2}]{Cristiano Villa}
\author[2\authfn{2}]{Stephen Walker}

\contrib[\authfn{1}]{Equally contributing authors.}

\affil[1]{School of Mathematics, Statistics and Actuarial Science, Univesity of Kent, Canterbury, UK.}
\affil[2]{Department of Mathematics, University of Texas at Austin.}

\corraddress{Author One PhD, Department, Institution, City, State or Province, Postal Code, Country}
\corremail{correspondingauthor@email.com}

\presentadd[\authfn{2}]{Department, Institution, City, State or Province, Postal Code, Country}

\fundinginfo{Funder One, Funder One Department, Grant/Award Number: 123456, 123457 and 123458; Funder Two, Funder Two Department, Grant/Award Number: 123459}

\runningauthor{Author One et al.}

\begin{document}
\raggedbottom
\maketitle

\begin{abstract}
Objective prior distributions represent an important tool that allows one to have the advantages of using the Bayesian framework even when information about the parameters of a model is not available. The usual objective approaches work off the chosen statistical model and in the majority of cases the resulting prior is improper, which can pose limitations to a practical implementation, even when the complexity of the model is moderate. In this paper we propose to take a novel look at the construction of objective prior distributions, where the connection with a chosen sampling distribution model is removed. We explore the notion of defining objective prior distributions which allow one to have some degree of flexibility, in particular in exhibiting some desirable features, such as being proper, or centered on specific values which would be of interest in nested model comparisons. The basic tool we use are proper scoring rules and the main result is a class of objective prior distributions that can be employed in scenarios where the usual model based priors fail, such as mixture models and model selection via Bayes factors. In addition, we show that the proposed class of priors is the result of minimising the information it contains, providing solid interpretation to the method.

\end{abstract}

\section{Introduction}\label{sc_introduction}
With the ever increasing popularity of Bayesian methods, attributable largely to the advent of Markov chain Monte Carlo methods and other sampling techniques, the need for default, otherwise known as objective or noninformative, priors is also in demand. Model based objective priors, such as the reference prior \citep{BBS2009} and Jeffreys prior \citep{Jeff1961}, are commonly used when available. However, as models become larger and more complex, so it is that such priors are becoming more difficult to obtain, if not altogether unavailable. Indeed, it is our contention that model based objective priors have now reached their natural ceiling with little progress or advances in recent years. Some recent developments include a class of prior for hierarchical models, introduced by \cite{Simpson2017}, though these penalizing complexity priors are not considered objective in the usual sense. For a recent comprehensive review of objective Bayesian procedures we refer the reader to \cite{Consonni2018}.

Our observation is that limits to the progress in the research on objective priors is connected to their improperness. In fact, with very few exceptions, objective priors are improper. Although this may not represent a problem, as long as the posterior is proper, it causes severe limitations to the use of objective prior distributions, as we discuss in Section \ref{sc_improperness}. Indeed, the improperness of objective priors is the main motivation which brought us to investigate a novel approach to derive objective priors. 

Before proceeding further, it is important to provide some background. We argue that the Bayesian prior is given by the probability measure $\Pi$ on a suitable space of density functions and constructed from statistical model $f(\cdot|\theta)$, $\theta\in\Theta$, and probability density $p(\theta)$ on $\Theta$, via
$$\Pi(f \in A)=\int_{\{\theta:\,f(\cdot|\theta)\in A\}}p(\theta)\,d\theta,$$
for all measurable sets $A$. This is a more direct interpretation of the prior and avoids the inconvenient separation between Bayesian parametric and nonparametric methods. Indeed, it is useful to visualize Bayesian inference as the generation of a random density and associated functions; even if one is only considering random normal density functions. From this perspective, the two components that are traditionally known as the likelihood, $f(\cdot|\theta)$, and the prior, $p(\theta)$, are mere tools used for constructing $\Pi$ and, perhaps, the reasons for this nomenclature are more historical than accurate.  

Hence, the prior $\Pi$ can never be assigned completely on objective grounds, as the function $f(\cdot|\theta)$ is assigned as a result of a subjective choice, based on good reasons. However, the function $p(\theta)$ can be derived through some objective method, making the actual prior $\Pi$ a combination of subjective and objective choices. This said, it is puzzling to understand why $p(\theta)$ and $f(\cdot|\theta)$ need to feed off each other in any way; other than the $\theta$ in both needs to be sitting in the same parameter space $\Theta$.

Consequently, in this paper we investigate the possibility of defining objective prior distributions that are not model dependent and based on the sole knowledge of the parameter space $\Theta$. The $p(\theta)$ using only $\Theta$ loses the connection with the subjective component of $\Pi$ and could be argued as a consequence to be \textit{more objective}. Conversely, model based priors, such as the Jeffreys prior or the reference prior, necessarily include the subjective choice of the model. In fact, models are by and large misspecified and, consequently, model based priors are propagating this misspecification. So, while a model based prior reinforces the connection between the misspecified model and the prior itself, a prior that depends on the parameter space loses only the connection.

\subsection{Improper priors}\label{sc_improperness}
Most limitations to the use of model based objective priors lie in the fact that they are improper. In Section \ref{sc_motivexamples} we discuss some motivating examples where the lack of properness in the prior makes them unsuitable. Here, we would like to present more general issues widely discussed in the literature.

A thorough discussion of the problems that improper priors cause can be found in \cite{KassWass1996}, where they illustrate the following five issues:
\begin{enumerate}
\item Incoherence, strong inconsistencies and nonconglomerability;
\item Dominating effect of the prior;
\item Inadmissibility;
\item Marginalisation paradoxes;
\item Improper posteriors.
\end{enumerate}
Although points 1. to 4. are undoubtedly important and noteworthy, it is probably the last issue that requires careful consideration. The main concern is that, as of today, general results that allow one to assess if a given improper prior yields a proper posterior are yet to be found. Research has progressed on a case by case basis; for example, see \cite{IbraLaud1991} for the use of Jeffreys prior in generalised linear models, extended to overdispersed models of the same kind by \cite{Deyetal1993}, \cite{NataMcC1995}, \cite{BergerStraw1993}, and \cite{YC1995}. More recently, \cite{RubioSteel2018} describe general conditions to use improper priors for linear mixed models with longitudinal and survival data. However, the key point is that when one wishes to use (improper) objective priors, unless they have been used before for that specific case, extensive work is required to ensure that the posterior is proper. As one would expect, the task becomes more onerous the more complex the model is. But, even for simple models, the risk is high; as, for example, the discussion in \cite{ValleSteel2013} about the use of the Jeffreys rule prior for the Student-$t$ regression model derived in \cite{Fonseca2008} shows.

The above limitations imposed by improperness of model based objective priors lead us to the development of a method that allows objectivity of the prior and properness to coexist.

\subsection{Motivating examples}\label{sc_motivexamples}
The following four sections give an idea of the challenges that the use of improper priors may pose, even for simple problems, and highlight the need for objective prior distributions that have the flexibility, when needed, of being proper.
\subsubsection{Mixture models}\label{example_1}
A powerful tool in statistical analysis is represented by mixtures models. Due to their flexibility, mixtures of probability distributions allow  models suitable for complex data by building on simple components. As an example, consider a mixture of normal densities,
\begin{equation}\label{eq_mixtureintro}
f(x) =\sum_{j=1}^k \omega_j N(x|\mu_j,\sigma_j^2),
\end{equation}
where $k$ is a positive integer, including $\infty$, and the $(\omega_j,\mu_j,\sigma_j)$ are the collection of parameters of the mixture model. Even under the scenario when $k$ is known, the reference prior for model \eqref{eq_mixtureintro} has yet to be derived, and Jeffreys prior can only be obtained under specific conditions; see \cite{Grazian2015}. Furthermore, this type of model is subject to other issues related to non-identifiability and unbounded likelihoods, among others. The issues mainly rise from the fact that improper priors may not be appropriate as we might not observe outcomes from every component of the mixture \citep{Titt1985}. For example, \cite{Grazian2015} show that Jeffreys prior is suitable for mixtures of normal densities only in certain circumstances; that is, when the unknown parameters are the weights. If the unknown parameters are the means or the variances, then using Jeffreys prior may lead to improper posteriors. In particular, if the unknown parameters are the means only, proper posteriors exist only when the number of mixture components is at most two; while, if the unknown parameters are the variance, or the mean and the variances, then Jeffreys prior is not suitable for inference. The above issues can be generalised to apply to any type of mixture model.

\subsubsection{Bayes factors}\label{example_2}
Another simple case where objective priors are problematic is in model comparison (or selection) via Bayes factors. So, if we wish to compare model $M_1=\{f_1(x|\pmb\theta_1),p_1(\pmb\theta_1)\}$ to model $M_2=\{f_2(x|\pmb\theta_2),p_2(\pmb\theta_2)\}$, where both $\pmb\theta_1$ and $\pmb\theta_2$ are vector of parameters with some elements not in common, then the Bayes factor
$$B_{12} = \frac{\int f_1(x|\pmb\theta_1)p_1(\pmb\theta_1)\,d\pmb\theta_1}{\int f_2(x|\pmb\theta_2)p_2(\pmb\theta_2)\,d\pmb\theta_2},$$
is, in general, meaningful if the priors assigned to non-common parameters are proper. If not, then the arbitrary multiplicative constant up to which they are defined do not cancel and the Bayes factor depends on an arbitrary constant. Solutions to the issue have been proposed, see for example \cite{OHagan1995} and \cite{BergerPericchi1996}, however, the resulting procedures are still quite tedious to implement and are limited to simple models. By and large the above issue stays; however, \cite{Bergeretal1998} give an exception of the issue.

\subsubsection{Hierarchical models}\label{example_3}
Improper priors may not always be used with hierarchical models. Consider the following simple example from \cite{KassWass1996},
\begin{align*}
y_i|\mu_i,\sigma &\sim N(\mu_i,\sigma^2) \\
\mu_i|\tau &\sim N(\mu,\tau^2),
\end{align*}
for $i=1,\ldots,n$. Assuming $\sigma$ is known, one could adopt the objective prior $\pi(\mu,\tau)\propto\tau^{-1}$, which is improper. However, the (marginal) posteriors are improper. One way of overcoming this issue is to use proper priors that approximate the behaviour of the objective priors, which means proper and vague densities (or uniform distributions on a compact set); due to the arbitrariness of the choices involved, the above are obviously not viable solutions in an objective context.

\subsubsection{Other issues related to the use of improper priors}\label{example_4}
There are many issues, practical and foundational, deriving from the use of improper priors. For example, the marginalisation paradox \citep{StoneDavid72}, related to the use of improper priors for multi-dimensional parameter spaces, or the Stein's paradox \citep{BS1994}, related to the use of vague proper priors (or uniform priors on compact sets). Another issue goes under the name of \emph{strong inconsistency}, illustrated by the following example taken from \cite{Syver1998}. Consider observations from the normal density $N(\mu,\sigma^2)$, where the variance is known, and define the event $E = \{|\overline{x}|\geq\mu\}$. From the model, we have
$$P(E|\mu) = \frac{1}{2} + \Phi\left(-2|\mu|\sqrt{n}/\sigma\right)>\frac{1}{2}.$$
Since $P(E|\mu)>1/2$ for all the values of $\mu$, we conclude that $P(E)>1/2$, and the posterior for $E$, assuming $p(\mu)\propto1$, is given by
$$P(E|x) = \frac{1}{2} - \Phi\left(-2|\overline{x}|\sqrt{n}/\sigma\right)<\frac{1}{2}.$$
As $P(E|x)<1/2$ for all values of $x$, we conclude that $P(E)<1/2$, showing inconsistency between the sampling distribution and the posterior.

\subsection{The idea:  Overview}\label{sc_ideaintro}
In this section we give an overview of the idea we propose to derive a class of objective prior distributions that depends on the parameter space only. While a formal presentation will be discussed later on, we deem it appropriate to present, at least at an intuitive level, the ideas here.

The key to the idea is to consider a loss function $l(\theta,p(\theta))$ which penalizes for each $\theta\in\Theta$ a choice of a prior density $p(\theta)$. The objective criterion is then based on the idea of finding the class of $p$ which makes
$l(\theta,p(\theta))$ constant. For obvious reasons, the loss function should have the following property 
\begin{equation}\label{eq:propersocring}
\int l(\theta,p(\theta)) q(\theta)\, \d\theta \geq \int l(\theta,q(\theta))q(\theta)\, \d\theta,
\end{equation}
for all $q$'s representing a density for the $\theta$. In other words, if a ``true'' density for $\theta$ exists, the expected loss should be minimised when such a density is chosen. The condition in \eqref{eq:propersocring} identifies a particular class of loss functions, known as \emph{proper scoring rules}. One way of interpreting (proper) scoring rules is as loss functions that measure the quality of a quoted density $p$ for an uncertain quantity $\theta$; see for example \cite{Parry2012}. We indicate proper scoring rules as $S(\theta,p)$, and we ask it to be constant for all $\theta\in\Theta$. So we set
$$S(\theta,p) = \mbox{constant} \qquad \forall\theta\in\Theta,$$
and the densities satisfying the above equality identify a class of objective priors. We set the constant to 1 and show later that this choice is without loss of generality. The criterion defining this class of priors is clearly objective, for if the scoring rule were not constant, some parts of the space $\Theta$ would be given preference above others.

As discussed in \cite{Parry2012}, any proper scoring rule is equivalent to the \emph{log score}, $-\log p(\theta)$, also known as the self information loss function. The log score has the property of depending on $p$ only through its value at $\theta$, which is known as the \emph{local} property. However, the above scoring rule lacks the flexibility for assigning an objective prior. For if we set $-\log p(\theta)=\mbox{constant}$, we only achieve $p(\theta)\propto 1$.

To solve this, we consider additionally the Hyv\"{a}rinen scoring rule \citep{Hyva2005}, which makes use of the first two derivatives of $p$, written as $p'$ and $p''$. We then have a scoring rule $S(\theta,p)$ which has two components; the log score and the Hyv\"{a}rinen score. Finding solutions to $S(\theta,p)=1$ will now involve solving a second order differential equation and we obtain the class of prior through the two constants connected with the two derivatives.

While the prior is deemed objective through the setting of the scoring rule to be a constant, we show that the prior distribution which solves $S(\theta,p)=1$ has an alternative derivation using variational methods; that is $S(\theta,p)=1$ is a solution to the Euler--Lagrange equation for minimising
$\int_\Theta L(\theta,p,p')\,\d\theta$,
where $L$ represents information in $p$, indeed a combination of the differential entropy information, given by $\int p(\theta)\log p(\theta)\,\d\theta$, and the Fisher information, given by $\int p^\prime(\theta)^2/p(\theta)\,\d\theta$. There is then an elegant alternative interpretation of obtaining the class of objective prior, which involves information aspects of the prior distribution itself.

In higher dimensions, where the parameter space is $(\theta_1,\ldots,\theta_k)$, and there are no constraints between parameters, we assume prior independence among the components and construct the prior as
$$p(\theta_1,\ldots,\theta_k) = \prod_{j=1}^k p_j(\theta_j),$$
where each $p_j(\theta_j)$ is a prior derived with our proposed approach. 

While we could theoretically obtain the multivariate prior it is convenient to assume independence, as is commonly done with objective priors, for example as with the independent Jeffreys prior. Indeed, we argue that in the absence of any prior information about possible constraints between the parameter spaces, the assumption of independence is an appropriate representation of absence of information. Note that, although reference priors allow us to obtain priors for multidimensional parameter spaces, they have the downside, when it is possible to find them, in that the yield a prior which depends on the order in which the parameters are considered. Furthermore, it is not uncommon to see the use of independence Jeffreys prior for multidimensional cases, as in \cite{Fonseca2008} or \cite{RB2014}, given that it can still remain the best (or only) option in this situations. Nevertheless, in Section 2 we do describe the full multivariate solution.

\subsection{Invariance}\label{sc_invariance}
A fundamental point of discussion about prior distributions and, in particular, objective prior distributions, is invariance. Indeed, Jeffreys' rule to derive a prior distribution for the parameters of a given model is based on an invariance requirement, in particular on invariance under one-to-one reparameterizations. Also, other common objective priors, such as reference priors, have been shown to be invariant and the same apply, for example, to the priors in \cite{Simpson2017}.

Here we discuss invariance from two opposite perspectives: that it is not important, and that it is important.
Before discussing this apparent contradiction, we need to point out that we define the objective prior by setting the scoring rule equal to a constant, that is $S(\theta,p(\theta))=\mbox{constant},$
is invariant under location transformations.

The core of the discussion about invariance revolves around transformations that are not of the location type. In general, we have the choice of the model $f(\cdot|\theta)$ defining the parameter space $\Theta$ which, in turn, defines the prior $p_\theta(\theta)$. For a non location transformation, say $\phi=\phi(\theta)$, it is that $f(\cdot|\phi)$ defines $p_\phi(\phi)$ via $S(\phi,p_\phi(\phi))=\mbox{constant}$ and, in general
\begin{equation}\label{eq:invariance}
\int_{\{\theta:\,f(\cdot|\theta)\in A\}}p_\theta(\theta)\, \d\theta \ne \int_{\{\phi:\,f(\cdot|\phi)\in A\}}p_\phi(\phi)\, \d\phi,
\end{equation}
so the prior distributions under the two scenarios are different. As an example, consider the transformation from a probability (e.g. the parameter of a Bernoulli distribution) to negative log; i.e. $\phi=-\log\theta$, for which the parameter space changes from $(0,1)$ to $(0,\infty)$.

The question is whether this has any practical implications given that only one of the two priors in \eqref{eq:invariance} will be used. Current model based objective procedures are bound to throw away some coherence properties to achieve invariance, see \cite{KassWass1996}. However, our point is that there is no practical consequence of any relevance arising from the lack of invariance, given that, as mentioned above, a single parameterization will be used. For example, in the case the chosen model is the normal density, one either considers the precision parameter or the variance parameter, not both.

The above points of discussion are concerned with the perspective that invariance is not important. To consider the opposite point of view let us assume that there is a canonical parameterization for the model $f(\cdot|\theta)$. Certainly, for most models the set of parameters for which priors would be assigned is obvious. For example, the exponential family has
$$f(x|\theta)=h(x)\,\exp\left\{\sum_{j=1}^p \theta_j\,T_j(x)-A(\theta)\right\}$$
with $\theta=(\theta_1,\ldots,\theta_p)$ being the canonical parameterization. 
We can then define the canonical objective prior for statistical model $f(\cdot|\theta)$, $\theta\in\Theta$, as
$$p_{\Theta}(\theta)=\prod_{j=1}^p p_{\Theta_j}(\theta_j),$$
where
$\Theta=\otimes_{j=1}^p \Theta_j.$ Then, any transformed prior can be obtained in the usual way involving variable transformations; that is
$p(\phi)=|J|\,p_\Theta(\theta(\phi))$, where $J$ is the Jacobian matrix for the transformation.

\subsection{Objectivity and uniqueness}
Before proceeding further it is important to discuss uniqueness and flexibility associated with objective priors.
It is widely acknowledged that a prior representing total ignorance is elusive (and it might not even be possible to obtain in principle, see \cite{BS1994}).  As a consequence, any prior distribution, objective or not, must out of necessity provide some knowledge about something, and this ``something'' is not necessarily unique.  For example, given a particular problem the corresponding objective prior over a given parameter space could be proper or improper; differentiable everywhere or not; convex; log-concave; etc. In other words, a prior can be objective and exhibit desirable features of choice without impinging on subjective components relating to information.

So while we will be introducing a Bayesian objective prior criterion, it does not lead to a unique prior, rather to a class of priors, where some desirable features may or may not be included. We believe that this level of flexibility is a point of strength of the proposed approach, making it adaptable to different scenarios, including those where model based priors do not work. Indeed, model based priors may lead to uniqueness, yet they provide features, mostly improperness, which may not actually be desirable. For example, Jeffreys prior for the parameter of a Bernoulli distribution has the feature of having spikes at 0 and 1.  On the other hand, the class of objective priors that are defined by our method contains  both proper and improper priors so, for example, the objective prior on a location parameter can be the flat prior (i.e. the usual objective prior obtained, for example, by applying Jeffreys method) but can be proper as well (for example, to be used for the location parameters in mixture models).

The definitions of objective priors are vague and certainly encompass our choice.  In \cite{Berger2006} the following four definitions of \emph{objective Bayes} are listed in order of generality;
\begin{enumerate}
\item A major goal of statistics (indeed science) is to find a completely coherent objective Bayesian methodology for learning from data.
\item Objective Bayesian analysis is the best method for objectively synthesizing and communicating the uncertainties that arise in a specific scenario, but is not necessarily coherent in a more general sense.
\item Objective Bayesian analysis is a convention we should adopt in scenarios in which a subjective analysis is not tenable.
\item Objective Bayesian analysis is simply a collection of \textit{ad hoc} but useful methodologies for learning from data.
\end{enumerate}
The claim is that 1. is not attainable; 2. is achievable (though not always); while 3. should always hold and should be always implemented.

If we consider the above definitions given by \cite{Berger2006}, we note the possibility of some degree of flexibility in an objective approach. As a starting point of discussion, let us consider the properness of the prior and let us assume that we are able to obtain, through the same procedure, a proper as well as an improper prior for a given parameter space; say, the mean of a normal density when the variance is known. As mentioned above, Jeffreys (and the reference) prior propose the flat prior $p(\mu)\propto1$. This solution is in general sensible and it would lead to a proper posterior. However, as discussed in the Example \ref{example_1} above, the flat prior will not be applicable to a mixture model with more than two components. In this scenario, it would be desirable (and sensible) to have the flexibility of using the same objective criterion and obtain a proper prior. In fact, if we consider the possibility of having a proper prior by leveraging on the flexibility of an objective method we surely satisfy the points from \cite{Berger2006}. With respect to the above four descriptions of objective Bayes, we see that points 3. and 4. are quite obvious, as the methodology will surely be \textit{useful} and it would be employed in a scenario where prior information is not available (or cannot be used). The choice of properties in a prior, such as properness, is not strictly a subjective action; in particular if it is driven by common sense (i.e. there are no alternatives). Furthermore, we argue that flexibility fits into point 2. as it is objective within a specific scenario. If we need proper objective priors for the means on the components of a mixture model, then that particular choice is objective.


\subsection{Literature review}\label{sc_background}
The development of  objective priors has been prolific in recent years. The idea is that in a scenario where prior elicitation is not feasible, or not desirable, a prior distribution can be formed through structural or formal rules \citep{KassWass1996}. Although a thorough review of objective Bayesian methods is beyond the scope of this paper, we deem it appropriate to briefly list the most common proposals. The first approach is due to \cite{Lap1820} with the principle of insufficient reason which leads to uniform, or flat, prior distributions.
The most popular objective prior is Jeffreys prior \citep{Jeff1946,Jeff1961}, who proposed a prior distribution for continuous parameter spaces which is invariant for one--to--one transformations of the parameter space. For example, if we consider the sampling distribution $f(x|\theta)$, the corresponding Jeffreys prior is given by
$p(\theta) \propto \sqrt{I(\theta)},$
where $I(\theta)$ is the Fisher information. Although in scenarios where there is only one parameter of interest, Jeffreys prior yields sensible posterior distributions. However, in cases where the parameter space has a dimension of two or more, the prior is known to yield posteriors with poor performance (sometimes giving paradoxical results, such as the marginalisation paradox). In such cases the priors are taken to be independent. 

Although other more general invariance priors have been proposed, such as in \cite{Dawid1983}, \cite{Harti1964} and \cite{Jaynes1968}, the reference prior of Bernardo represents an alternative to Jeffreys prior.
Here the idea is to derive a prior distribution which carries as minimal information as possible. The prior is identified as the one which maximises the (expected) missing information between the prior and the posterior. The most up--to--date results on reference priors can be found in \cite{BBS2009} and, for an extension to discrete parameter spaces, in \cite{BBS2012}. A limitation of reference priors is sensitivity to the order of importance of parameters; this issue and possible solutions have been discussed in \cite{BBS2015}.

Other objective priors proposed include that of \cite{BT1973}, based on data--translated likelihoods, and maximum entropy priors, see for example \cite{Jaynes1957,Jaynes1968}. The first type aims to use uniform priors in models where the likelihood can be translated producing posteriors which, for different samples, have the same shape and differ in location only. As discussed in \cite{Kass1990}, these priors turn out to be very restrictive.

Another important class of objective priors are the probability matching priors, first proposed in \cite{WP1963}. The aim is to obtain a prior distribution under which the posterior probabilities of certain regions coincide with their coverage probabilities, either exactly or approximately. For example, if we consider the model $f(x|\theta)$, and $t(p,\alpha)$ is the $\alpha$--quantile of the posterior, and
$$P_f\{\theta\leq t(p,\alpha)|x\} = \int_{-\infty}^{t(p,\alpha)} p(\theta|x)\, \d\theta=\alpha,$$
then $p(\theta)$ is a probability matching prior. Recent developments of this method can be found in \cite{Sweetetal2006} and \cite{Swee2008}.

A different method, based on information theoretical concepts, has been proposed by \cite{ZellMin1993}, giving the so called maximal data information prior. Although the method gave rise to some interesting results, such as the derivation of the right-Haar measure for location-scale problems, applications remain limited.

Possibly, the most recent development in defining prior distributions, although not in a strictly objective sense, is discussed in \cite{Simpson2017}. The idea is to identify the parts in a complex model that require subjective input, while the remaining can be associated to non-informative priors. The comparison between a simple model, say $f_0(\cdot|\eta_0)$ and a richer and more flexible alternative, say $f(\cdot|\eta)$, is done by assigning a prior on $\eta$ that penalises for complexity (on the basis of the Kullback--Leibler divergence between the two models). The other parameters (may) have objective priors assigned upon. 

A final consideration is reserved for discrete parameter spaces, whose systematic discussion can be seen to be generated by the paper of \cite{Rissan1983}. The lack of general methods, due to the challenges that discreteness imposes, has been filled by \cite{BBS2012} first, and by \cite{VW2015} later.

As previously mentioned, for a recent and thorough review of the objective Bayesian approaches so far developed, we refer the reader to \cite{Consonni2018}.

\subsection{Organisation of the paper }\label{sc_organisation}
The paper is organised as follows. In Section \ref{sc_scoringrules} we introduce the foundations of the proposed prior on the basis of scoring rules and their properties.  An interesting aspect of the prior based on scoring rules is its interpretation in terms of the information content carried by the prior itself. This aspect is explored in Section \ref{sc_variational}. In Section \ref{sc_illustrations} we present the objective priors concentrating on $\Theta=(0,1), \,\,(0,\infty)$ and $(-\infty,+\infty)$. The implementation of the prior for some specific applications is presented in Section \ref{sc_application}.
Finally, Section \ref{sc_discussion} is dedicated to some final remarks.

\section{Priors from scoring rules}\label{sc_scoringrules}
Let us consider a quantity of interest, $\theta$, which can take values in the space $\Theta$. The fundamental argument behind objective prior distributions is that they should represent a state of actual or alleged prior ignorance about the true value of $\theta$. Several criteria have been proposed to select such a prior, all of which assume that a probabilistic model generating the data (given $\theta$) has been chosen. What we propose is to avoid this choice and derive a prior depending on $\Theta$ only. The idea is to measure the quality of the prior $p$ with a proper scoring function, say $S(\theta,p)$, and assume it to be constant, as discussed in the Introduction.
\begin{definition}\label{def_globalprior}
A density $p$ with respect to the Lebesque measure on $\Theta$, is \emph{objective} (in accordance with commonly accepted meaning of the expression) if
$S(\theta,p)=1$ for all  $\theta\in\Theta,$
where $S$ is a proper scoring rule.
\end{definition}

Before proceeding we provide a brief discussion on scoring rules. Scoring rules are \textit{proper} if 
$\int_\Theta S(\theta,p)\,q(\theta)\,\d\theta$ is minimized at $p=q$ and \textit{local} if it depends on $p$ only through the value $p(\theta)$. The unique proper local scoring rule is the  log score, defined as
\begin{equation}\label{eq_logscore}
S_L(\theta,p) = -\log p(\theta).
\end{equation}
\cite{Parry2012} extend the local property to $m$--local, in that now $S(\theta,p)$ depends also on the $l$-derivative $p^{(l)}(\theta)$, for $0\leq l\leq m$. In particular, for $m=2$, there is the  Hyv\"{a}rinen scoring rule, \citep{Hyva2005}, given by
\begin{equation}\label{eq_hyvascore}
S_H(\theta,p) = \frac{\partial^2}{\partial\theta^2}\log p(\theta) + \frac{1}{2}\left\{\frac{\partial}{\partial\theta}\log p(\theta)\right\}^2.
\end{equation}
Consequently, our choice of scoring rule is
\begin{eqnarray}\label{eq_globalscoreext}
S(\theta,p) &=& S_L(\theta,p)+S_H(\theta,p)\nonumber\\
&=& -\log p(\theta) + \frac{p''(\theta)}{p(\theta)} - \frac{1}{2}\left\{\frac{p'(\theta)}{p(\theta)}\right\}^2.
\end{eqnarray}
That this is a proper scoring rule is derived from the fact that it is a sum of two proper scoring rules. It is also clearly $2$--local. 
Previously, priors have been sought based solely on $\log p$; for example, the reference prior, and the math becomes unnatural as a consequence. On the other hand, including higher derivatives yields well defined solutions to optimization procedures. That we set this score to 1 for all $\theta$ is done without loss of generality, as we shall see later on. That we understand this to be an objective procedure is evident from the fact that no part of $\Theta$ is being given preference; the {\em loss} at $\theta$ for our choice of $p(\theta)$ is the same for all $\theta$. For, if $S(\theta,p)$ did depend on $\theta$ then we argue that this could only be driven by {\em information}; i.e. parts of $\Theta$ space are preferential to others.

It is to be noted that we have summed the two scores directly without introducing any weighting; for we could have used $S_w(\theta,p)=S_L(\theta,p)+w\,S_H(\theta,p)$.
The choice of $w=1$ is a calibration issue between the two scores; i.e. to put them on a comparable scale. The reason for $w=1$ is that for the benchmark standard normal density function, i.e. $p(\theta)\propto e^{-\frac{1}{2}\theta^2}$, the difference between the scores $S_L$ and $S_H$ is a constant (i.e. does not depend on $\theta$), and so one does not end up dominating the other, only for $w=1$.

Hence, we see that the objective prior $p(\theta)$ is obtained by solving the following differential equation:
\begin{equation}\label{eq_diff}
\frac{p''(\theta)}{p(\theta)}-\frac{1}{2}\left\{\frac{p'(\theta)}{p(\theta)}\right\}^2= 1+\log p(\theta).
\end{equation}
To derive the solution, we have the following result.
\begin{theorem}\label{lem_diffequationequivalency}
The solution to (\ref{eq_diff}) is given by $p(\theta)\propto e^{-u(\theta)}$ with $u$ solving
$$u'(\theta) = \pm\sqrt{ce^{u(\theta)}-2\{1+u(\theta)\}},$$
for some suitable constant $c$.
\end{theorem}

\begin{proof}
Solving the differential equation \eqref{eq_diff} is equivalent to solving the following differential equation;
\begin{equation}\label{eq_diffequationequivalent}
u(\theta)''-\frac{1}{2}\{u(\theta)'\}^2 = u(\theta),
\end{equation}
where we have $u(\theta) = -\log p(\theta)$.
Strictly we have $u''-\half (u')^2=u-1$ but if $u$ solves this then $u+1$ solves (\ref{eq_diffequationequivalent}). The 1 here is the same one which appears in the $S(\theta,p)=1$ and hence confirming the ``without loss of generality'' in the choice of 1 as a constant.
By now letting $v=u'$ we have
$$v\frac{dv}{du} = u + \frac{1}{2}v^2,$$
which has the solution
\begin{equation}\label{eq_solutiondiffeq}
v(\theta) = u'(\theta) = \pm\sqrt{ce^{u(\theta)}-2\{1+u(\theta)\}},
\end{equation}
for a suitable $c$.
\end{proof}

The missing pieces in equation \eqref{eq_solutiondiffeq} are $c$ and say $u(0)$, two constants of integration. We will see how to complete these when we look at illustrations in Section 4. In general, as the solution depends on the above two arbitrary constants, our method provides a class of solutions, where some are proper and some are improper and, more general, where the priors will have some assigned properties via specification of $(c,u(0))$.

We also note here that we do not need the normalizing constant for $p$ and neither do we need to find an explicit solution for $u$, and $p$, beyond (\ref{eq_solutiondiffeq}). The reason for this is that we can find an accurate solution via numerical methods; i.e. if we have $u(\theta)$ at a particular $\theta$ value, then  we can evaluate
$$u(\theta+\varepsilon) = u(\theta) + \varepsilon u'(\theta) + \frac{1}{2}\varepsilon^2u''(\theta)+o(\varepsilon^2)$$
for small $\varepsilon$, and the $u'$ and $u''$ are available explicitly, with $u'$ as in equation \eqref{eq_solutiondiffeq}, $u''=\half c e^u-1$, and with the ease of obtaining higher derivatives if needed. From here we can evaluate $p(\theta)$. 

Finally, in this section, we describe the multidimensional solution. So suppose that $\theta=(\theta_1,\ldots,\theta_k)$. Without needing to replicate the mathematics, we can easily derive the solution as
$p(\theta)\propto \exp(-u(\theta))$ with
$$\partial u/\partial\theta_j=\pm\,\sqrt{c_j\,e^{u(\theta)}-2(1+u(\theta))},\quad j=1,\ldots,k,$$
where the additional  free parameters are the $(c_j)$. 
We can again solve for $u$ using the multivariate version of the Taylor approximation.

\section{Variational problems and solutions}\label{sc_variational}
Here we provide an alternative derivation of (\ref{eq_diff}) using information theory, specifically entropy information and Fisher information.  We show that the $p$ solving \eqref{eq_diff} can also be regarded as a density carrying minimal local information.  This material then is to provide support for the solution to (\ref{eq_diff}) being an objective prior.

The entropy information (negative entropy) of a density function $p$ is given by
$$I_E(p) = \int p(\theta)\log p(\theta)\,\d\theta,$$
which is related to Shannon's entropy and is equal to negative the expected self--information loss. In addition to $I_E(p)$, we consider a measure of the information in the density $p$ known as Fisher information, given by
\begin{equation}\label{eq_I_F}
I_F(p)= \int \frac{p'(\theta)^2}{p(\theta)}\,\d\theta =\int p(\theta)\left\{\frac{\partial}{\partial\theta}\log p(\theta)\right\}^2\,\d\theta.
\end{equation}
See for example \cite{Bobkov2014}.

Now consider $I(p)=I_E(p)+\half I_F(p)$ and the aim is to find the $p$ which minimizes $I(p)$. Recalling variational methods \citep{Rustagi1976}, if we wish to minimise $\int_a^b L(\theta,p,p')\,\d\theta$,  a necessary condition for a local extremum of the integral of the Lagrangian $L(\theta,p,p')$ is that
\begin{equation}\label{eq_geuler}
\frac{\partial L}{\partial p}=\frac{\d}{\d\theta}\,\frac{\partial L}{\partial p'}.
\end{equation}
Minimising $\int_a^b L(\theta,p,p')\,\d\theta$ reduces to the classical calculus of variation problem where we want to extremize the integral of the function 
\begin{equation}\label{lagrange}
L(\theta,p,p')= \frac{1}{2}\frac{p'(\theta)^2}{p(\theta)}+p(\theta)\,\log p(\theta).
\end{equation}
The solution to the extremal problem, if it exists, is obtained from the Euler--Lagrange equations, given by (\ref{eq_geuler}).
According to page 44 of \cite{Rustagi1976}, if $L(p,p')$ is strictly convex on $(0,\infty)\times (-\infty,+\infty)$, and $p$ satisfies the Euler equation, then $p$ is a minimum of 
$\int_a^b L(\theta,p,p')\,\d\theta$.
Now $L(p,p')$ is strictly convex if the matrix
$$H=\left(\begin{array}{ll}
\partial^2 L/\partial p^2  &  \partial^2 L/\partial p\partial p' \\ \\
\partial^2 L/\partial p\partial p'  &  \partial^2 L/\partial (p')^2 \end{array}
\right)$$
is positive definite. 

\begin{theorem}\label{teo_variations}
A minimum satisfying the Euler--Lagrange equations is given by the $p$ solving the differential equation
$$p'=\pm p\,\sqrt{c/(e\,p)+2\log p},$$
for some suitable $c$.
\end{theorem}

\begin{proof}
Calculations give
$$H=\frac{1}{p}\left(\begin{array}{ll}
\kappa^2+1  &  -\kappa \\ \\
-\kappa  &  1 \end{array}\right)
$$
where $\kappa=p'/p$. This is easily seen to be a positive definite matrix; the eigenvalues are given by
$$\frac{1}{p}\left[\half (2+\kappa^2)\pm\sqrt{\qrts(2+\kappa^2)^2-1}\right],$$
which are positive.

Then equation \eqref{eq_geuler}, after some elementary algebra and differentiation, leads to the following differential equation,
$$
\frac{p''(\theta)}{p(\theta)}-\frac{1}{2}\left\{\frac{p'(\theta)}{p(\theta)}\right\}^2= 1+\log p(\theta),
$$
which is the same as (\ref{eq_diff}).
\end{proof}
\noindent
This differential equation has the solution derived in the previous section. It is interesting that the Euler--Lagrange equations are solved by precisely the same $p$ solving equation (\ref{eq_diff}).
 
\section{Illustrations}\label{sc_illustrations}
To illustrate the proposed method we consider three common parameter spaces. In particular, we consider the space for a parameter representing a probability, that is $\Theta=(0,1)$, the space $\Theta=(0,\infty)$, usually representing the support of scale parameters and, finally, the support for (location) parameters $\Theta=(-\infty,+\infty)$. The aim here is to solve \eqref{eq_geuler} for particular motivated choices of $(c,u(0))$, equivalently, $(c,p(0))$ or $(p'(0),p(0))$. In fact, the solutions to the Euler--Lagrange equations are many, and the choice of the two constants $(c,u(0))$ will then determine a unique solution.

Now there is the flat solution for all $\Theta$ in \eqref{eq_diff} given by $p(\theta)\propto 1$. This is achieved by setting $c=2$ and $u(0)=0$.  However, in each of the settings of $\Theta$ considered we can find alternate priors with particular features. So, e.g. for $\Theta=(0,1)$ we ask that $p(0)=p(1)=0$ and for $\Theta=(0,\infty)$ we ask that $p$ is convex and decreasing.


\paragraph{Case $\bm{\Theta=(0,1)}$.} Here we consider the $u$ function, recall $p\propto e^{-u}$, and so for $p(0)=p(1)=0$ we require $u(0)=u(1)=\infty$. For additional symmetry, we can take $u\left(\half\right)=w>0$ and taking $c=2$ as the extremal value, we have
$$
\begin{array}{ll}
u'=\sqrt{2}\sqrt{e^u-1-u} & \theta>\half \\
u'=-\sqrt{2}\sqrt{e^u-1-u} & \theta<\half.
\end{array}
$$
Note there is a discontinuity in the derivative of $u$ at $\theta=\half$. As $w$ increases it is that $u(0)$ and $u(1)$ got to $\infty$. A plot of $u$ is given in Fig.~\ref{fone} for $w=1.1$ and in Fig.~\ref{ftwo} for $w=1.14$. For these figures we used a grid of 1000 either side of $\theta=\half$ to obtain the numerical solutions. In the latter case, the corresponding density for $p$ is presented in Fig.~\ref{fthree}.

\begin{figure}[H]
\begin{center}
\includegraphics[scale=0.4]{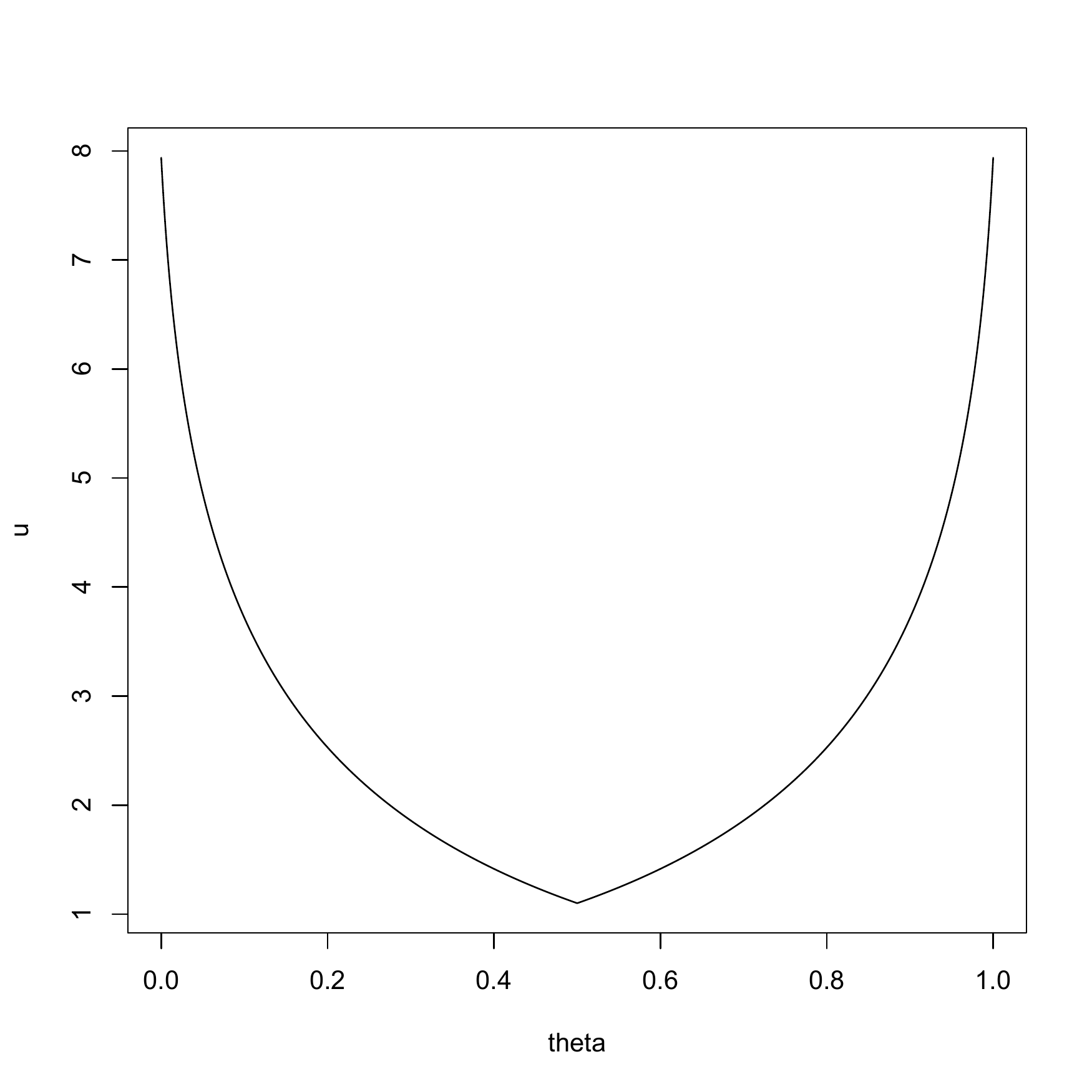}
\caption{Numerical solution for $u$ with $w=1.1$}
\label{fone}
\end{center}
\end{figure}

\begin{figure}[H]
\begin{center}
\includegraphics[scale=0.4]{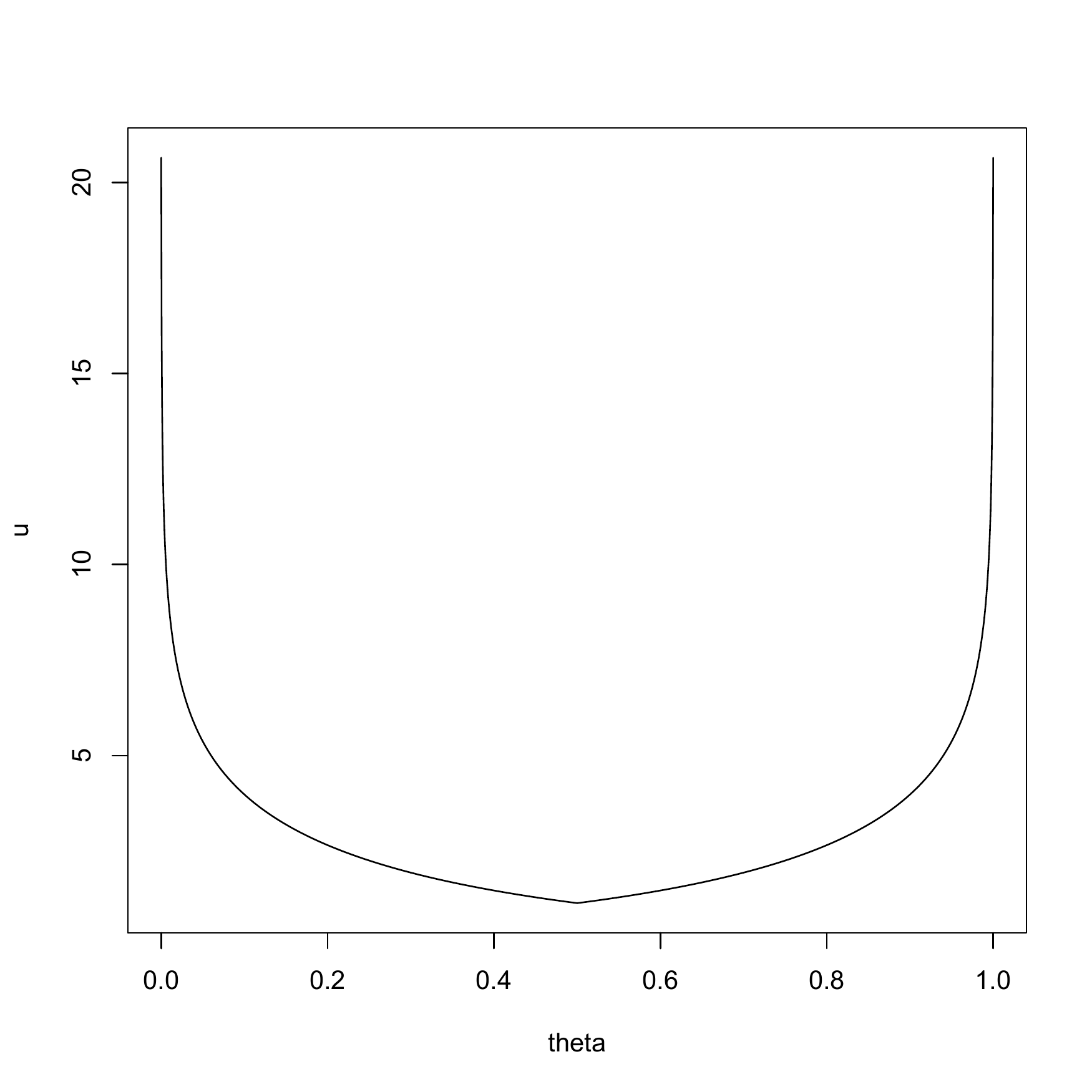}
\caption{Numerical solution for $u$ with $w=1.14$}
\label{ftwo}
\end{center}
\end{figure}

\begin{figure}[H]
\begin{center}
\includegraphics[scale=0.4]{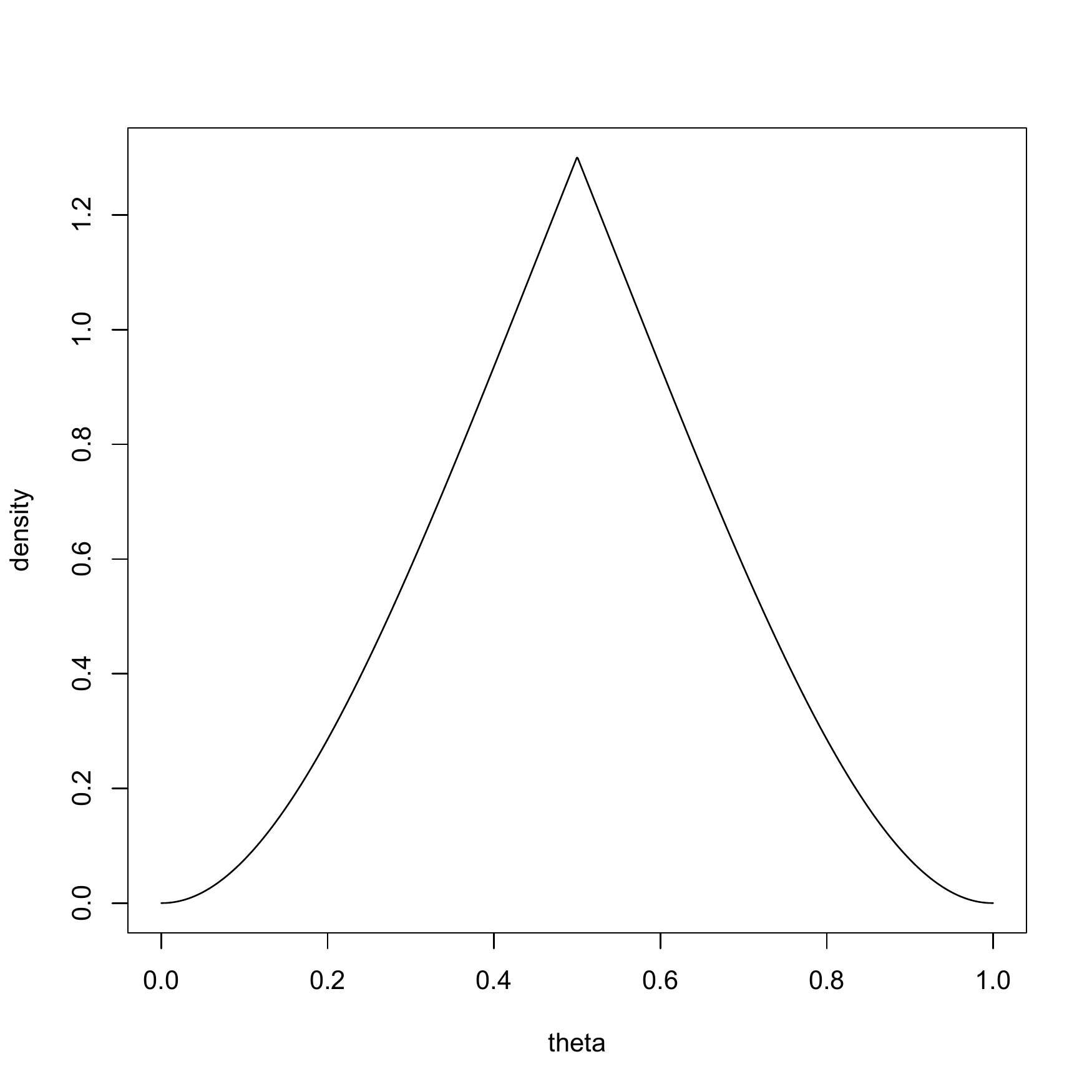}
\caption{Numerical solution for normalized $p$ with $w=1.14$}
\label{fthree}
\end{center}
\end{figure}
It also possible to obtain a prior that mimics Jeffreys'; that is, a distribution that has spikes at $\theta=0$ and $\theta=1$ with the lowest value at $\theta=\half$. This is done by simply inverting $u$: i.e. set $u=-u$ in the above prior.

\paragraph{Case $\bm{\Theta=(0,\infty)}$.} For a prior defined on the space $(0,\infty)$ we require a specific shape property for $p$ (convex and decreasing) and then take extremal values for the $c$ and $u(0)$. This property  is common to most objective priors on $(0,\infty)$. Thus, since $p'<0$ we require $u'>0$ and so
$u'=\sqrt{ce^u-2(1+u)}$, and for $u'$ to exist for all $u$ we must have $c\geq 2$. Thus, as an extremal value, we take $c=2$.

For the prior to be convex we require $p''\geq0$. Now $p'=-u'p$ implying $p''=p\left\{(u')^2-u''\right\}$. Therefore, $p$ is convex when $(u')^2\geq u''$. From \eqref{eq_diffequationequivalent} and \eqref{eq_solutiondiffeq} we have
$(u')^2=ce^u-2(1+u)\quad\mbox{and}\quad u''=\half c\,e^u-1$,
and hence we are interested in the $u(0)$ for which
$c=2\geq 2(1+2u)e^{-u}$ for all $u$; i.e. $(1+2u)e^{-u}\leq 1$ for all $u$.
Given that the function $(1+2u)e^{-u}$ is maximum, with a value of $1.31$, at $u=\half$, and $u$ is increasing, since $u'>0$, we need $u(0)>\half$ and $(1+2u(0))e^{-u(0)}=1$, again as an extremal value. Solving this gives a value for $u(0)$ of approximately 1.31.

In the next result we show that $u'$ is bounded away from 0, and this will have important consequences for the properness of $p$.

\begin{lemma}\label{lm_boundawayzero}
It is that $u'$ is bounded away from 0.
\end{lemma}

\begin{proof}
To show this we need to show that $e^u-1-u$ is bounded away from 0 for $u\geq u(0)$. This follows trivially since
$e^u-1-u\geq \half u^2\geq \half u(0)^2$.
\end{proof}
The result of Lemma \ref{lm_boundawayzero} has also the implication that $p$ is a proper density function. To show this, we require Gronwall's inequality \citep{Gronwall1919}. This inequality states that, if $f$ and $g$ are real valued functions on $\Theta=(0,\infty)$, $g$ is differentiable on $\mbox{int}(\Theta)$, and
$g'(t)\leq g(t)\,f(t)$ for all $t\in\Theta$
then
$g(t)\leq g(0)\,\exp\left\{\int_0^t f(s)\,\d s\right\}.$

\begin{lemma}\label{lm_properness}
If $p(0)<\infty$, it is that $p(\theta)\leq p(0)\,e^{-\epsilon\theta},$
and hence $p$ is proper.
\end{lemma}

\begin{proof}
Since $p'=-u'\,p$ and we have $u'\geq \epsilon$ for some $\epsilon>0$, it is that $p'\leq -\epsilon\,p$. From Gronwall's lemma, with $f(t)=-\epsilon$ and $g=p$, we have that
$$p(\theta)\leq p(0)\,\exp\left\{-\int_0^\theta \epsilon\,\d s\right\},$$
and hence the proof is complete.
\end{proof}

To have a graphical image, in Figure \ref{fig1} we plot the prior using the approximation available via a numerical solution to the differential equation for $p$. Note that this is the unnormalised $p$.

\begin{figure}[H]
\begin{center}
\includegraphics[scale=0.4]{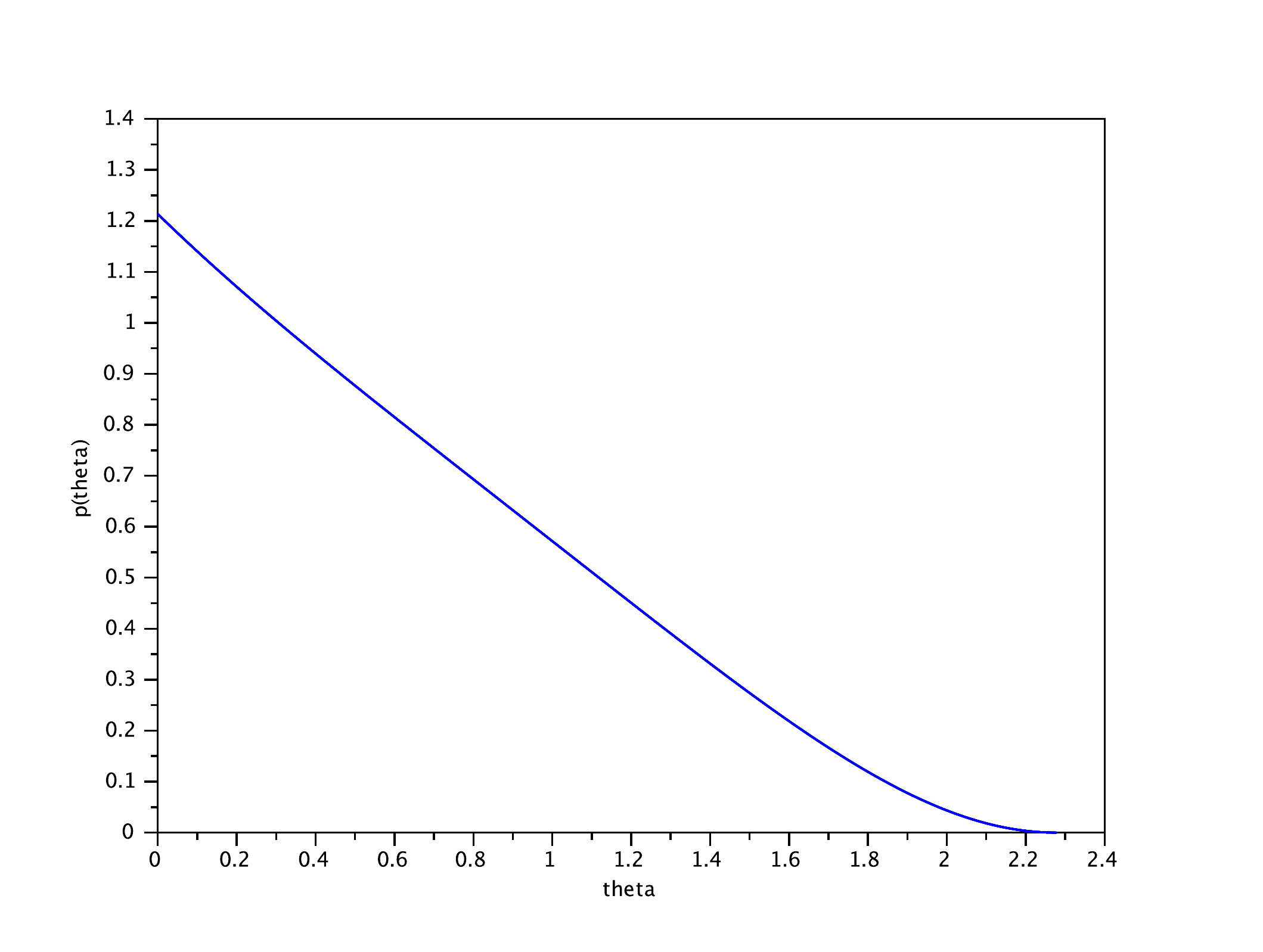}
\caption{Global prior distribution for the case $\Theta=[0,\infty)$, with $c=2$ and $u(0)=1.31$}
\label{fig1}
\end{center}
\end{figure}

\paragraph{Case $\bm{\Theta=(-\infty,+\infty)}$.} A solution here is a symmetric version of the case $\Theta=(0,+\infty)$, which will represent a proper prior.

On the other hand, if we ask that $p$ is smooth at the origin, i.e. $p'(0)=0$, then we need $u'(0)=0$ and hence we must take $(c,u(0))$ to satisfy $c\,e^{u(0)}=2+2u(0)$. If now we take $c=2$, then $u(\theta)$ is a constant, resulting in a flat (improper) prior for $p$.

For a proper prior here one could take $u(0)$ to be small, say $u(0)=0.01$ and then to take 
$c=2\{1+u(0)\}/\exp\{u(0)\}$. We computed numerically the right side; i.e. the $(0,\infty)$ side, for $p$, which is shown in Fig.~\ref{ffour}.

\begin{figure}[H]
\begin{center}
\includegraphics[scale=0.4]{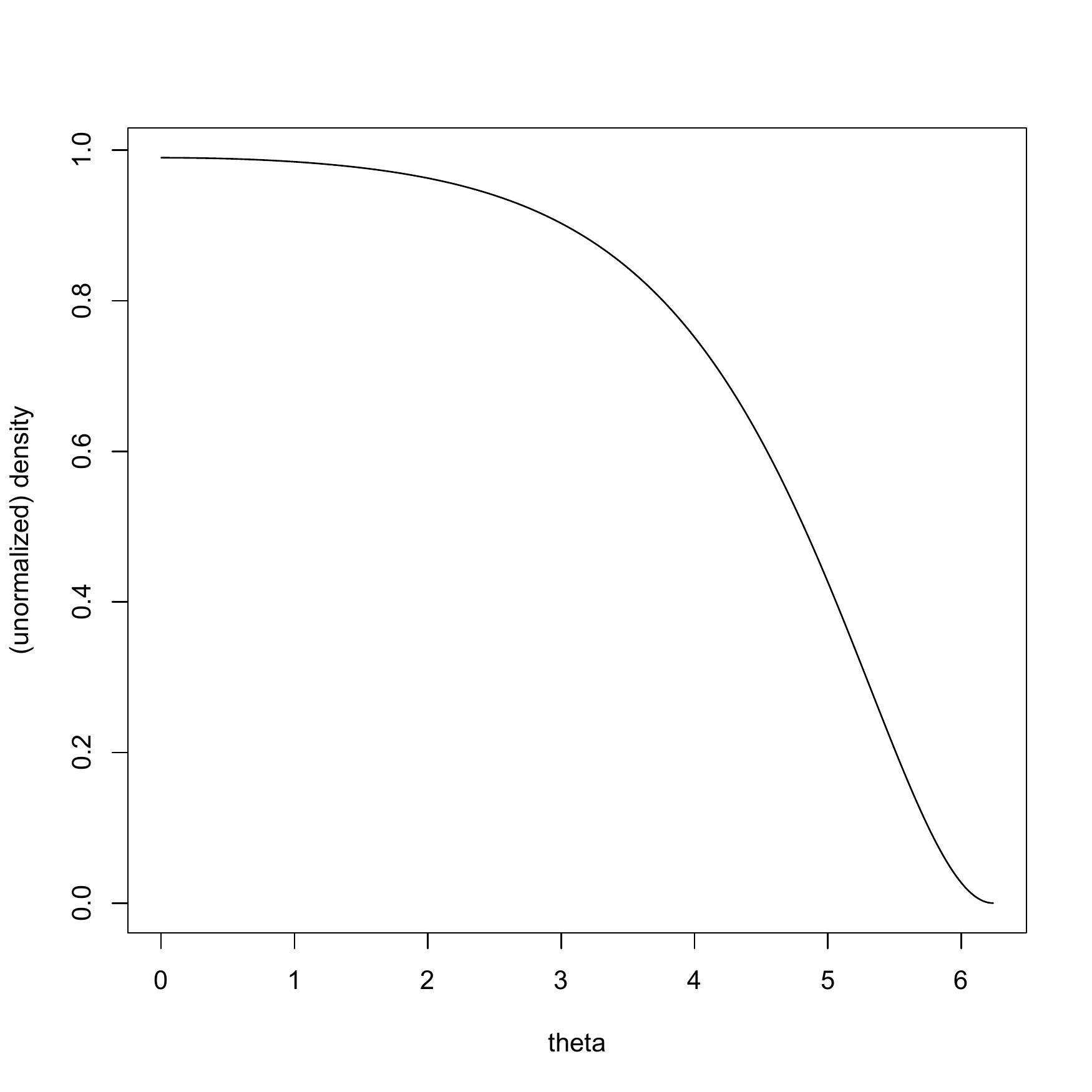}
\caption{Right side of (unnormalized) prior $p$ which is  proper and differentiable at origin}
\label{ffour}
\end{center}
\end{figure}

To make the value of $u(0)$ more diverted, we could equally set a motivated choice for $p(0)=1/\{\sigma\sqrt{2\pi}\}$, corresponding to a normal density with zero mean and variance $\sigma^2$.

\section{Applications}\label{sc_application}
Although we do not have an explicit form for $p(\theta)$, we can use \eqref{eq_solutiondiffeq} to calculate it numerically quite easily. In particular, if we know $p(\theta)$ then we calculate $p(\theta+\delta\theta)$ for small $\delta\theta$, hence setting up the possibility of a posterior estimation process via Metropolis--Hastings sampling. The algorithm employed is detailed in the Appendix.

To be specific, suppose we are currently at $\theta$ and the proposal value is $\theta'$. The acceptance probability is 
\begin{equation}\label{Eq_MH_alpha}
\alpha=\min\left\lbrace 1,\frac{l(\theta')}{l(\theta)}\frac{p(\theta')}{p(\theta)}\frac{q(\theta|\theta')}{q(\theta'|\theta)}\right\rbrace,
\end{equation}
where $l(\theta)$ is the likelihood function, and $q(\theta'|\theta)$ is the proposal density.
The evaluation of $p(\theta')/p(\theta)$ in equation \eqref{Eq_MH_alpha} does not represent any particular challenge. In fact, we have
$p(\theta')/p(\theta)=\exp\{-[u(\theta')-u(\theta)]\},$ where $u$ is the solution of the differential equation
\begin{equation}\label{Eq_diffeq}
u'=\sqrt{ce^{u}-2(1+u)}.
\end{equation}
Equation \eqref{Eq_diffeq} allows us to evaluate $u(\theta')-u(\theta)$ numerically, via 
\begin{equation*}
\begin{split}
u(\theta')&\approx u(\theta)+(\theta'-\theta)u'(\theta)+\frac{(\theta'-\theta)^2}{2}u''(\theta)+\frac{(\theta'-\theta)^3}{6}u'''(\theta),
\end{split}
\end{equation*}
where 
$u'(\theta)=\sqrt{ce^{u(\theta)}-2(1+u(\theta))}$, $u''(\theta)=\half ce^{u(\theta)}-1$, and 
$$u'''(\theta)=\half ce^{u(\theta)}\sqrt{ce^{u(\theta)}-2(1+u(\theta))},$$
and so on. Depending on how far $\theta'$ is from $\theta$ we can either use the direct approximation just given or otherwise get from $\theta$ to $\theta'$ using smaller step sizes.


Before showing the results of a thorough simulation study in Section \ref{sc_simulations}, we consider the following two applications based on single i.i.d. samples:
\begin{enumerate}
\item $X_1,\ldots,X_n\sim\mbox{Poisson}(\theta)$, where we assume our prior on the unknown parameter $\theta>0$, that is a sample space $(0,\infty)$.
\item $X_1,\ldots,X_n\sim\mbox{N}(\mu,1)$, where we assume our prior on the unknown parameter $\mu\in\mathbb{R}$, that is a sample space $(-\infty,\infty)$.
\end{enumerate}
In the first illustration we sample $n=100$ observations from a Poisson distribution with mean value $\theta=2.5$. The Metropolis--Hastings has been run with $c$ and $u(0)$ set as discussed in Section \ref{sc_illustrations}, and the results for 100,000 iterations are shown in Figure \ref{Fig:lambdachain} (posterior sample) and in Figure \ref{Fig:lambdahist} (posterior histogram). Both figures show a good mixing and a posterior distribution that accumulates around the true parameter value. 
\begin{figure}[h]
\centering
\includegraphics[scale=0.4]{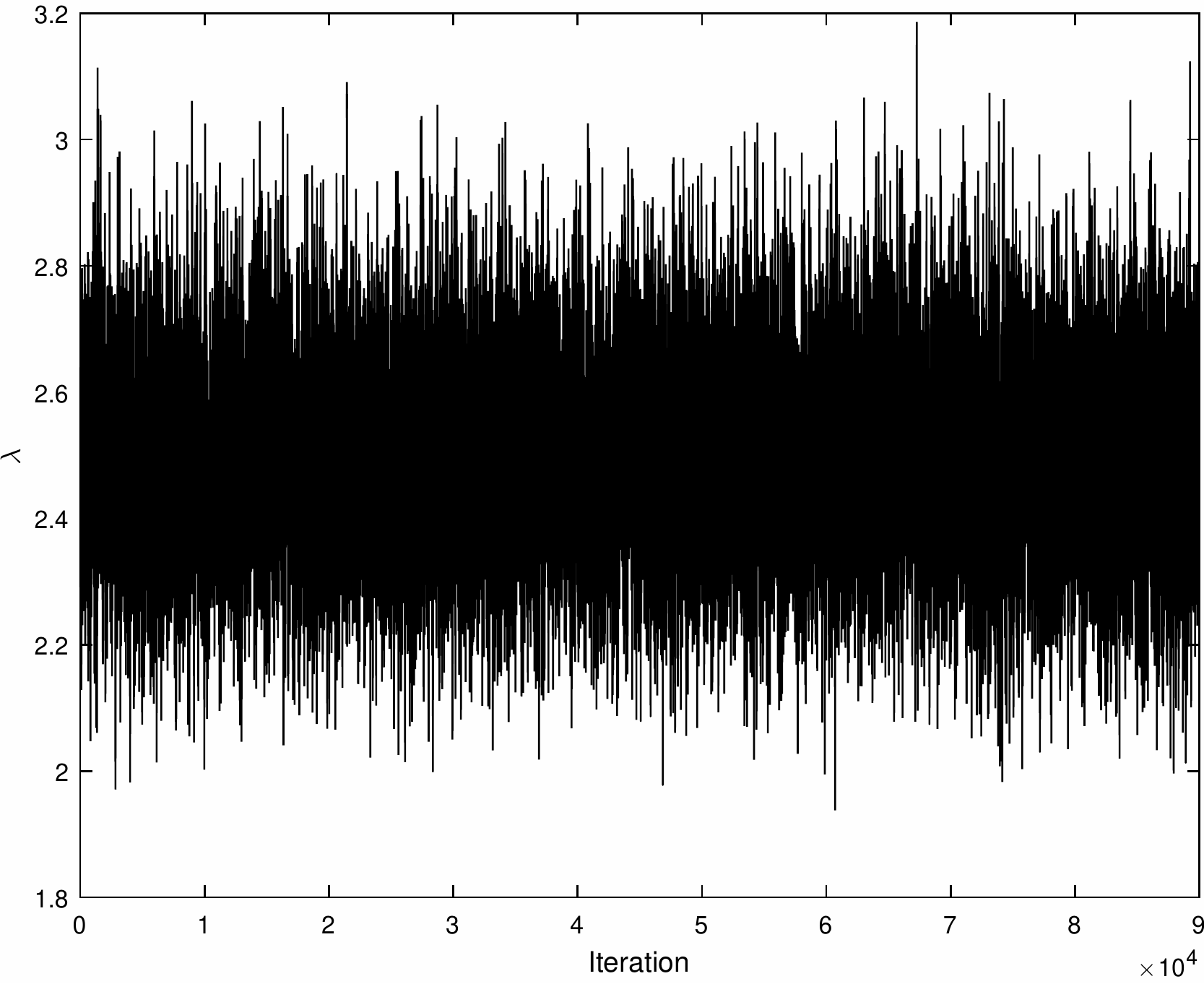}
\caption{Posterior chain for $\theta$ of a Poisson distribution with mean $\theta=2.5$.}
\label{Fig:lambdachain}
\end{figure}
\begin{figure}[h]
\centering
\includegraphics[scale=0.4]{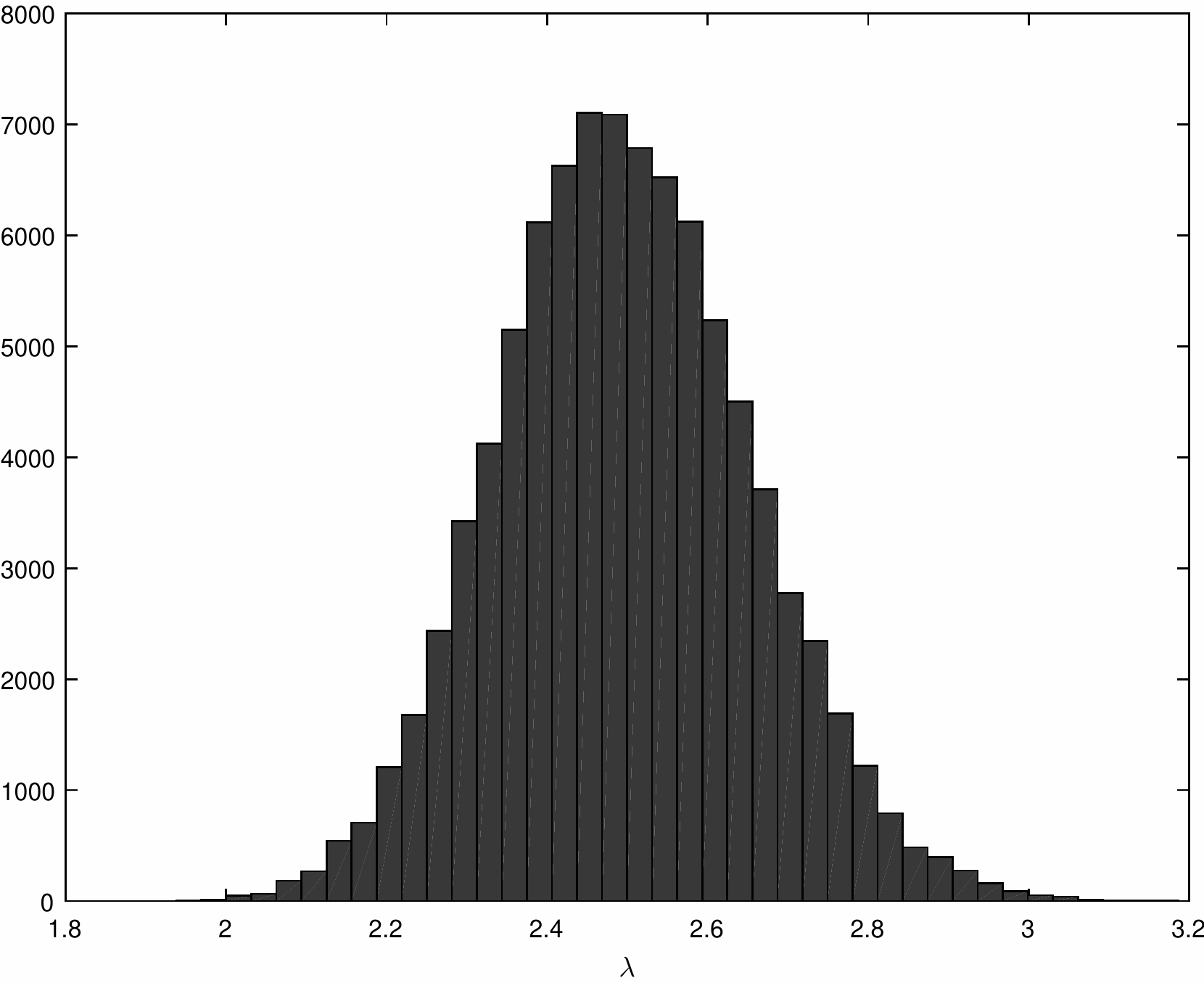}
\caption{Posterior histogram for $\theta$ of a Poisson distribution with mean $\theta=2.5$.}
\label{Fig:lambdahist}
\end{figure}

For the second illustration we have sampled $n=100$ observations from a normal density with mean $\mu=5$ and (known) variance equal to one. Here, for a total of 100,000 iterations, we had the results in Figure \ref{Fig:muchain} (posterior sample) and Figure \ref{Fig:muhist} (posterior histogram). In this case too the results are satisfactory.
\begin{figure}[h]
\centering
\includegraphics[scale=0.4]{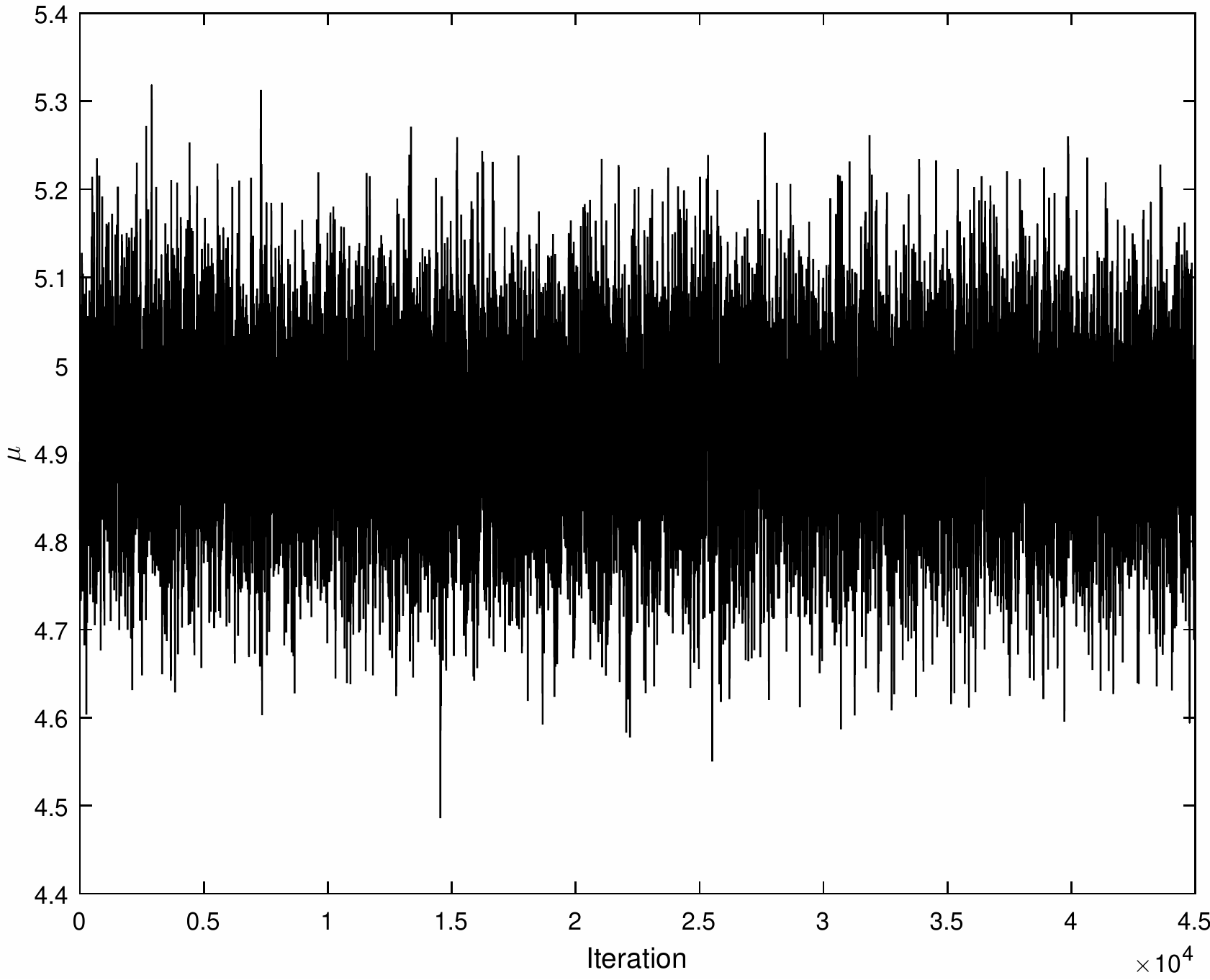}
\caption{Posterior chain for $\mu$ of a normal density with mean $\mu=5$.}
\label{Fig:muchain}
\end{figure}
\begin{figure}[h]
\centering
\includegraphics[scale=0.4]{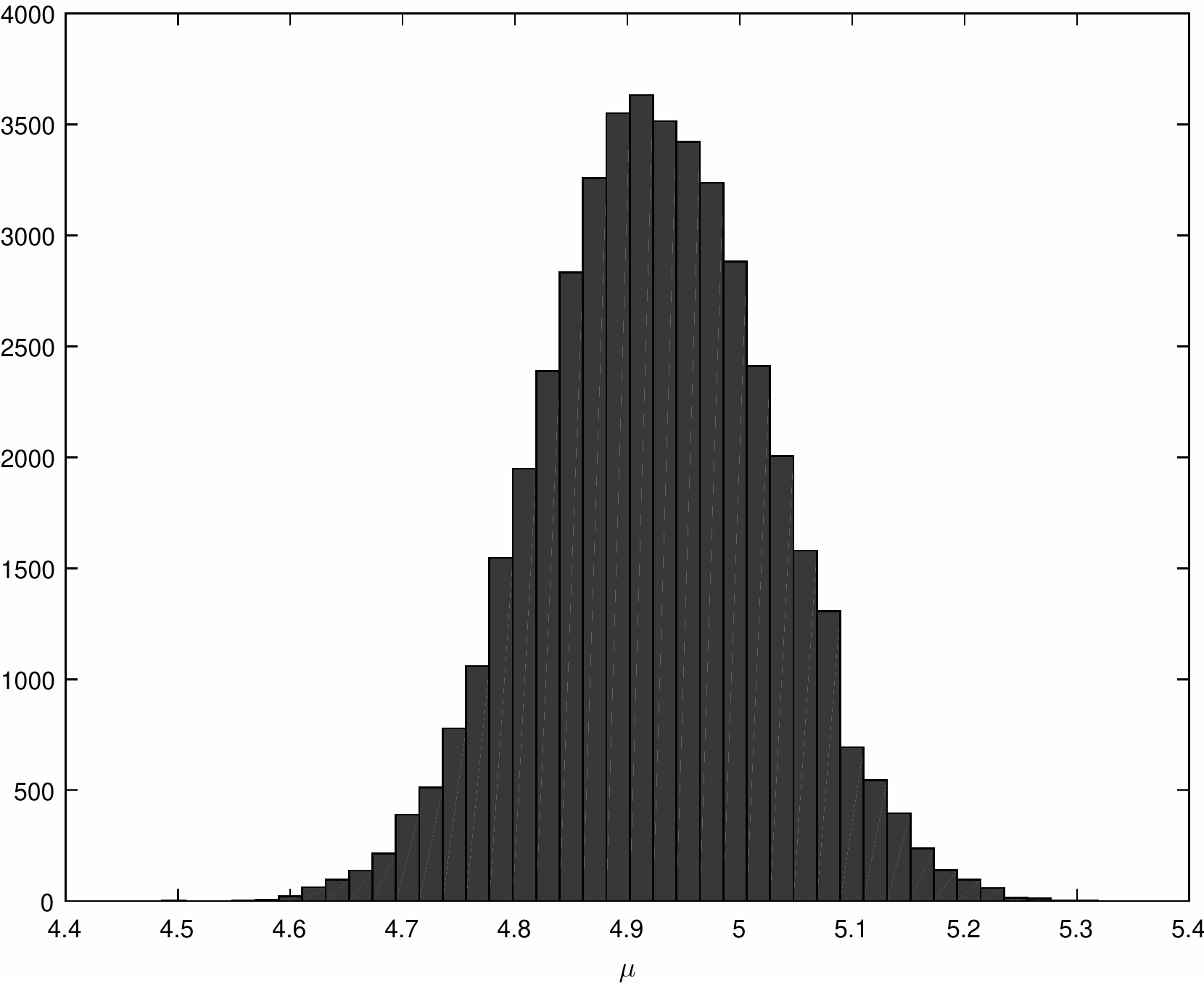}
\caption{Posterior histogram for $\mu$ of a normal density with mean $\mu=5$.}
\label{Fig:muhist}
\end{figure}

\subsection{Mixture models}\label{sc_mixture}
In this section we discuss the application of the proposed method to a scenario where objective priors have been notoriously challenging. If we consider mixture models, it is well known that the use of objective priors \citep{Grazian2015} has to be done carefully, as this type of model is subject to issues related to non-identifiability, unbounded likelihoods, etc. The fact is that improper priors may not be appropriate as we might not observe outcomes from every component of the mixture \citep{Titt1985}. For example, \cite{Grazian2015} show that Jeffreys prior is suitable for mixtures of normal densities only in certain circumstances; that is when the unknown parameters are the weights. If the unknown parameters are the means or the variances, then using Jeffreys prior leads to improper posteriors. In particular, if the unknown parameters are the means only, we need at most two components to have proper posteriors; while, if the unknown parameters are the variance or the mean and the variances, then Jeffreys priors are not suitable for inference. The authors also generalise the result to any type of mixture.

Given that the objective prior we propose is proper, it allows to make inference on the parameters of a mixture density as the yielded posteriors will be proper. As an illustration, we consider a mixture of three normal densities, where the weights and the parameters of the components are unknown. In particular, we sample from the following model
\begin{equation}\label{eq_mixture}
f(y|\omega_1,\omega_2,\omega_3,\mu_1,\mu_2,\mu_3,\sigma_1^2,\sigma_2^2,\sigma_3^2)= \sum_{i=1}^3 \omega_i N(\mu_i,\sigma_i^2),
\end{equation}
with weights $\omega_1=0.25$, $\omega_2=0.35$ and $\omega_3=0.40$, means $\mu_1=-3.5$, $\mu_2=0$ and $\mu_3=2.5$, and variances $\sigma_1^2=0.5$, $\sigma_2^2=0.1$ and $\sigma_3^2=1.2$.
Note that we have chosen mixture components that are reasonably distant, so not to be forced to impose any constraint to overcome identifiability, as the focus of the paper is not in this sense. However, the implementation of constraints in that sense is straightforward. For the parameters we have chosen prior independence, where each prior is the prior on the space $(0,1)$ for the weights, on the space $(-\infty,\infty)$ for the means and on the space $(0,\infty)$ for the variances, in agreement with Section \ref{sc_illustrations}. The prior on $(-\infty,\infty)$ is the symmetrised version from $\Theta=(0,\infty)$.

In the first illustration we sample $n=100$ observations from the above mixture. In Figure \ref{fig:mixtsamp100} we can see an histogram of the sampled data.
\begin{figure}[h]
\centering
\includegraphics[width=.45\linewidth]{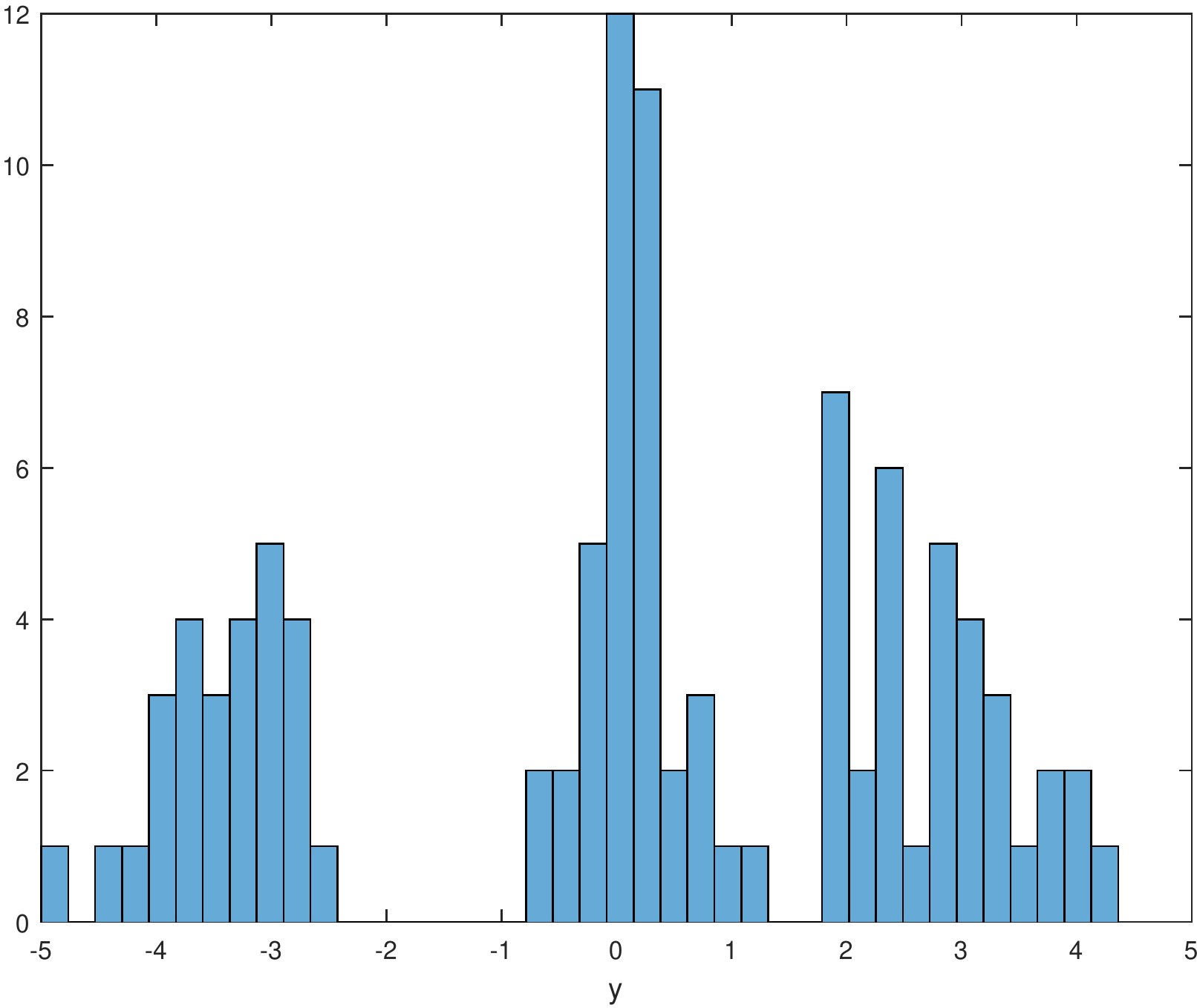}
\caption{Sample of size $n=100$ from mixture \eqref{eq_mixture}.}
\label{fig:mixtsamp100}
\end{figure}

The posterior distributions have been obtained by implementing the Metropolis sampling defined above within a Gibbs sampler. Using the same initial settings of $c$ and $u(0)$ as above, for a total of 10,000 iterations, we obtain the results in Figure \ref{fig:misture100}. The plots show clear convergence of the posterior chains and reasonable inferential results.
\begin{figure}[h]
\begin{subfigure}{.5\textwidth}
  \centering
  \includegraphics[width=.8\linewidth]{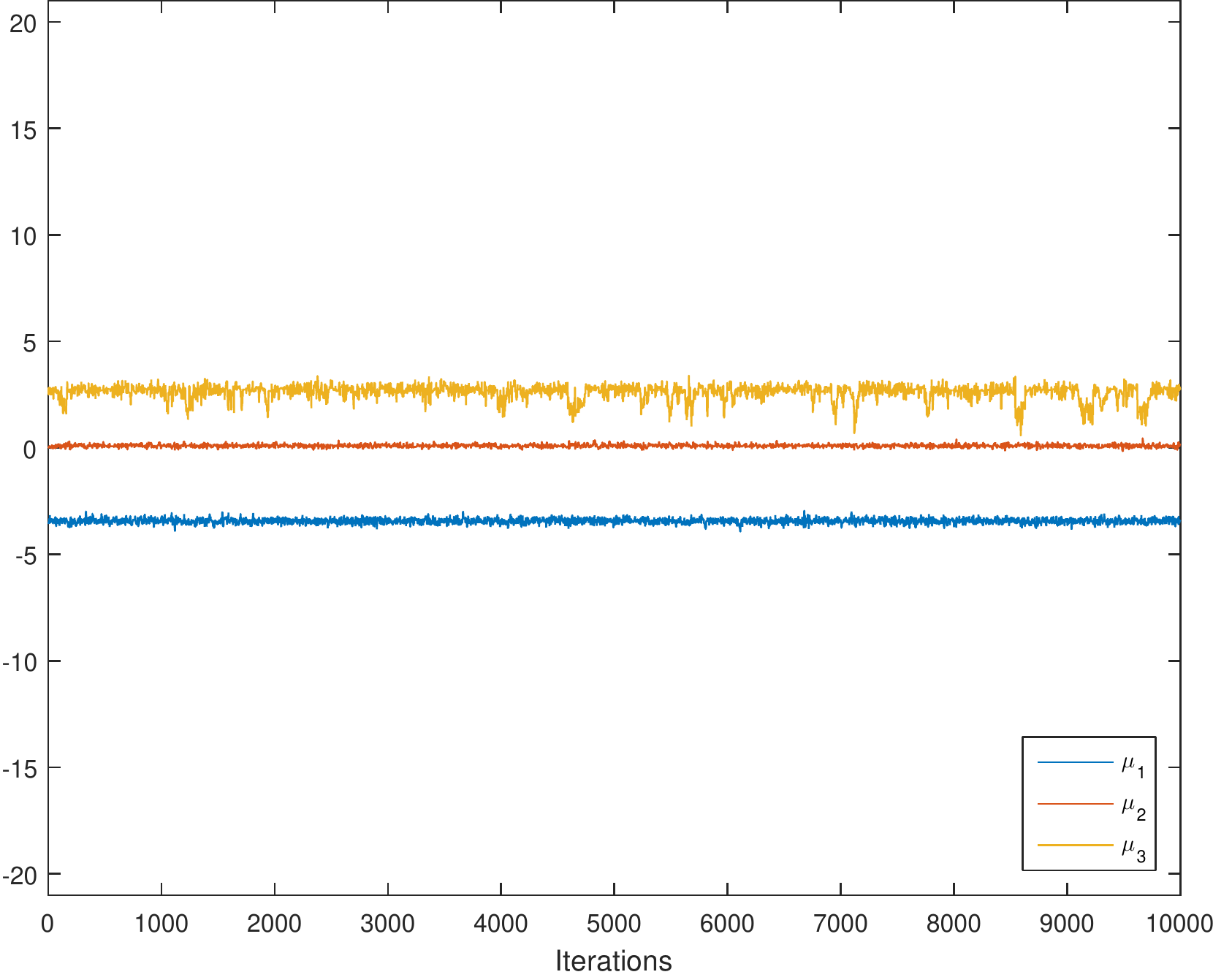}
  \caption{}
  \label{fig:mixa100}
\end{subfigure}%
\begin{subfigure}{.5\textwidth}
  \centering
  \includegraphics[width=.8\linewidth]{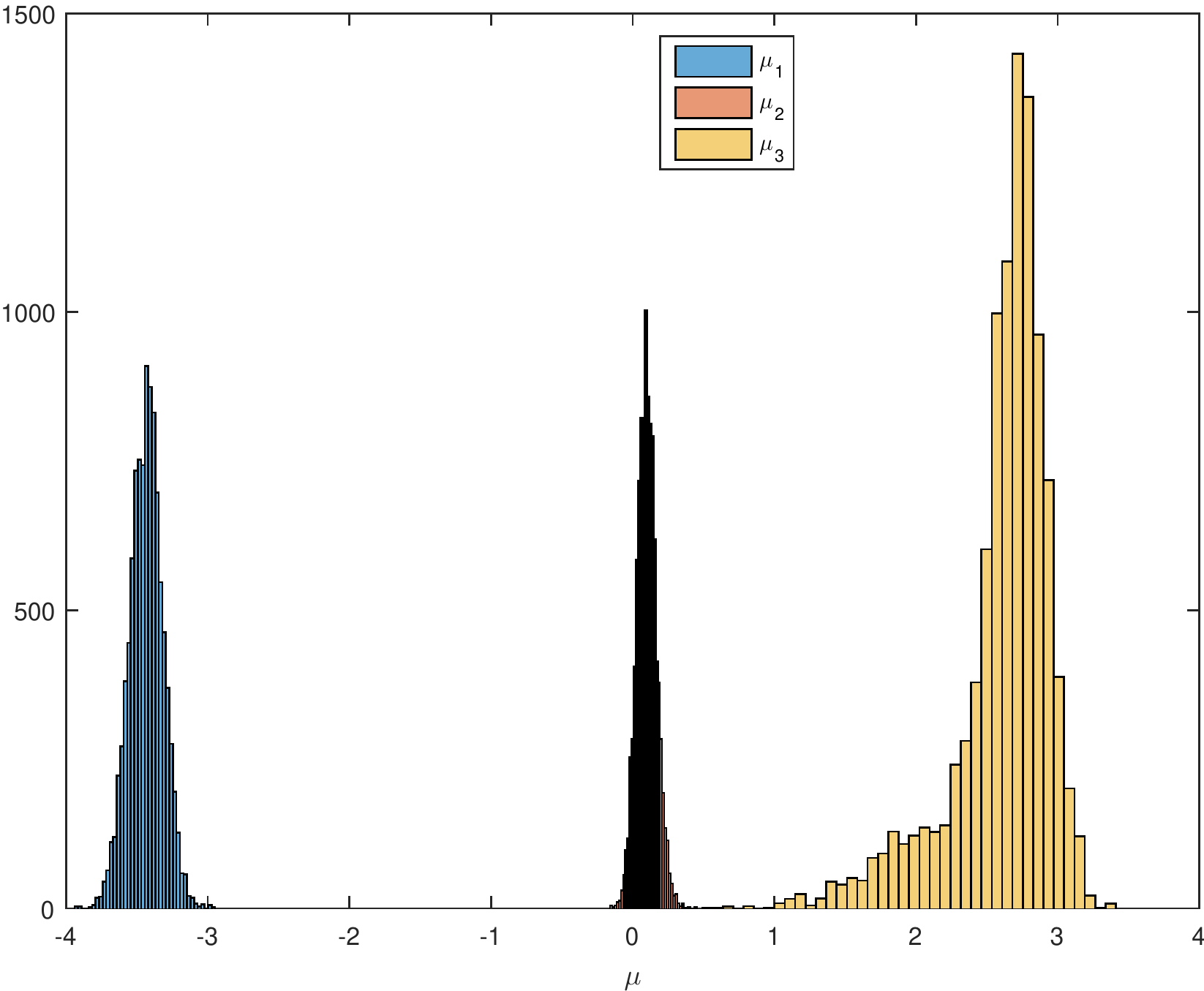}
  \caption{}
  \label{fig:mixb100}
\end{subfigure}
\begin{subfigure}{.5\textwidth}
  \centering
  \includegraphics[width=.8\linewidth]{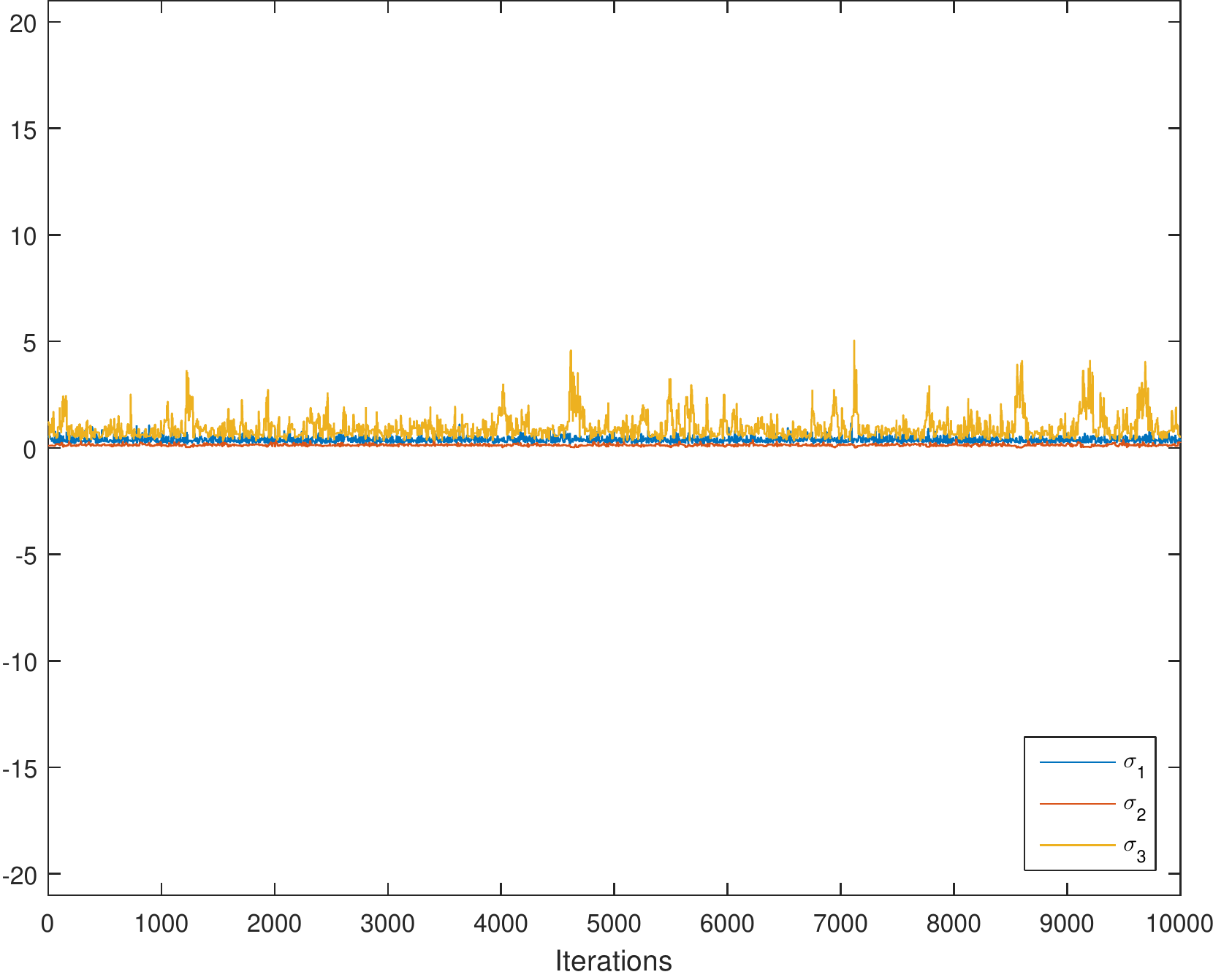}
  \caption{}
  \label{fig:mixc100}
\end{subfigure}%
\begin{subfigure}{.5\textwidth}
  \centering
  \includegraphics[width=.8\linewidth]{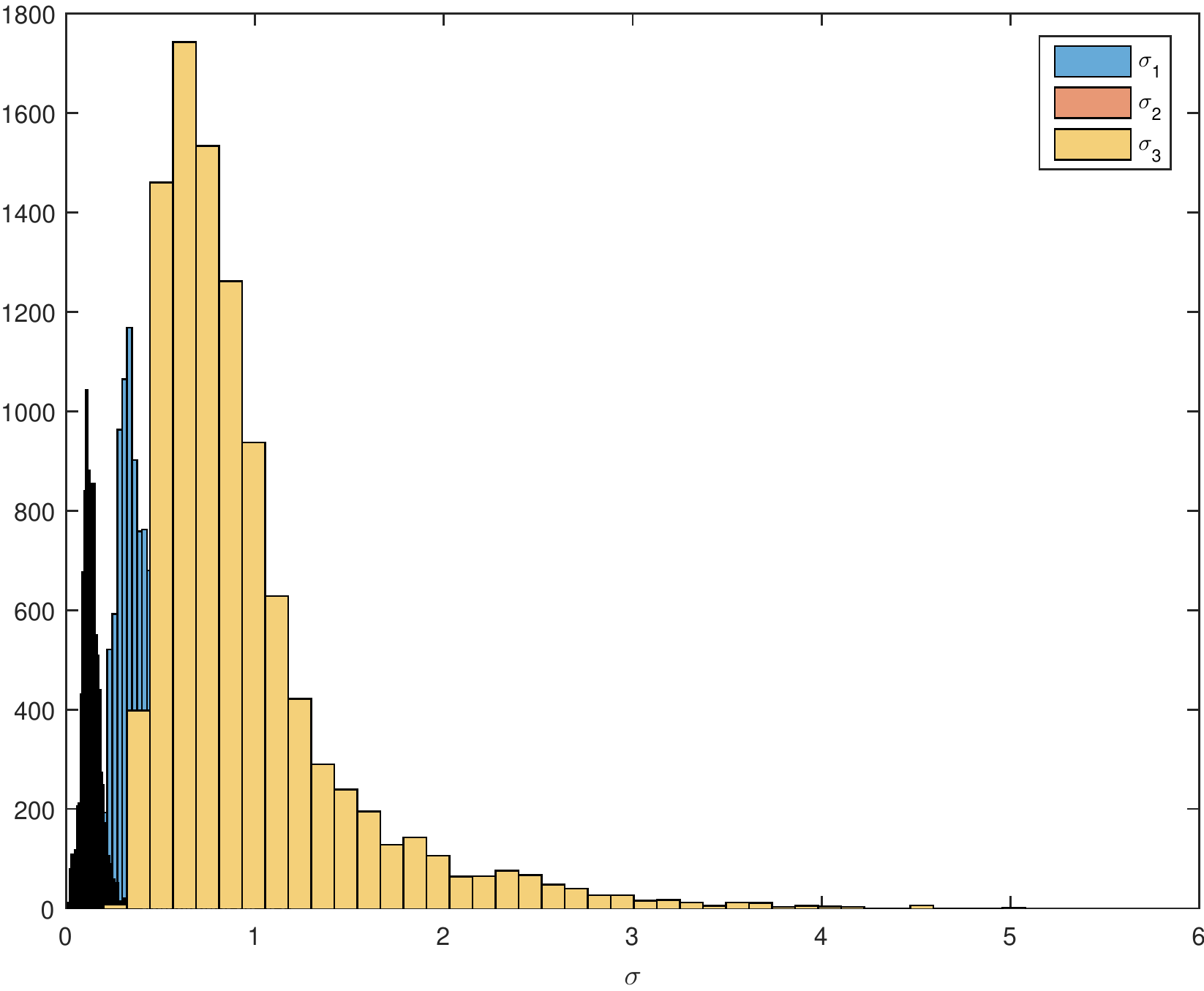}
  \caption{}
  \label{fig:mixd100}
\end{subfigure}
\begin{subfigure}{.5\textwidth}
  \centering
  \includegraphics[width=.8\linewidth]{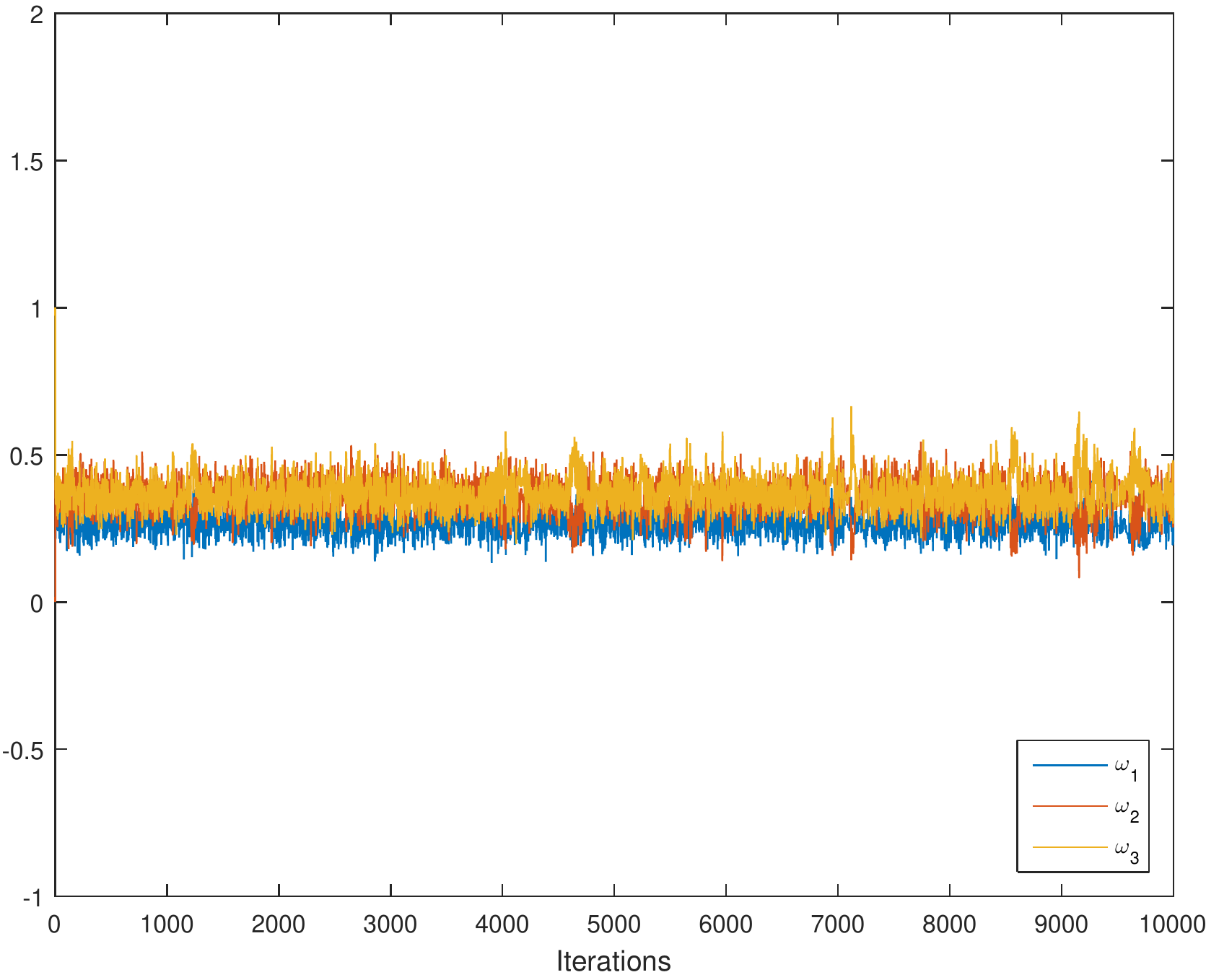}
  \caption{}
  \label{fig:mixe100}
\end{subfigure}%
\begin{subfigure}{.5\textwidth}
  \centering
  \includegraphics[width=.8\linewidth]{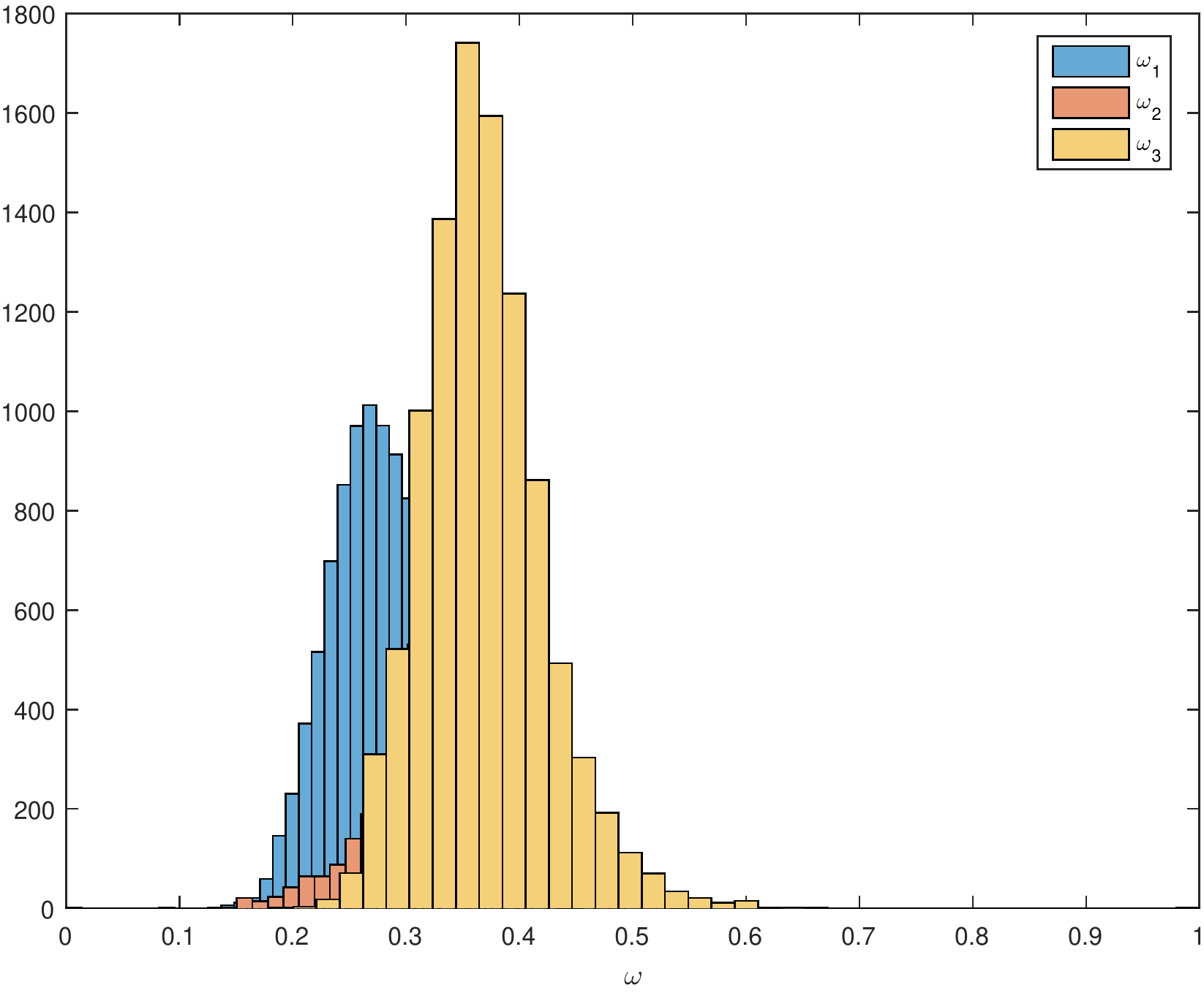}
  \caption{}
  \label{fig:mixf100}
\end{subfigure}
\caption{Posterior traces (left) and histograms (right) for the means, (a) and (b), the variances, (c) and (d), the weights, (e) and (f), for a sample of size $n=100$ from mixture \eqref{eq_mixture}.}
\label{fig:misture100}
\end{figure}

In the second illustration we have increased the sample size $(n=250)$ from the above model \eqref{eq_mixture}, performing the same procedures with the same prior distributions. The histogram of the sample data is in Figure \ref{fig:mixtsamp250}, while the posterior traces and histograms are in Figure \ref{fig:misture250}.
\begin{figure}[H]
\centering
\includegraphics[width=0.45\linewidth]{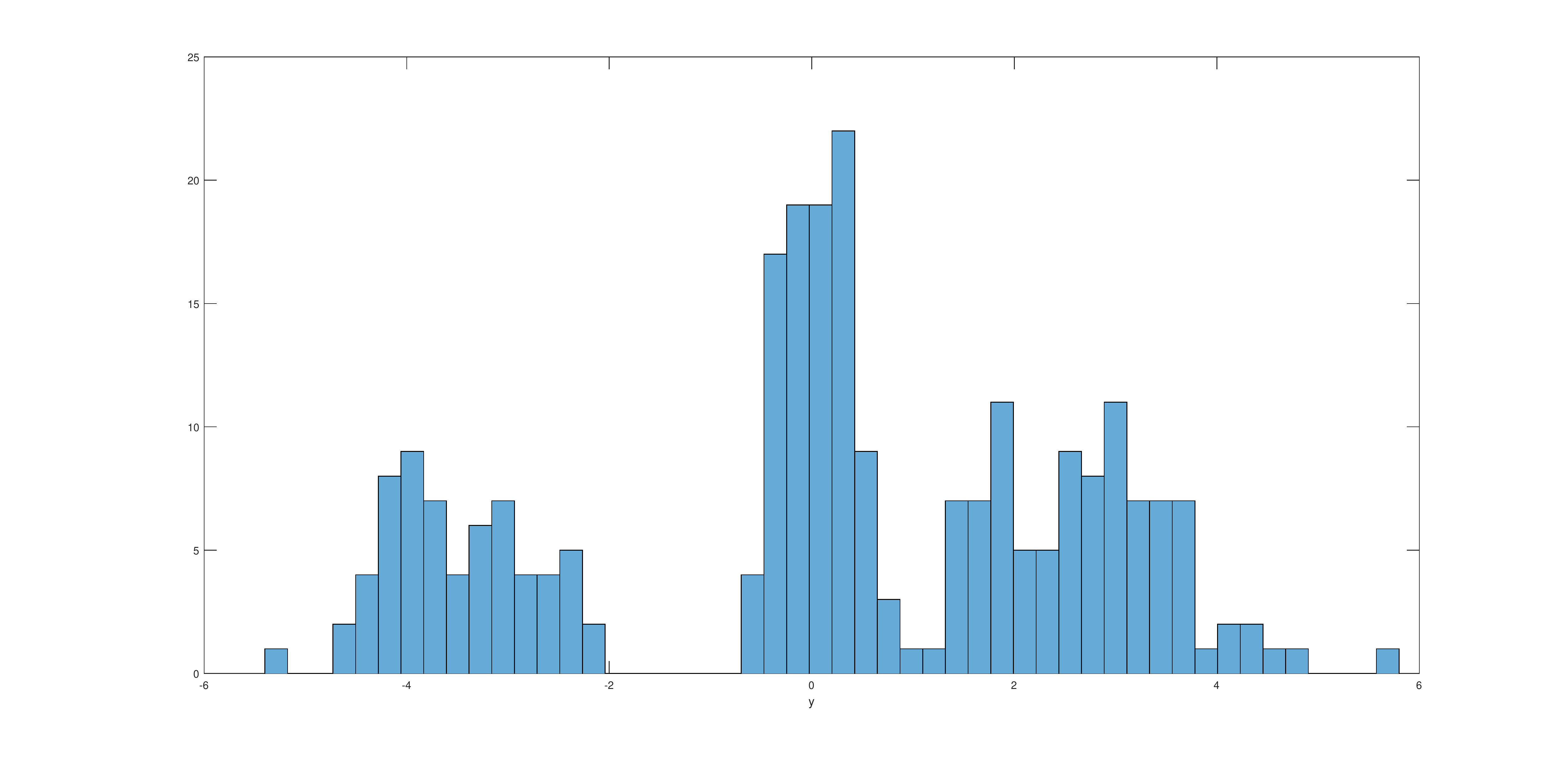}
\caption{Sample of size $n=250$ from mixture \eqref{eq_mixture}.}
\label{fig:mixtsamp250}
\end{figure}

\begin{figure}[h]
\begin{subfigure}{.5\textwidth}
  \centering
  \includegraphics[width=.8\linewidth]{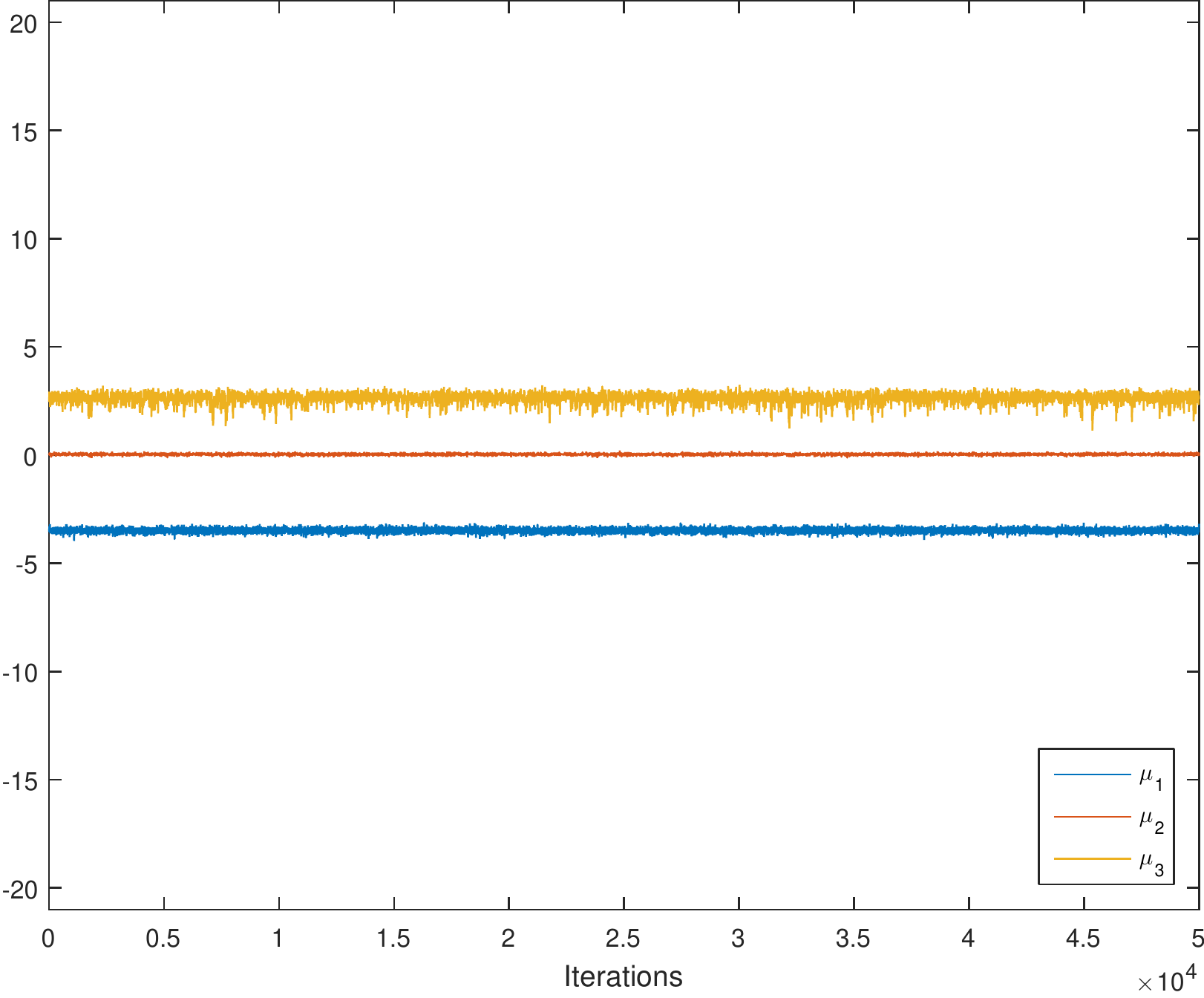}
  \caption{}
  \label{fig:mixa100}
\end{subfigure}%
\begin{subfigure}{.5\textwidth}
  \centering
  \includegraphics[width=.8\linewidth]{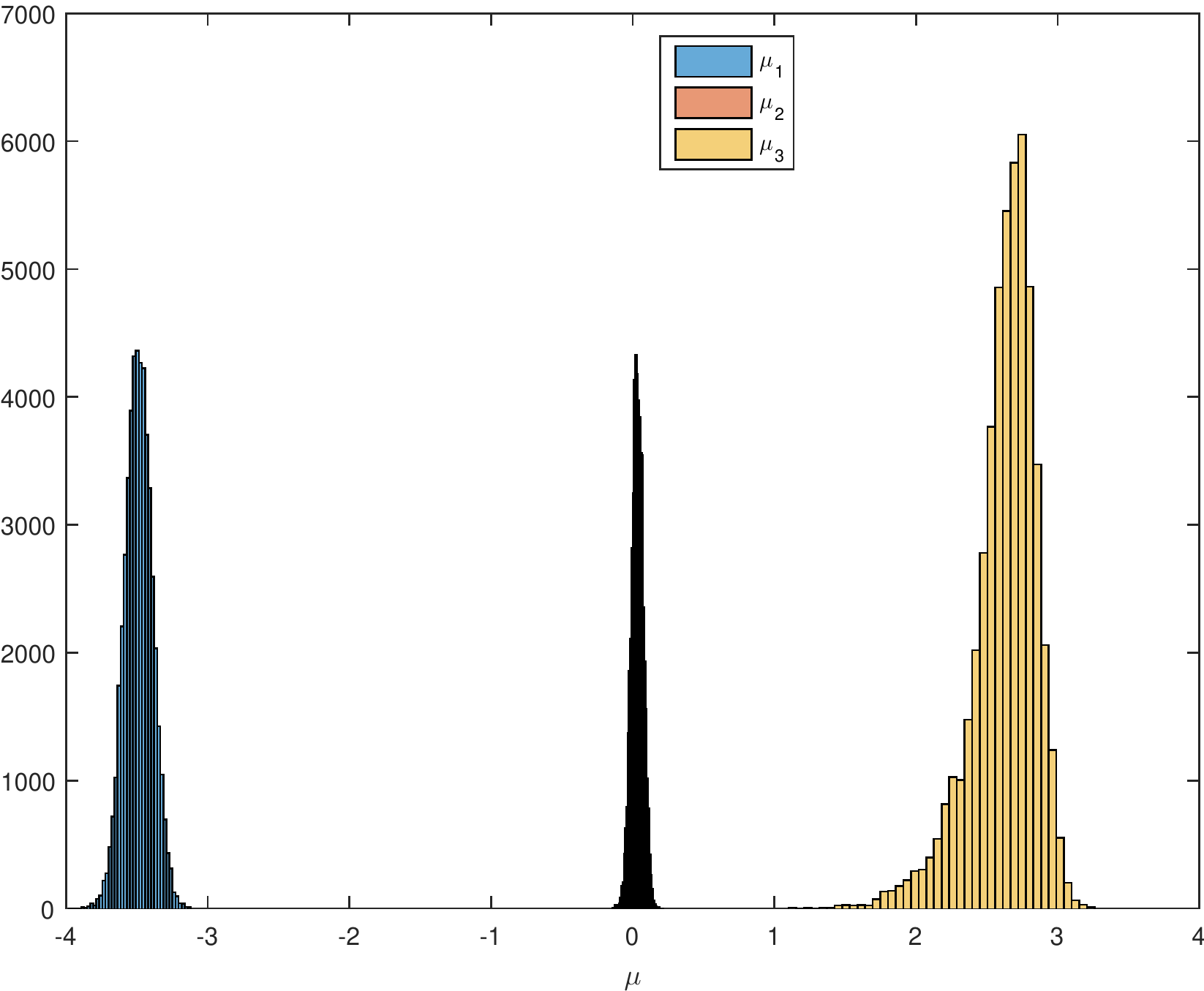}
  \caption{}
  \label{fig:mixb100}
\end{subfigure}
\begin{subfigure}{.5\textwidth}
  \centering
  \includegraphics[width=.8\linewidth]{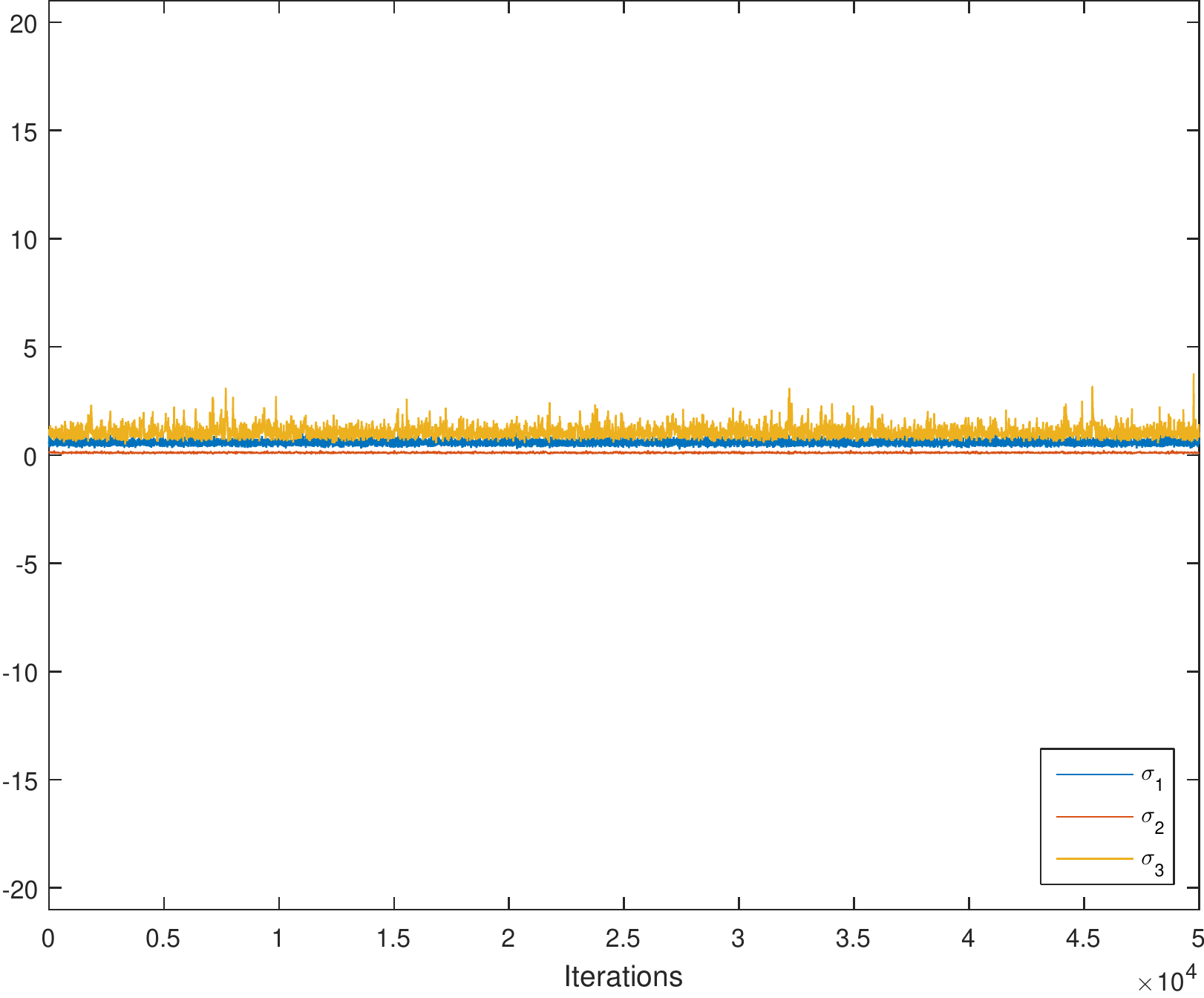}
  \caption{}
  \label{fig:mixc100}
\end{subfigure}%
\begin{subfigure}{.5\textwidth}
  \centering
  \includegraphics[width=.8\linewidth]{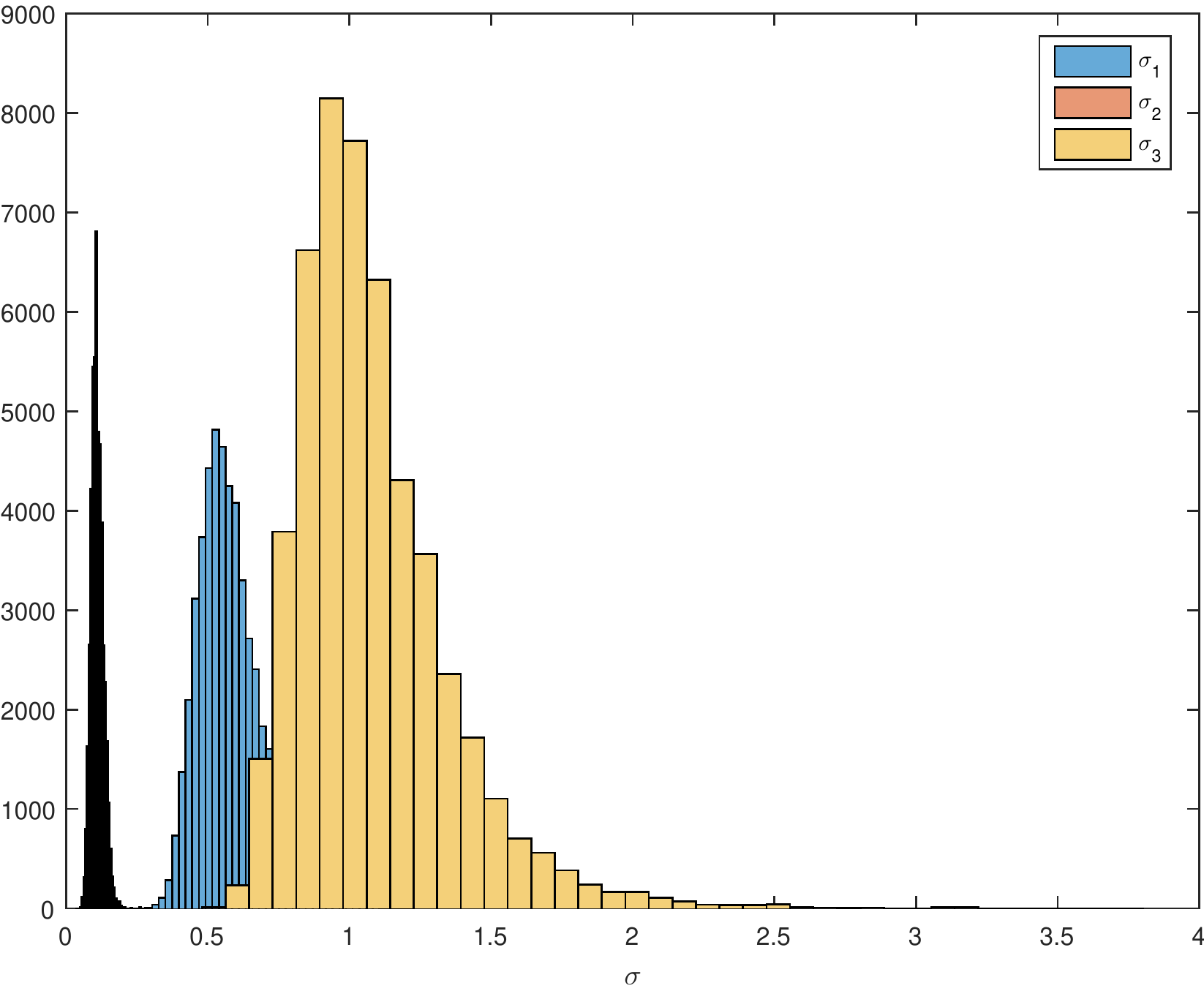}
  \caption{}
  \label{fig:mixd100}
\end{subfigure}
\begin{subfigure}{.5\textwidth}
  \centering
  \includegraphics[width=.8\linewidth]{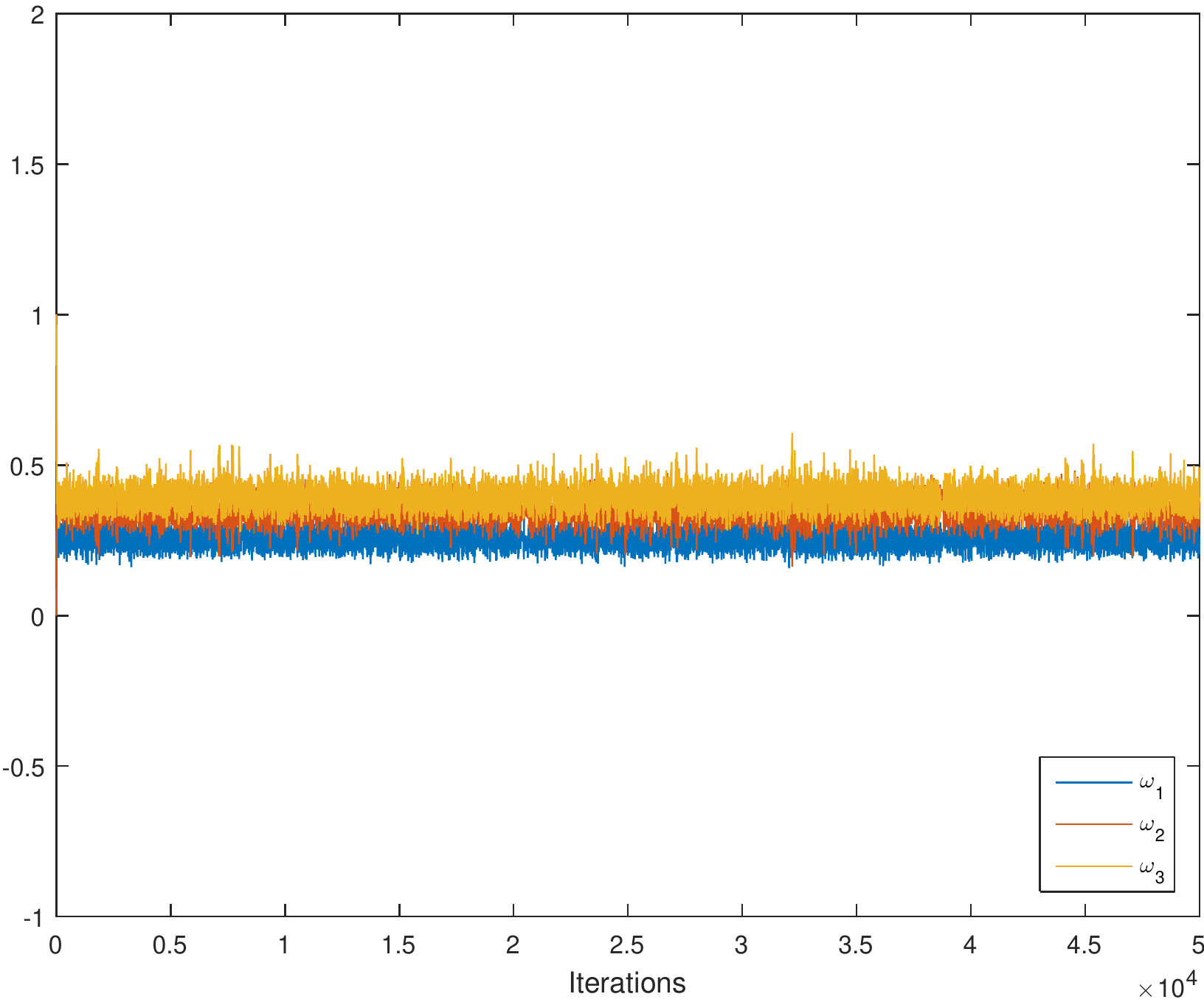}
  \caption{}
  \label{fig:mixe100}
\end{subfigure}%
\begin{subfigure}{.5\textwidth}
  \centering
  \includegraphics[width=.8\linewidth]{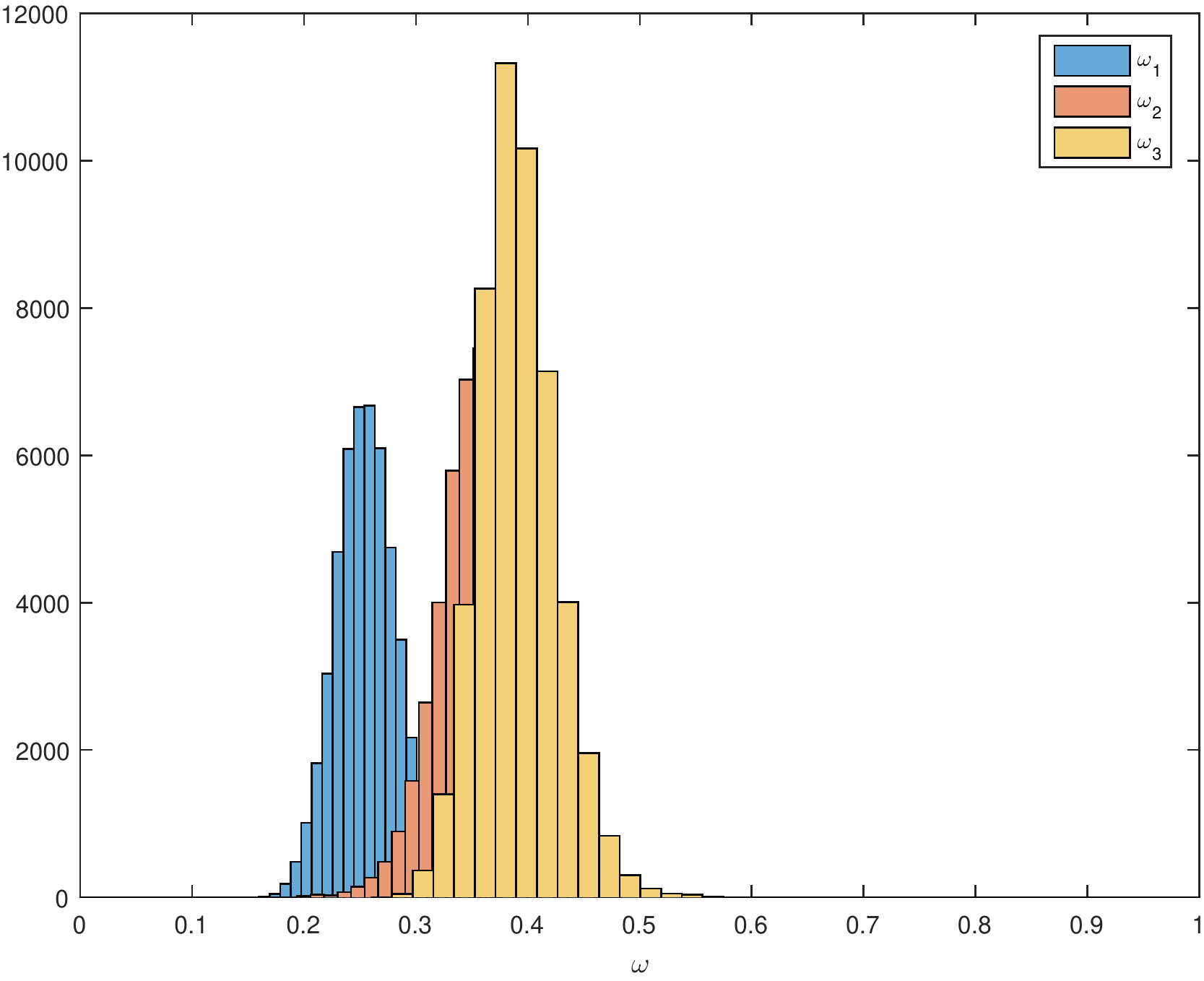}
  \caption{}
  \label{fig:mixf100}
\end{subfigure}
\caption{Posterior traces (left) and histograms (right) for the means, (a) and (b), the variances, (c) and (d), the weights, (e) and (f), for a sample of size $n=250$ from mixture \eqref{eq_mixture}.}
\label{fig:misture250}
\end{figure}

We see in this case as well good performances of the prior with good inferential results.

\subsection{Model comparison}\label{sc_modcomp}
In this section we employ the objective prior in model selection. To illustrate a more comprehensive use, we decide to compare a Poisson distribution and a geometric distribution for some values of the parameters and for two sample sizes: $n=30$ and $n=100$. The problem is then formalised by considering the following two models
$$M_1 = \left\{f_1(x|\theta)=\theta^x e^{-\theta}/x!\, ,p_1(\theta)\right\} \quad \mbox{and} \quad M_2 = \left\{f_2(x|\phi)=\phi(1-\phi)^x\, ,p_2(\phi)\right\},$$
where $p_1(\theta)$ is the prior for the case $(0,\infty)$, as defined in Section \ref{sc_illustrations}, and $p_2(\phi)$ is the uniform prior on the interval $(0,1)$. By setting $P(M_1)=P(M_2)$, we perform the model comparison by means of the Bayes Factor:
$$B_{12} = \frac{m_1(x)}{m_2(x)} = \frac{\int f_1(x|\theta)p_1(\theta)\, d\theta}{\int f_2(x|\phi)p_2(\phi)\, d\phi}.$$
The simulation exercise is performed by first drawing 100 samples from $M_1$ and counting the proportion of times the $BF_{12}>1$; then, we draw 100 samples from $M_2$ and count the proportion of times $BF_{12}<1$. ideally, in both cases, we would like to have a proportion close to $100\%$. Table \ref{tab:modcomp_30} and Table \ref{tab:modcomp_100} show the simulation results for, respectively, $n=30$ and $n=100$. Under each true model we report the minimum and the maximum Bayes Factor $BF_{12}$ and the number of exceptions, indicating the number of times the ``wrong'' model is selected. While for $n=100$ there are no exceptions, we note one exceptions in three case for $n=30$. Given the amount of information about the true model in the data is limited, the exceptions do not raise any particular concern.

\begin{table}[H]
\centering
\begin{tabular}{|c|c|c|c|c|c|c|c|}
\hline 
\multicolumn{2}{|c|}{True model} & \multicolumn{2}{c|}{$M_1$} & \multicolumn{2}{c|}{$M_2$} & \multicolumn{2}{c|}{Exceptions} \\ 
\hline 
$\theta$ & $\phi$ & $\min$ & $\max$ & $\min$ & $\max$ & $M_1$ & $M_2$ \\ 
\hline 
5 & 0.5 & $611.77$ & $1.07 \times 10^{11}$ & $2.97 \times 10^{-45}$ & $7.08 \times 10^{-21}$ & 0 & 0 \\ 
\hline 
2 & 0.5 & $0.69$ & $4.02 \times 10^{05}$ & $1.45 \times 10^{-11}$ & $0.45$ & 1 & 0 \\ 
\hline 
2 & 0.2 & $3.00$ & $4.13 \times 10^{06}$ & $2.87 \times 10^{-77}$ & $58.39$ & 0 & 1 \\ 
\hline 
2 & 0.8 & $0.22$ & $8.15 \times 10^{05}$ & $1.66 \times 10^{-23}$ & $2.28 \times 10^{-08}$ & 1 & 0 \\ 
\hline 
5 & 0.8 & $1.28 \times 10^{03}$ & $6.58 \times 10^{10}$ & $2.92 \times 10^{-60}$ & $3.69 \times 10^{-45}$ & 0 & 0 \\ 
\hline 
\end{tabular}
\caption{Model comparison for $n=30$. Minimum and maximum Bayes Factor under true model $M_1$ (Poisson) and $M_2$ (geometric) for 100 draws.}\label{tab:modcomp_30}
\end{table}

\begin{table}[H]
\centering
\begin{tabular}{|c|c|c|c|c|c|c|c|}
\hline 
\multicolumn{2}{|c|}{True model} & \multicolumn{2}{c|}{$M_1$} & \multicolumn{2}{c|}{$M_2$} & \multicolumn{2}{c|}{Exceptions} \\ 
\hline 
$\theta$ & $\phi$ & $\min$ & $\max$ & $\min$ & $\max$ & $M_1$ & $M_2$ \\ 
\hline 
5 & 0.5 & $2.47 \times 10^{13}$ & $1.68 \times 10^{30}$ & $6.62 \times 10^{-132}$ & $2.15 \times 10^{-79}$ & 0 & 0 \\ 
\hline 
2 & 0.5 & $5.13 \times 10^{03}$ & $2.90 \times 10^{14}$ & $4.10 \times 10^{-31}$ & $2.96 \times 10^{-09}$ & 0 & 0 \\ 
\hline 
2 & 0.2 & $5.87 \times 10^{03}$ & $1.76 \times 10^{16}$ & $9.43 \times 10^{-127}$ & $2.19 \times 10^{-28}$ & 0 & 0 \\ 
\hline 
2 & 0.8 & $2.80 \times 10^{04}$ & $2.40 \times 10^{15}$ & $8.97 \times 10^{-67}$ & $1.25 \times 10^{-43}$ & 0 & 0 \\ 
\hline 
5 & 0.8 & $7.57 \times 10^{15}$ & $1.39 \times 10^{29}$ & $4.96 \times 10^{-189}$ & $8.95 \times 10^{-147}$ & 0 & 0 \\ 
\hline 
\end{tabular}
\caption{Model comparison for $n=100$. Minimum and maximum Bayes Factor under true model $M_1$ (Poisson) and $M_2$ (geometric) for 100 draws.}\label{tab:modcomp_100}
\end{table}

\subsubsection{Nested models} When models under comparison are nested, there are particular considerations which are needed to be taken into account; see, for example, \cite{CFR2013}. The point is that a diffuse type prior for the larger model will end up lacking focus so that the mass assigned to the smaller model is too much. However,  our argument is that if two nested models are under comparison, it is essential, at least from a coherent point of view, to center the larger prior on the fixed part of the smaller one. Let us elaborate.

Suppose $f(y|\theta)$ for $\theta\in\Theta_1$ is the larger model and the smaller one is given by $\theta\in\Theta_0$ where $\Theta_0\subset \Theta_1$. Typically $\Theta_1$ will be a higher dimension to $\Theta_0$ and to get the latter from the former one fixes a particular value in the higher dimension. To make this concrete, let us consider Example 2.1 from \cite{CFR2013}, where $M_0: f(y|\theta_0)$ is binomial$(n,\theta_0)$, with $\theta_0=1/4$ fixed, and $M_1:f(y|\theta)$ is binomial$(n,\theta)$, for which a prior for $\theta$,  $p(\theta)$, is required. Given the nature of the comparison it is our argument that $p(\theta)$ must be centered on $\theta_0$. 

We can adapt quite easily the prior obtained in Section 4, the $\Theta=(0,1)$, to be centered on $1/4$ rather than $\half$. Without repeating the mathematics, we can take $u(1/4)=w$ and $c=2$ and then
$$
\begin{array}{ll}
u'=\sqrt{2}\sqrt{e^u-1-u} & \theta>1/4 \\
u'=-\sqrt{2}\sqrt{e^u-1-u} & \theta<1/4.
\end{array}
$$
For the illustration of the prior $p(\theta)$, obtained numerically from $u$, in Fig~\ref{fint} we took $w=1.5$.

\begin{figure}[H]
\begin{center}
\includegraphics[scale=0.4]{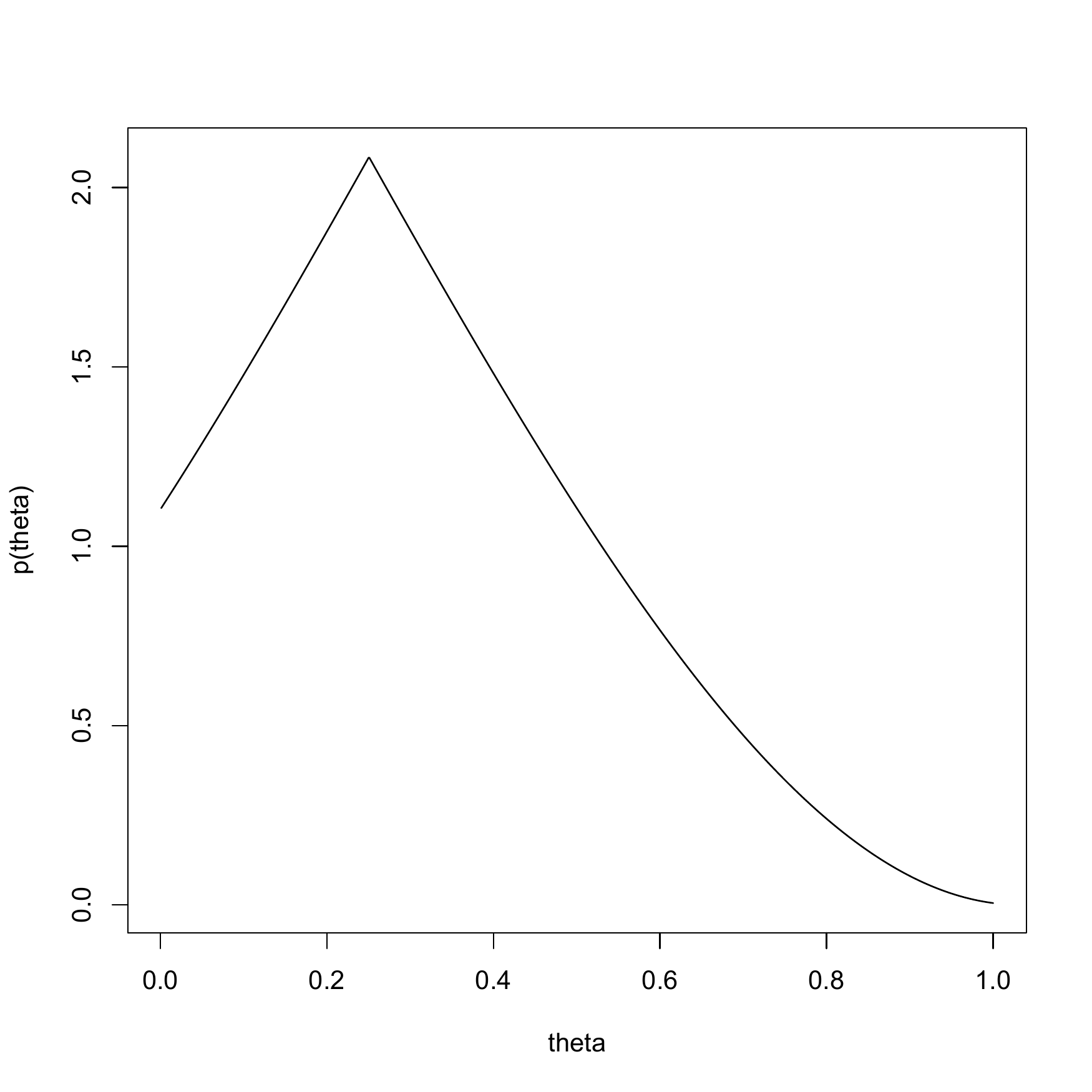}
\caption{Numerical solution for normalized $p$ with $w=1.5$}
\label{fint}
\end{center}
\end{figure}

The proposed prior, centered at $\theta_0=0.25$, is compared with the intrinsic prior in \cite{CFR2013}, that is
$$p^I(\theta|b,t) = \sum_{x=0}^t\mbox{Beta}(\theta|b+x,b+t-x)\mbox{Bin}(x|t,\theta_0),$$
where $b=1$ and $t=8$. The intrinsic prior defined above is centered at $w\theta_0+(1-w)\half$, where $w=t/(2b+t)$, and has behaviour similar to the one in Figure \ref{fint}, giving the Intrinsic Bayes Factor in favour of $M_1$
$$B^I_{10} = \sum_{x=0}^8\frac{B(1+x+y,1+t-x+n-y)}{B(1+x,1+t-x)\theta_0^y(1-\theta_0)^{n-y}},$$
where  $n=12$. A relatively small sample size allows to better capture differences in the performance of the two priors. Figure \ref{fig:example2_1} shows the posterior probability for $M_1$, i.e. $P(M_1|y)=(1+1/B^I_{10})^{-1}$.
\begin{figure}[H]
\begin{center}
\includegraphics[scale=0.4]{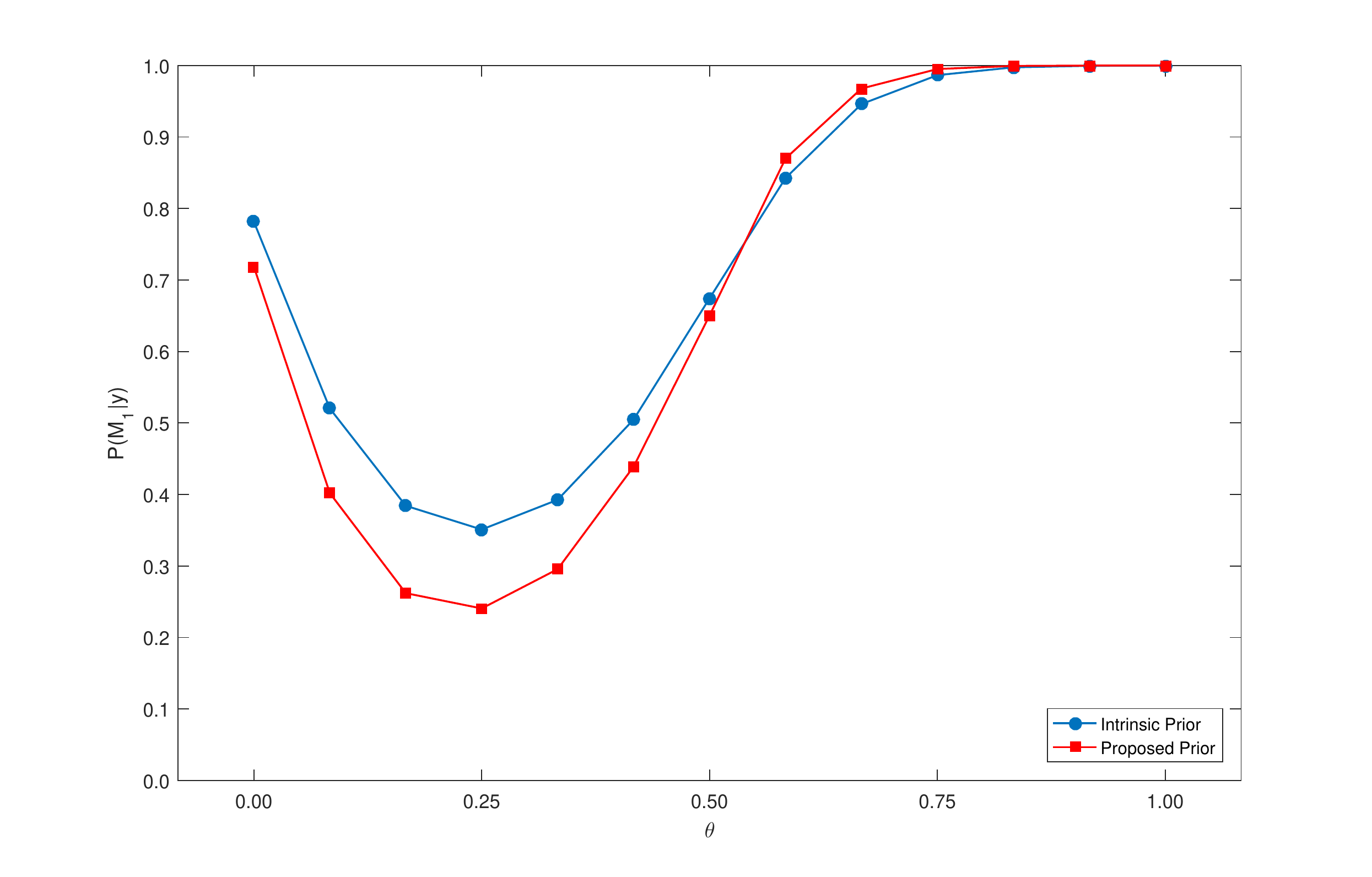}
\caption{Small sample evidence for the Binomial example. The graph shows the posterior probability for model $f(y|\theta)=\mbox{Bin}(n=12,\theta)$ using the proposed prior (squares) and the intrinsic prior (circles).}
\label{fig:example2_1}
\end{center}
\end{figure}
The priors yield model probabilities that are similar; in fact in both cases the lowest point is at $\theta=\theta_0$ and, the more $\theta$ moves away from $\theta_0$ the higher the posterior probability for $M_1$.

\subsection{Poisson regression}
In this section we show how the proposed prior can be used in estimating the parameters of a Poisson regression model. We will apply the prior to both simulated and real data.

Let us consider the following Poisson regression model
$$y_i\sim\mbox{Poisson}(\theta_i), \qquad i=1,\ldots,n$$
with
$$\theta_i = \exp\left(\sum_{j=1}^k x_{ij}\beta_j\right).$$
The covariates are the $x_{ij}$'s, with the corresponding vector of parameters $\boldsymbol\beta=\left(\beta_1,\ldots,\beta_k\right)$.
In this section we show how we make inference on the vector $\boldsymbol\beta$ by using the prior in the form $p(\boldsymbol\beta)=p(\beta_1)\times\cdots\times p(\beta_k)$, where $p(\beta_j)$ is the prior defined on the interval $(-\infty,\infty)$. Again, this prior is the symmetrised version of the prior on $(0,\infty)$.

We start by setting $k=5$, with $\boldsymbol\beta=(-0.8,-0.5,0,0.5,0.8)$, and we simulate the covariate values from a multivariate normal with vector of means $\boldsymbol\beta$ and variance matrix a diagonal matrix of dimension $k$. The vector of observations of the response variable $y$ is obtained from a Poisson distribution with mean $\exp(\mathbf{x}\boldsymbol\beta)$. Applying the algorithm in the Appendix, with 50,000 simulations and a burn-in period of 25,000, the posterior 95\% credible intervals for the $\beta$'s are represented in the caterpillar plot in Figure \ref{fig:poissregk5}.
\begin{figure}[h]
\centering
\includegraphics[scale=0.6]{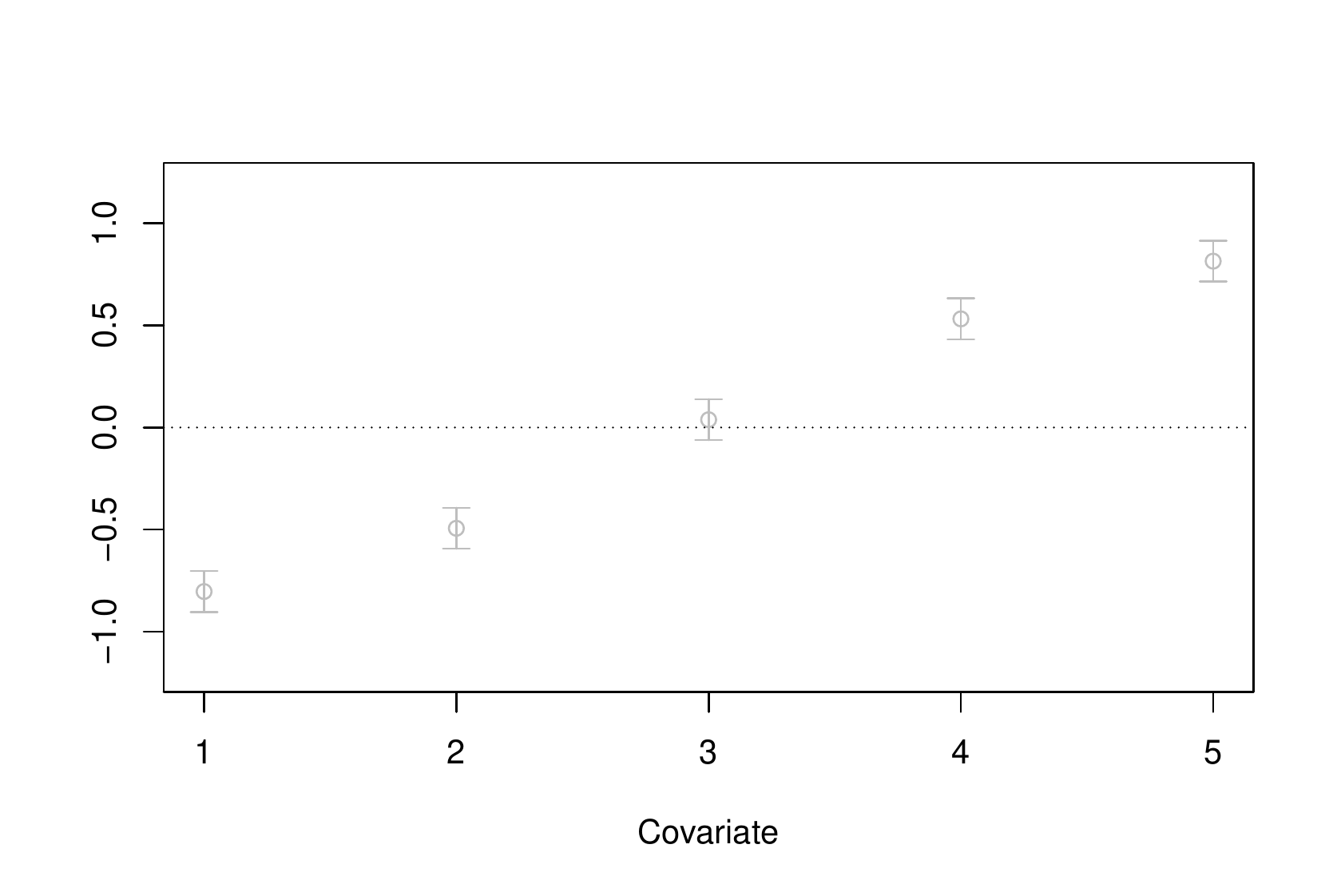}
\caption{Caterpillar plot representing the 95\% credible intervals of the $\beta$'s of a Poisson regression with $k=5$ covariates.}
\label{fig:poissregk5}
\end{figure}

We can see that the intervals contain the true parameter values, showing the effectiveness of the inferential approach using the proposed prior. To verify the robustness of the procedure, we have included random noise by adding, first 5 and then another 10, covariates with true value equal to zero. The results are in Figure \ref{fig:poissregk10} and Figure \ref{fig:poissregk20}. Again, we see that the posterior intervals comprise the true values with no evident impact on their size when noise is added.
\begin{figure}[h]
\centering
\includegraphics[scale=0.6]{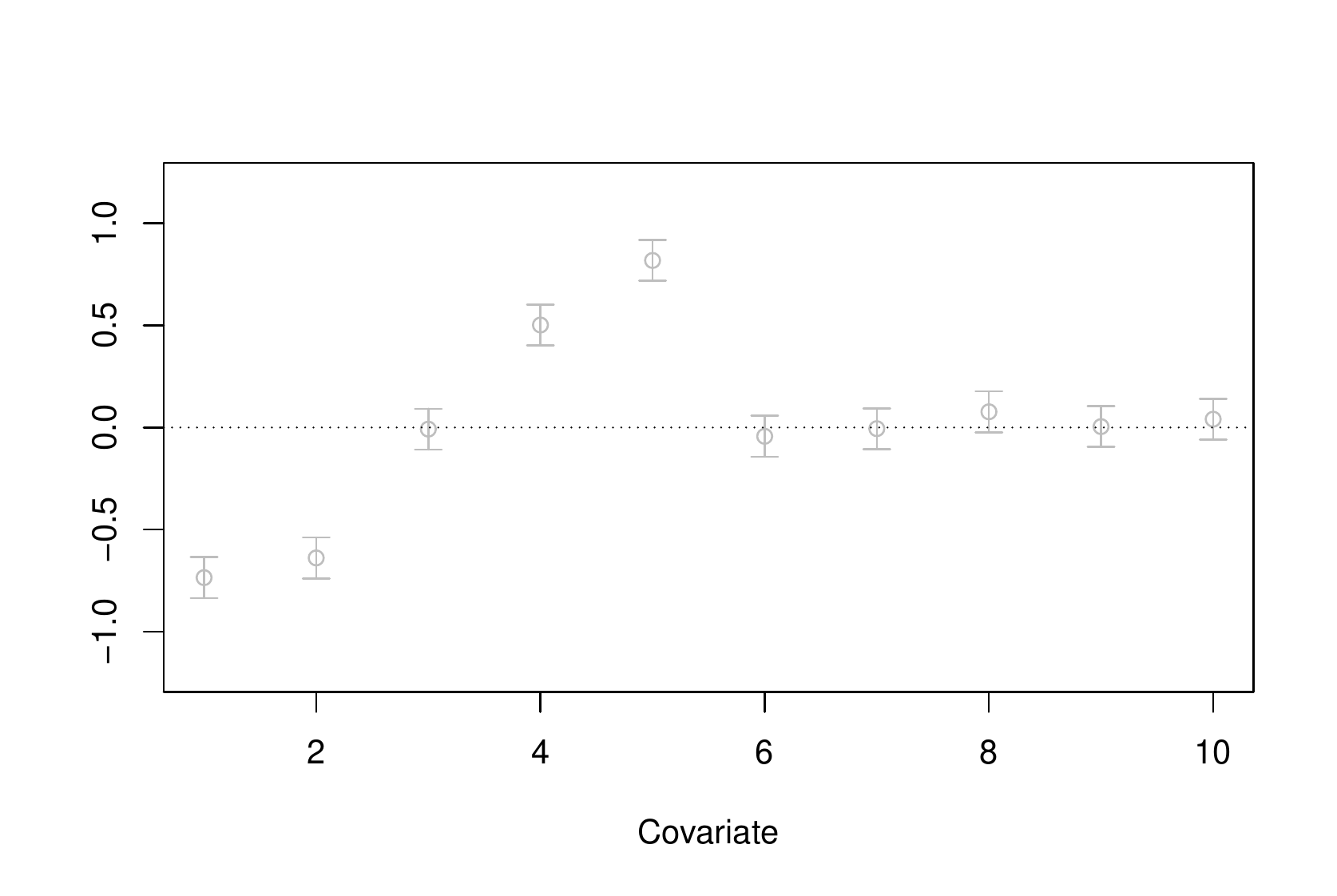}
\caption{Caterpillar plot representing the 95\% credible intervals of the $\beta$'s of a Poisson regression with $k=10$ covariates.}\label{fig:poissregk10}
\end{figure}
\begin{figure}[h]
\centering
\includegraphics[scale=0.6]{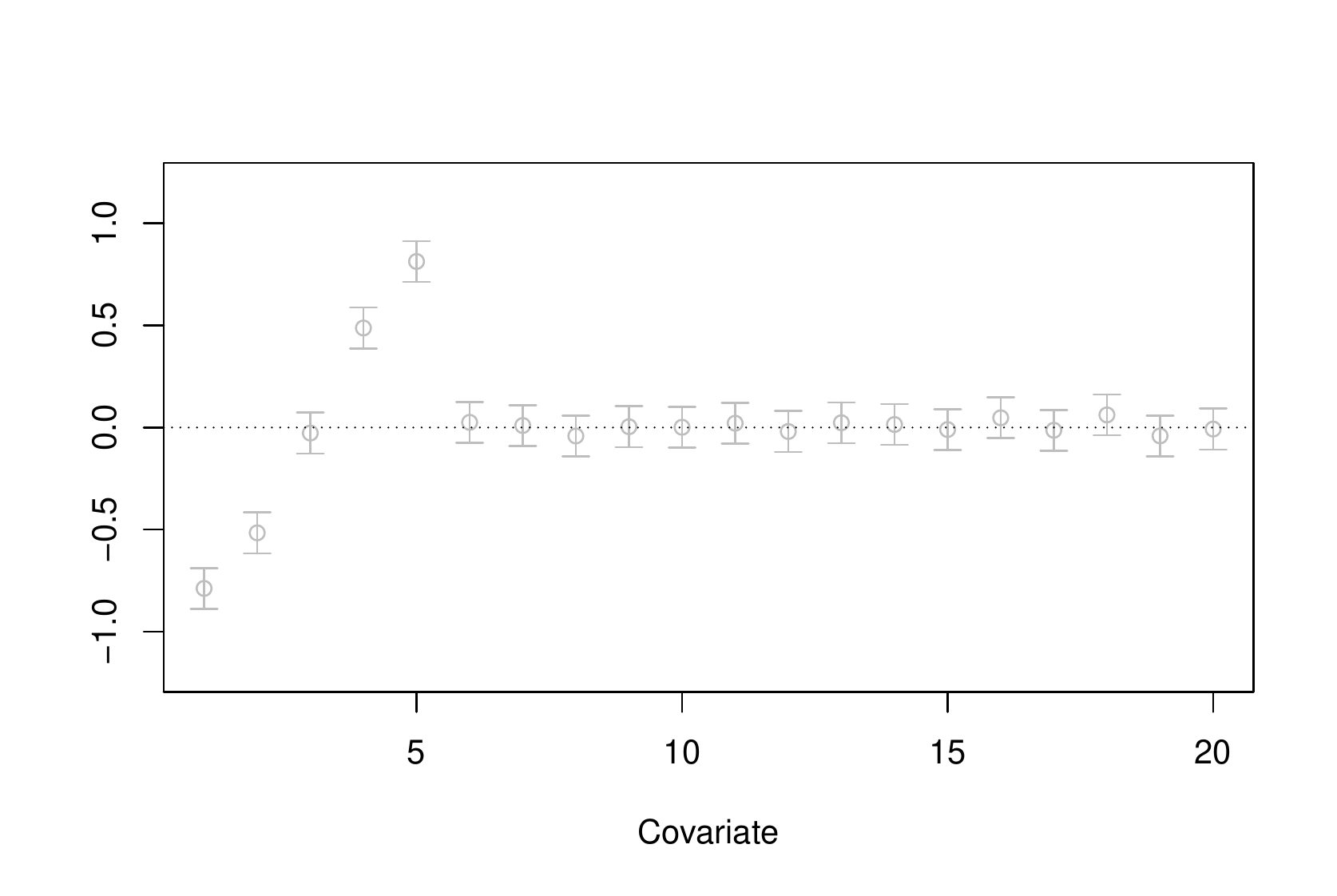}
\caption{Caterpillar plot representing the 95\% credible intervals of the $\beta$'s of a Poisson regression with $k=20$ covariates.}
\label{fig:poissregk20}
\end{figure}

To show the application of the prior on a real data set, we consider the $n=78$ observations on bilateral sanctions behaviour, for selected years, in the period 1939--1983 \citep{Martin1992}. In particular, we replicate the analysis in \cite{GoodLu2013}, where the number of sanctions (`num') is regressed on the political stability and economic health of the country targeted by the sanction ('target') and the level of international cooperation (`coop'). For the covariates, the coding is described in Tabl \ref{tab:regcoding}.

\begin{table}[H]
\centering
\begin{tabular}{|c|c|c|}
\hline 
 & \multicolumn{2}{c|}{Meaning} \\ 
\hline 
Value & `target' & `coop' \\ 
\hline 
1 & Severe stress & No cooperation \\ 
\hline 
2 & Mildy stable & Minor cooperation \\ 
\hline 
3 & Relatively stable & Modest cooperation \\ 
\hline 
4 & NA & Significant cooperation \\ 
\hline 
\end{tabular}
\caption{Coding of the covariates representing the political stability and economic wealth of the nation targeted by the sanction (values from 1 to 3) and of the level of international cooperation (values from 1 to 4).}\label{tab:regcoding}
\end{table}

Therefore, the Poisson regression model will be
$$y_i \sim \mbox{Poisson}\left(x_{i1}\beta_1 + x_{i2}\beta_2 + x_{i3}\beta_3\right), \qquad i=1,\ldots,78,$$
where $\beta_1$ is the intercept, $x_{i1}=1$, for $i=1,\ldots,78$, $x_{i2}$ is the political and wealth level of the target and $x_{i3}$ the level of international cooperation; $\beta_2$ and $\beta_3$ the corresponding coefficients. Assuming independent prior knowledge, we have $p(\beta_1,\beta_2,\beta_3)=p(\beta_1)p(\beta_2)p(\beta_3)$ and, as for the simulation above, we assign a prior defined on $(-\infty,\infty)$ to each coefficient.

In the first step of the analysis we have run 50,000 iterations for 4 chains per parameter, where the starting points have been randomly chosen. Figure \ref{fig:regmultichain} shows the first 10,000 iterations for each parameter where a fast convergence is clear. To obtain the inferential results, we have therefore focused on a chain and, considering a burn-in period of 25,000 iterations, we have obtained the trace plots and histograms in Figure \ref{fig:realreg_posttracehist} and the corresponding posterior statistics in Table \ref{tab:realdatacomp}. In Table \ref{tab:realdatacomp} we have included the posterior summaries obtained by using default priors, that is normal densities centered at 0 with large variance \citep{GoodLu2013}. We can see that the results are very similar, in terms of point estimate, standard deviation and size of the credible interval.
\begin{figure}[h!]
\begin{subfigure}{0.5\textwidth}
  \centering
  \includegraphics[scale=0.35]{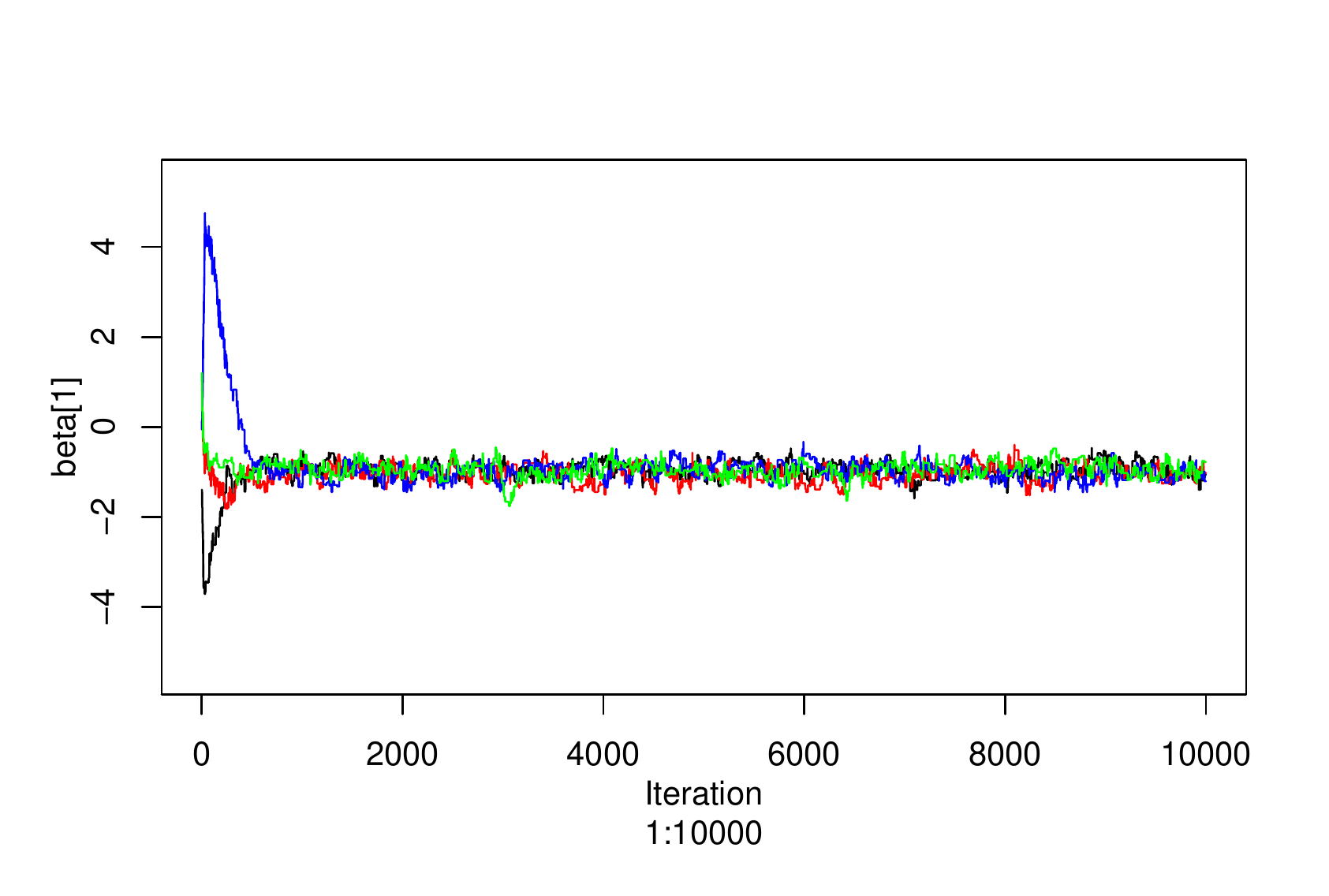}
  \caption{}
  \label{fig:multi1}
\end{subfigure}%
\begin{subfigure}{0.5\textwidth}
  \centering
  \includegraphics[scale=0.35]{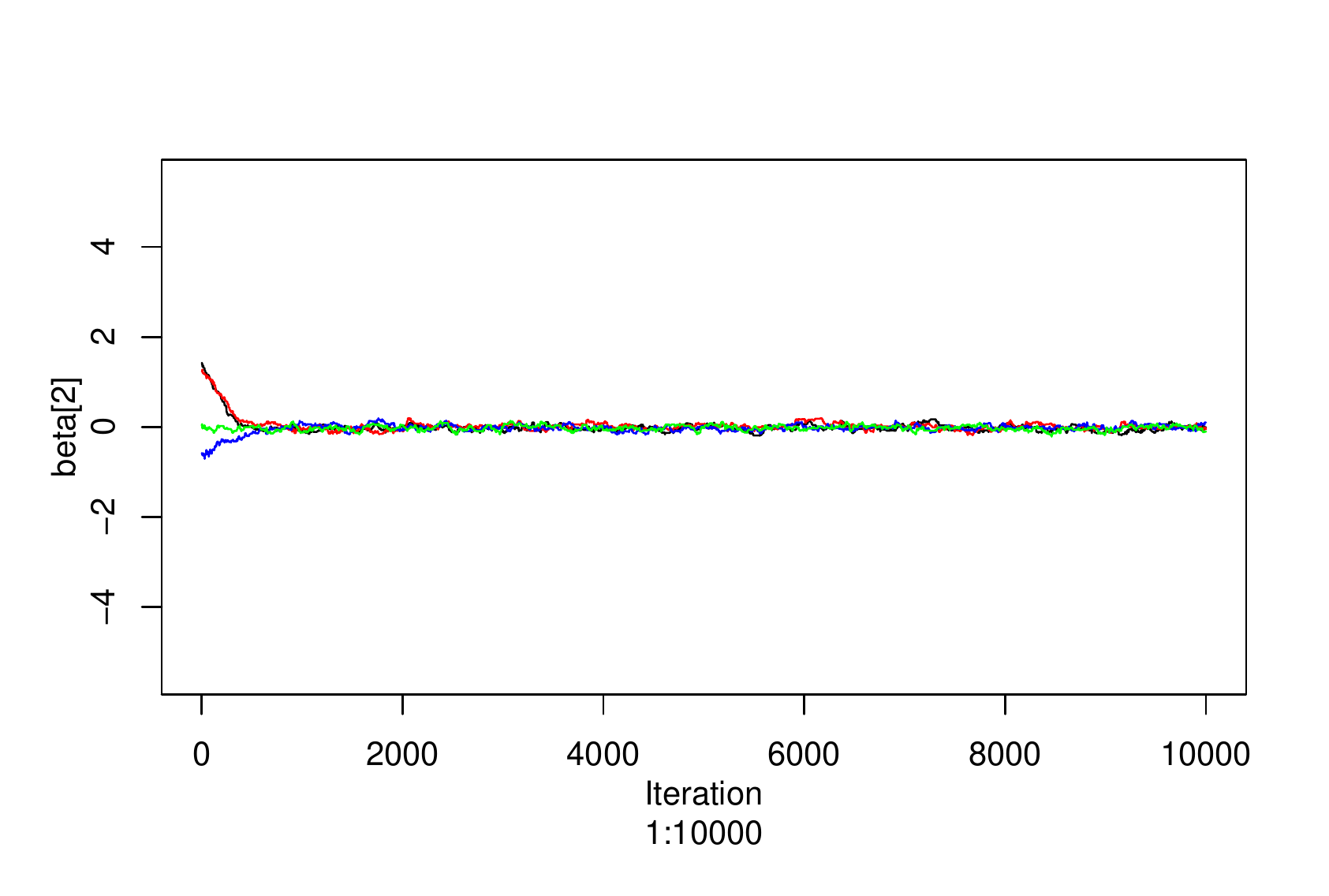}
  \caption{}
  \label{fig:multi2}
\end{subfigure}%
\vspace{0.1cm}
\begin{subfigure}{0.5\textwidth}
  \centering
  \includegraphics[scale=0.35]{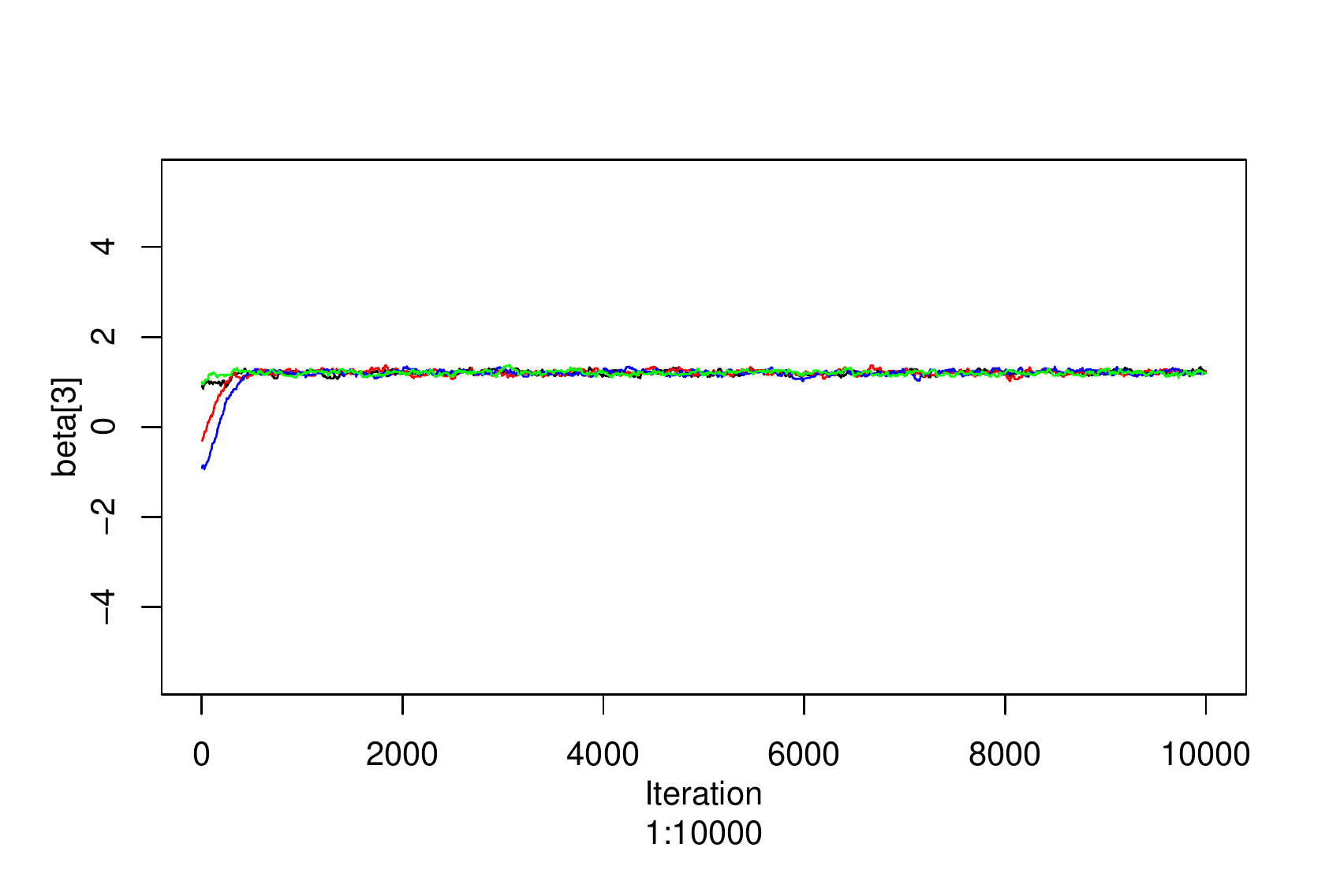}
  \caption{}
  \label{fig:multi3}
\end{subfigure}%
\caption{Multiple-trace plots of posterior chains for (a) $\beta_1$, (b) $\beta_2$ and (c) $\beta_3$, for the first 10,000 iterations.} 
\label{fig:regmultichain}
\end{figure}
\begin{figure}
\centering
\includegraphics[scale=0.8]{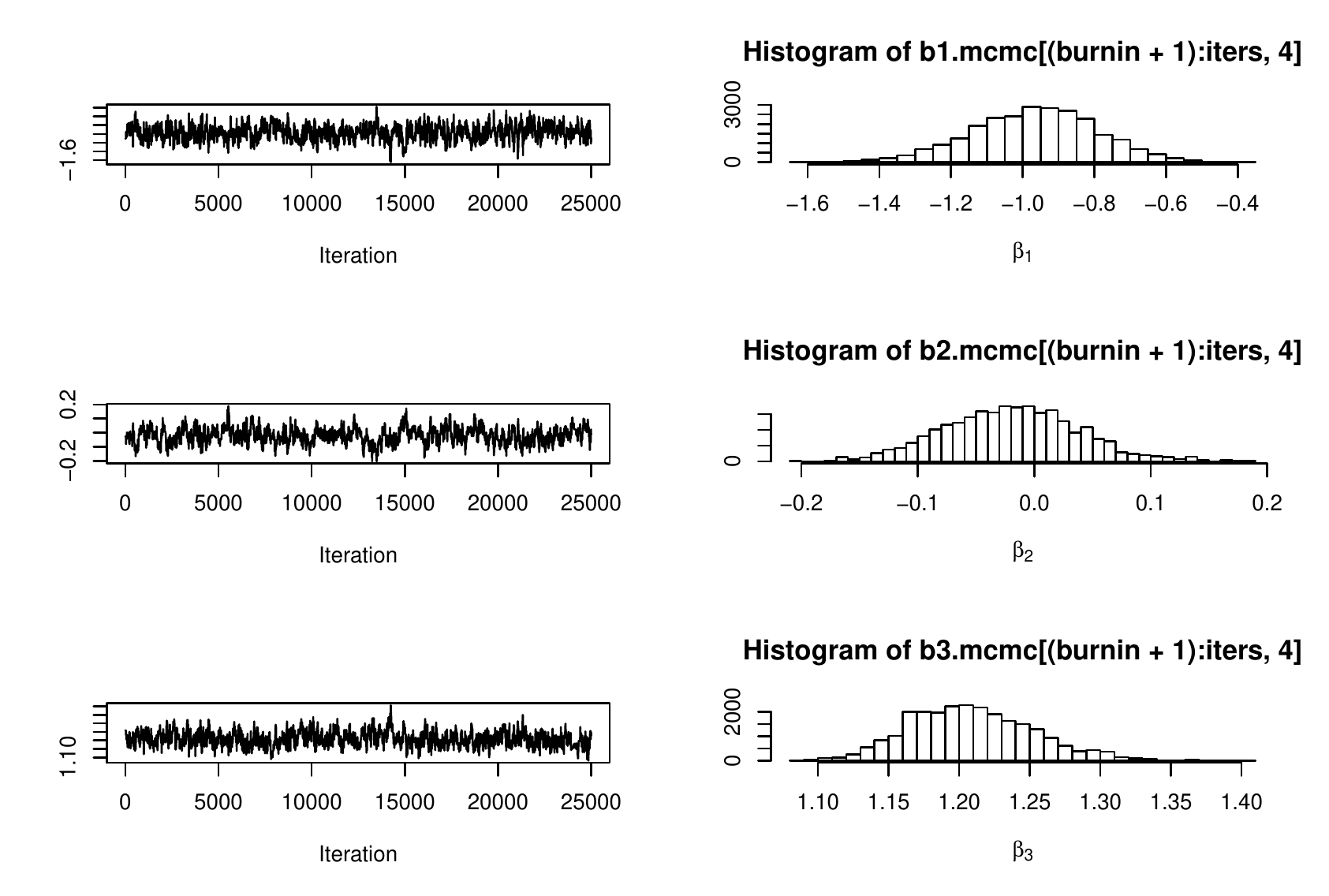}
\caption{Posterior trace plots and histograms for $\beta_1$ (top), $\beta_2$ (middle) and $\beta_3$ (bottom).}\label{fig:realreg_posttracehist}
\end{figure}

\begin{table}[H]
\centering
\begin{tabular}{|c|c|c|c|}
\hline 
 & \multicolumn{3}{c|}{Proposed Prior} \\ 
\hline 
 & Mean & sd & 95\% C.I. \\ 
\hline 
$\beta_1$ & -0.96 & 0.18 & (-1.35,-0.62) \\ 
\hline 
$\beta_2$ & -0.02 & 0.06 & (-0.13,0.09) \\ 
\hline 
$\beta_3$ & 1.21 & 0.05 & (1.13,1.30) \\ 
\hline 
\end{tabular} 
\quad
\begin{tabular}{|c|c|c|c|}
\hline 
 & \multicolumn{3}{c|}{Default Prior} \\ 
\hline 
 & Mean & sd & 95\% C.I. \\ 
\hline 
$\beta_1$ & -0.98 & 0.17 & (-1.33,-0.65) \\ 
\hline 
$\beta_2$ & -0.02 & 0.06 & (-0.13,0.09) \\ 
\hline 
$\beta_3$ & 1.21 & 0.05 & (1.12,1.31) \\ 
\hline 
\end{tabular} 
\caption{Posterior mean, standard deviation and 95\% credible interval for the real data analysis, under the proposed priors and the default priors.}\label{tab:realdatacomp}
\end{table}

\subsection{Simulation study}\label{sc_simulations}
For objective priors, it is common practice to analyse some frequentist properties of the yielded posteriors. Furthermore, when available, the performance is compared to the one of other objective priors. We consider two frequentist indexes: the coverage (COVE) of the $95\%$ credible interval of the posterior distribution and the square root of the mean squared error (MSE) from the posterior mean. Note that, for a scale parameter, say $\theta$, the MSE is computed relatively to the parameter value; that is $\sqrt{MSE(\theta)}/\theta$ (RMSE). This is necessary to have a more realistic quantification of the estimation accuracy, as one would expect the uncertainty to increase as the value of the parameter increases. The study compares the performances of the proposed prior with the ones of the appropriate Jeffreys prior, which can be considered as commonly accepted benchmark of reference.

It has to be noted that one would expect a larger RMSE when the proposed prior is used, in comparison to any model based prior. In fact, while the former used as information only the sample space, the latter includes the sampling distribution as well. Hence, the informational content of the posterior yielded by the proposed prior has to be less than the one of the posterior yielded by, say, the Jeffreys prior.

The first study considers the Poisson as sampling distribution. We repeat the procedure introduced in Section \ref{sc_application} on 250 samples drawn from a Poisson with $\theta=1,10,100,500$. The sample sizes considered are $n=3$, $n=10$, $n=30$ and $n=100$. The small sample sizes, although not very realistic, allow to better discern the performance of the proposed prior in comparison to the Jeffreys prior. Recalling that, if $x\sim\mbox{Poisson}(\theta)$, then the Jeffreys prior is $\pi_J(\theta)\propto\theta^{-1/2}$, we present the comparison in Tables \ref{tab:poiss3} to \ref{tab:poiss100}. 

\begin{table}[H]
\centering
\begin{tabular}{c|cc|cc}
\hline 
 & \multicolumn{2}{c|}{Proposed Prior} & \multicolumn{2}{c}{Jeffreys Prior} \\ 
\hline 
$\theta$ & RMSE & COVE & RMSE & COVE \\ 
\hline 
1 & 0.679 & 0.95 & 0.616 & 0.92 \\ 
10 & 0.184 & 0.94 & 0.180 & 0.96 \\ 
100 & 0.053 & 0.94 & 0.052 & 0.94 \\ 
500 & 0.026 & 0.92 & 0.025 & 0.94 \\ 
\hline 
\end{tabular} 
\caption{Frequentist analysis for $\theta$ for $n=3$, for the proposed objective prior and Jeffreys prior.}\label{tab:poiss3}
\end{table}

\begin{table}[H]
\centering
\begin{tabular}{c|cc|cc}
\hline 
 & \multicolumn{2}{c|}{Proposed Prior} & \multicolumn{2}{c}{Jeffreys Prior} \\ 
\hline 
$\theta$ & RMSE & COVE & RMSE & COVE \\ 
\hline 
1 & 0.316 & 0.97 & 0.300 & 0.96 \\ 
10 & 0.100 & 0.93 & 0.094 & 0.94 \\ 
100 & 0.034 & 0.92 & 0.033 & 0.93 \\ 
500 & 0.014 & 0.94 & 0.014 & 0.96 \\ 
\hline 
\end{tabular} 
\caption{Frequentist analysis for $\theta$ for $n=10$, for the proposed objective prior and Jeffreys prior.}\label{tab:poiss10}
\end{table}

\begin{table}[H]
\centering
\begin{tabular}{c|cc|cc}
\hline 
 & \multicolumn{2}{c|}{Proposed Prior} & \multicolumn{2}{c}{Jeffreys Prior} \\ 
\hline 
$\theta$ & RMSE & COVE & RMSE & COVE \\ 
\hline 
1 & 0.188 & 0.94 & 0.177 & 0.94 \\ 
10 & 0.056 & 0.98 & 0.054 & 0.95 \\ 
100 & 0.019 & 0.93 & 0.016 & 0.94 \\ 
500 & 0.008 & 0.94 & 0.007 & 0.97 \\ 
\hline 
\end{tabular} 
\caption{Frequentist analysis for $\theta$ for $n=30$, for the proposed objective prior and Jeffreys prior.}\label{tab:poiss30}
\end{table}

\begin{table}[H]
\centering
\begin{tabular}{c|cc|cc}
\hline 
 & \multicolumn{2}{c|}{Proposed Prior} & \multicolumn{2}{c}{Jeffreys Prior} \\ 
\hline 
$\theta$ & RMSE & COVE & RMSE & COVE \\ 
\hline 
1 & 0.093 & 0.95 & 0.091 & 0.97 \\ 
10 & 0.030 & 0.96 & 0.030 & 0.95 \\ 
100 & 0.010 & 0.96 & 0.009 & 0.94 \\ 
500 & 0.004 & 0.93 & 0.003 & 0.95 \\ 
\hline 
\end{tabular} 
\caption{Frequentist analysis for $\theta$ for $n=100$, for the proposed objective prior and Jeffreys prior.}\label{tab:poiss100}
\end{table}

As expected, the RMSE associated to the proposed prior is always larger than the RMSE yielded by the Jeffreys prior (the two exceptions are a mere consequence of rounding). However, the difference tends to be smaller the larger the sample size, given that the information in the data becomes stronger, so the prior information becomes less important. The coverage is in line with the behaviour one would expect from a minimally informative prior, and the MSE increases as the sample size decreases as there is less information contained in the data. A similar conclusion can be drawn for the second simulation study, where a normal with known variance is considered. In this case, with observations from $N(\mu,\sigma^2)$ (where $\sigma^2$ is assumed to be known), where $\mu=\{-5,-4,\ldots,0,\ldots,4,5\}$, Jeffreys prior is $\pi_J(\mu)\propto1$. The reported results have been obtained for $\sigma^2=1$ as there is no impact of the variance on the prior performances. 

By inspecting Figure \ref{fig:normal_freq}, we do not notice any sensible difference, in terms of frequentist performance, between the two priors, with the exception of a slightly larger MSE for the proposed prior (as expected).
\begin{figure}[h]
\begin{subfigure}{.5\textwidth}
  \centering
  \includegraphics[width=.8\linewidth]{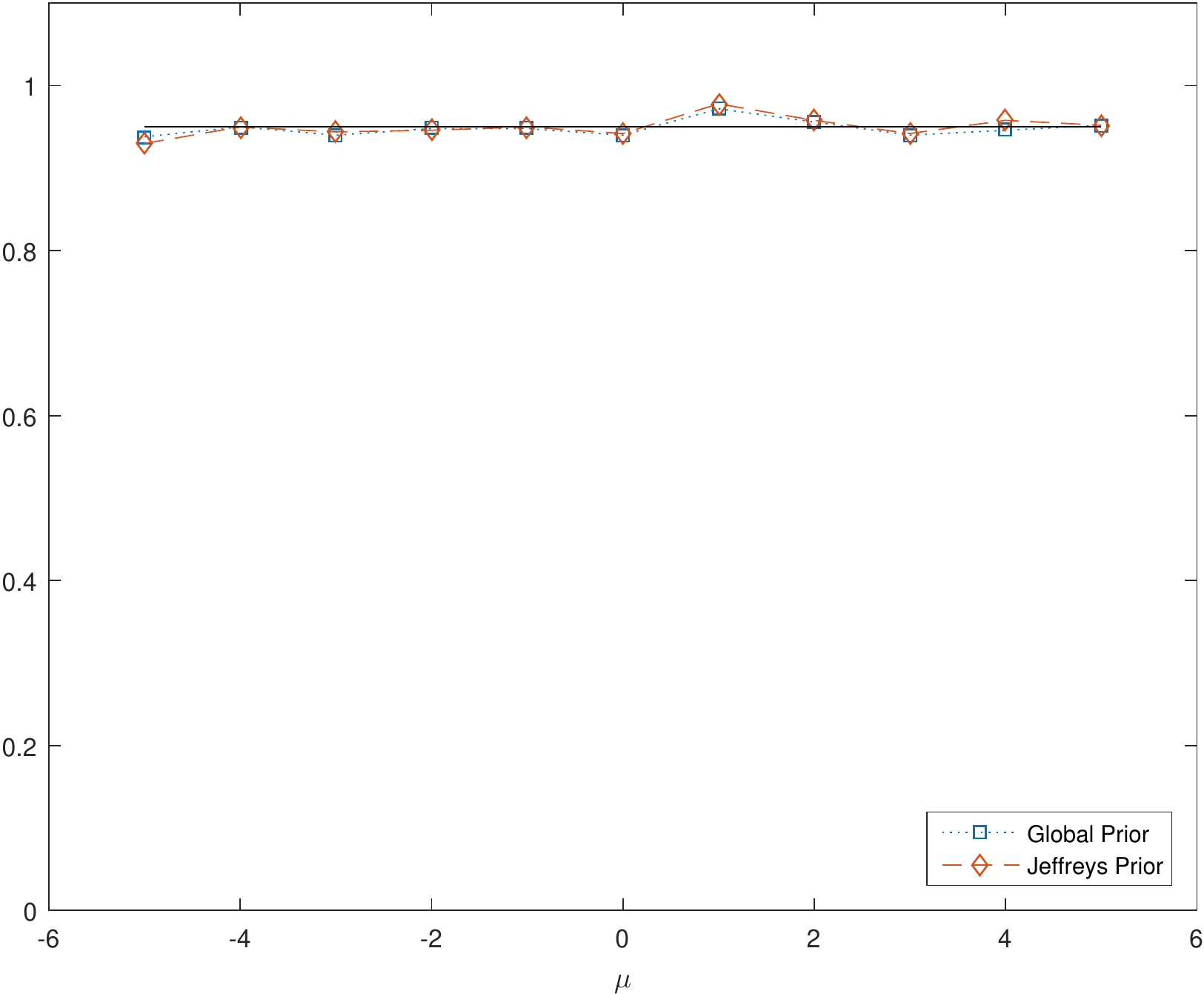}
  \caption{}
  \label{fig:norsfig11}
\end{subfigure}%
\begin{subfigure}{.5\textwidth}
  \centering
  \includegraphics[width=.8\linewidth]{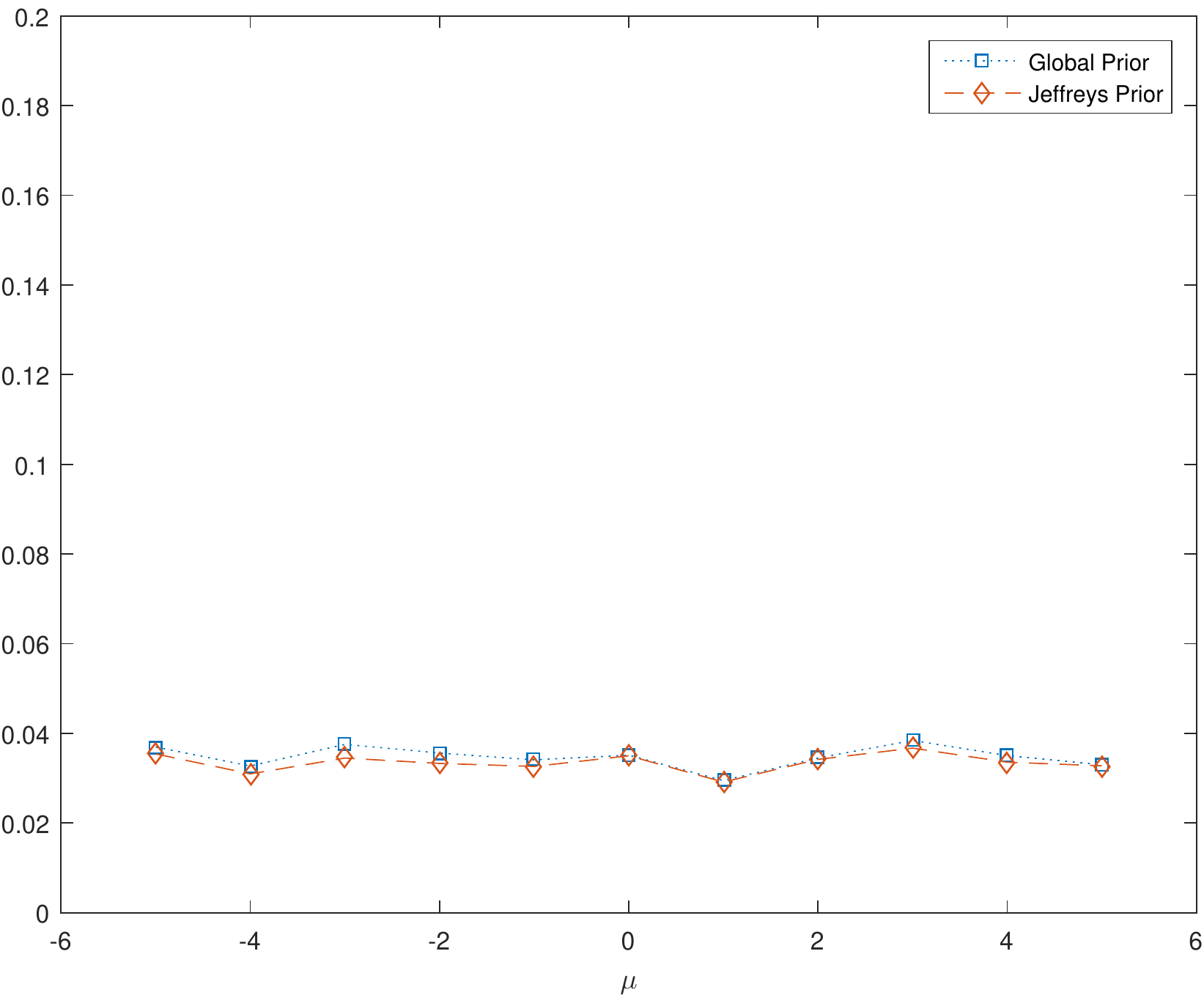}
  \caption{}
  \label{fig:norsfig12}
\end{subfigure}
\begin{subfigure}{.5\textwidth}
  \centering
  \includegraphics[width=.8\linewidth]{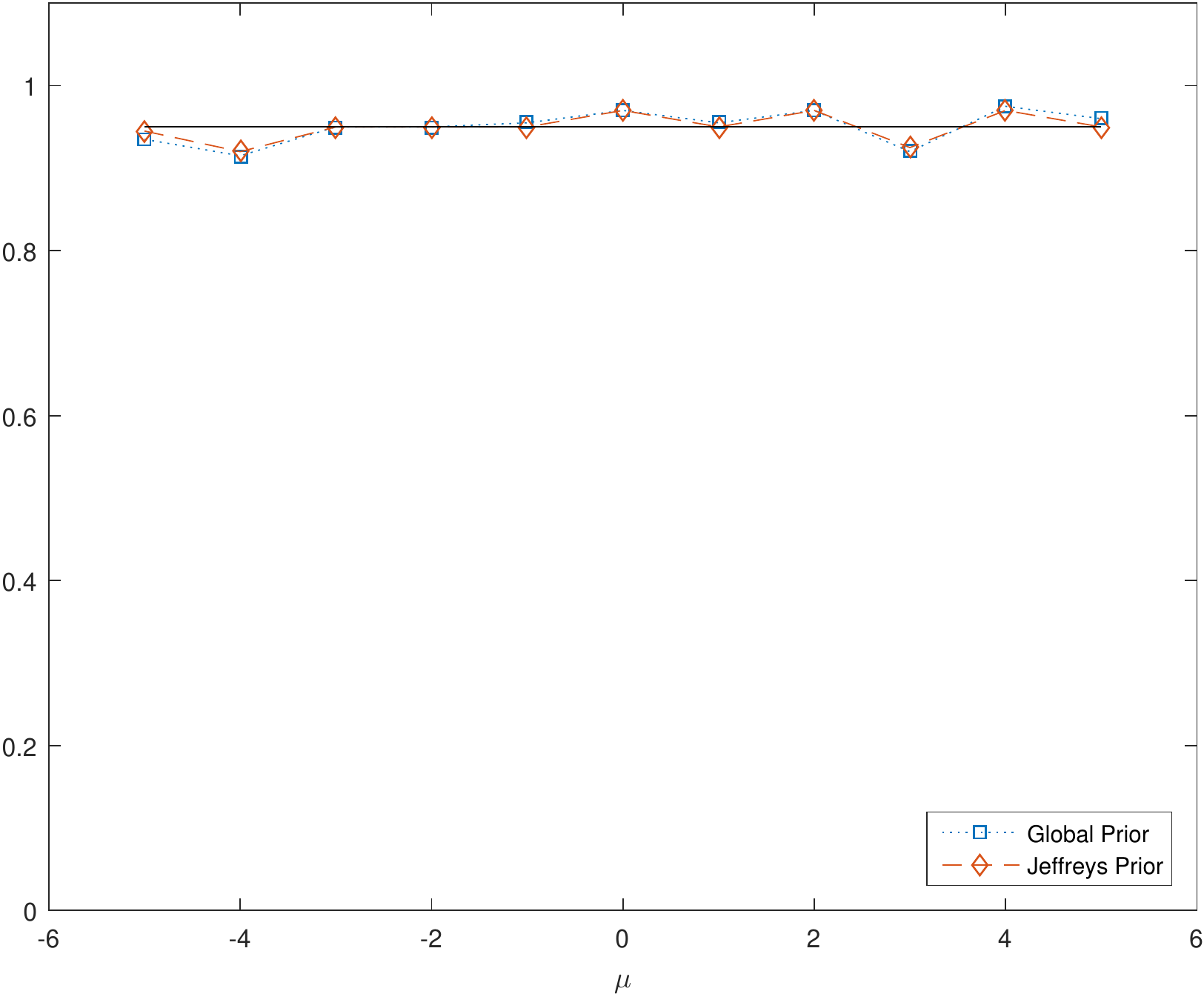}
  \caption{}
  \label{fig:norsfig21}
\end{subfigure}%
\begin{subfigure}{.5\textwidth}
  \centering
  \includegraphics[width=.8\linewidth]{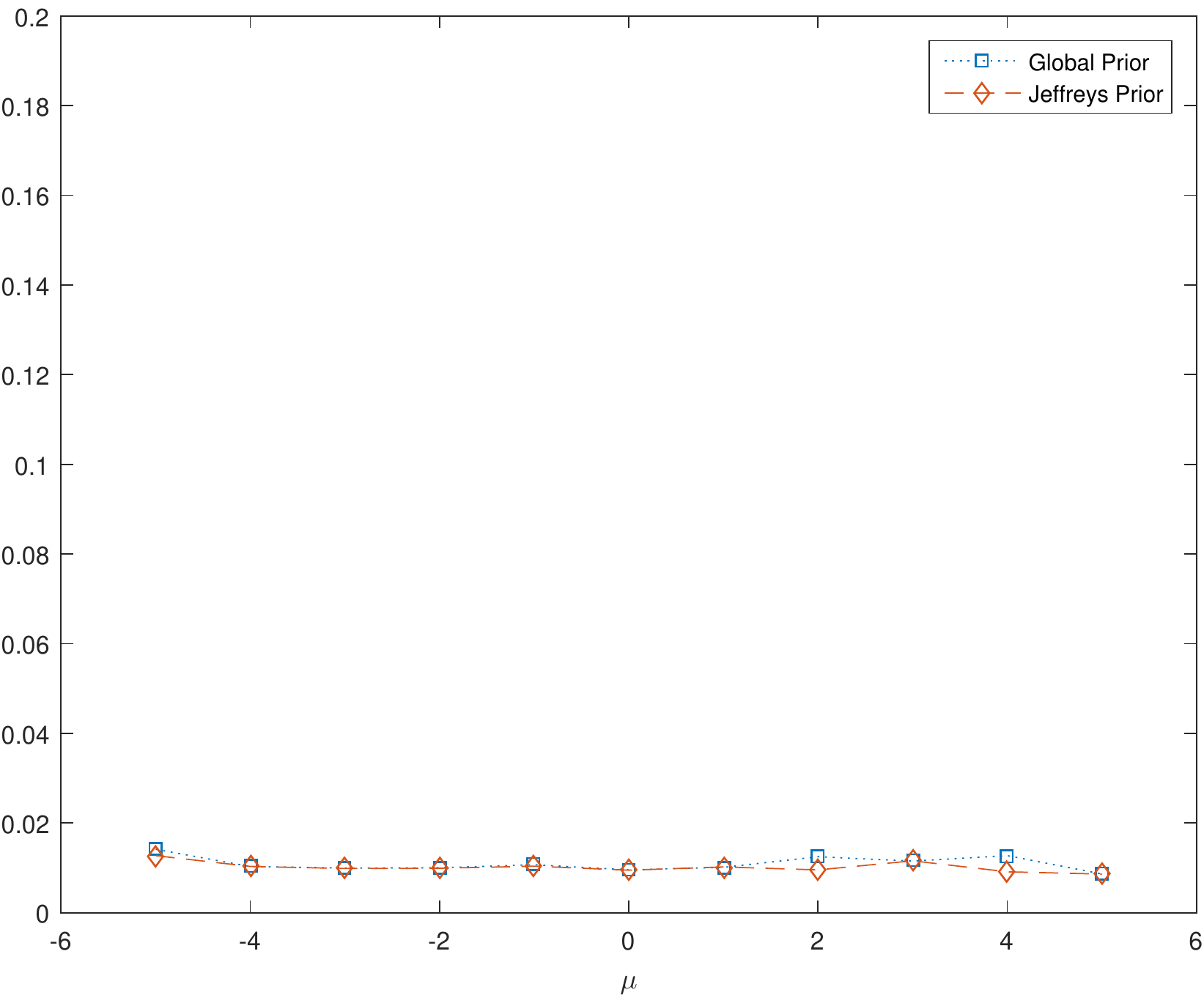}
  \caption{}
  \label{fig:norsfig22}
\end{subfigure}
\caption{Normal model. Coverage of the $95\%$ credible interval of the posterior for $\theta=1,2,\ldots,10$ for $n=30$ (a) and for $n=100$ (c), and MSE from the mean for $n=30$ (b) and for $n=100$ (d). Each graph shows the results for the proposed prior (blue square) and the Jeffreys prior (red diamond).}
\label{fig:normal_freq}
\end{figure}

\section{Discussion}\label{sc_discussion}
In this paper we have introduced a new class of objective priors derived from scoring rules. 
A remarkable aspect is that we have been able to show that the same result can be achieved via the rigour of calculus of variations, by finding objective priors which solve the Euler--Lagrange equation for finding extremum to integrals of the type
$$\int L(\theta,p,p')\,\d\theta.$$
If we can establish suitable choices of $L(\theta,p,p')$ which can be motivated and satisfy conditions for the existence of extremum, then new classes of objective prior can be sought. The case we have considered, which we can consider as a first step, is to use a combination of two well known measures of information in a prior density function; i.e.
$$L(\theta,p,p')=\half \frac{p'(\theta)^2}{p(\theta)}+p(\theta)\,\log p(\theta).$$
The objective priors here defined have two desirable properties. The first is that they are somewhat detached from the choice of the sampling distribution and are dependent on the parameter space only. In other words, the information required to derive the prior is limited to the range of values that the quantity of interest can take. 

The second property is that the prior can be proper. Besides the advantage of not having to check properness of the posterior, it allows to exploit the prior in scenarios where improper objective priors have been challenging. For example, as illustrated in Section \ref{sc_mixture}, the proposed prior is used to estimate the means of a mixture of normal densities with three components. Another potential application, discussed in Section \ref{sc_modcomp}, is in model selection. In particular, the objective prior may be used to represent minimal information on the parameters that are not common to two models. In fact, the Bayes factor used to compare two models is, in general, sensitive to the proportionality constant of improper priors. While for common parameters the constant will cancel out, this is not the case if the parameter is either at the numerator or at the denominator of the ratio only. Hence, the necessity of having a proper prior assigned to this kind of parameters.

The simulation study, aimed to compare the frequentist performances of the proposed prior with the ones of the Jeffreys prior, has shown no appreciable differences, with the exception of a slightly larger MSE for the proposed prior; the last result is expected as it is a consequence of the smaller information used to define the proposed prior in comparison with any model based objective prior.

\section*{Acknowledgements} The authors would like to thank Guido Consonni for feedback on an earlier draft of the paper. Fabrizio Leisen was supported by the European Community's Seventh Framework Programme [FP7/2007-2013] under grant agreement no: 630677.The third author is partially supported by NSF grant DMS 1612089.

\bibliographystyle{Natbib}

\begin{thebibliography}{00}

\bibitem[Berger(2006)]{Berger2006}
Berger, J.~O. (2006). The case for objective Bayesian analysis. {\em Bayesian Analysis}, {\bf 1}, 1--17.

\bibitem[Berger et al.(2009)]{BBS2009}
Berger, J.~O., Bernardo, J.~M. and Sun, D. (2009). The formal definition of reference priors. {\em Annals of Statistics}, {\bf 37}, 905--938.

\bibitem[Berger et al.(2012)]{BBS2012}
Berger, J.~O., Bernardo, J.~M. and Sun, D. (2012). Objective priors for discrete parameter spaces. {\em Journal of the American Statistical Association}, {\bf 107}, 636--648.

\bibitem[Berger et al.(2015)]{BBS2015}
Berger, J.~O., Bernardo, J.~M. and Sun, D. (2015). Overall objective priors (with discussion). {\em Bayesian Analysis}, {\bf 10}, 189--221.

\bibitem[Berger and Pericchi(1996)]{BergerPericchi1996}
Berger, J.~O. and Pericchi, L.~R. (1996). The intrinsic Bayes factor for model selection and prediction. {\em Journal of the American Statistical Association}, {\bf 91}, 109--122.

\bibitem[Berger et al.(1998)]{Bergeretal1998}
Berger, J.~O., Pericchi, L.~R., and Varshavsky, J. (1998). Bayes factors and marginal distributions in invariant situations. {\em Sankhya}, {\bf 60}, 109--122.

\bibitem[Berger and Strawderman(1993)]{BergerStraw1993}
Berger, J.~O. and Strawderman, W. (1993). Choice of hierarchical priors: admissibility in estimation of normal means. {\em Technical report}, 93-34C, Purdue University, Dept. of Statistics.

\bibitem[Bernardo(1979)]{Bernardo1979}
Bernardo, J.~M. (1979). Reference posterior distributions for Bayesian inference (with discussion). \emph{Journal of the Royal Statistical Society, Series B}, {\bf 41}, 113--147.

\bibitem[Bernardo(1997)]{Bernardo1997}
Bernardo, J.~M. (1979). Noninformative priors do not exist: a discussion with Jos\'{e} M. Bernardo. \emph{Journal of Statistical Planning and Inference}, {\bf 65}, 159--189.

\bibitem[Bernardo and Smith(1994)]{BS1994}
Bernardo, J.~M. and Smith, A.~F.~M. (1994). {\em Bayesian Theory}. John Wiley \& Sons, Inc., Hoboken, NJ, USA. doi: 10.1002/9780470316870.ch1. 

\bibitem[Bhattacharyya(1943)]{Bhatta1943}
Bhattacharyya, A.~K. (1943). On a measure of divergence between two statistical populations defined by their probability distribution. \emph{Bulletin of the Calcutta Mathematical Society}, {\bf 35}, 99--109.

\bibitem[Bobkov et al(2014)]{Bobkov2014}
Bobkov, S.~G., Gozlan, N., Roberto, C. and Samson, P.~M. (2014). Bounds on the deficit in the logarithmic Sobolev inequality. {\em Journal of Functional Analysis}, {\bf 267}, 4110--4138.

\bibitem[Box and Tiao(1973)]{BT1973}
Box, G.~E.~P. and Tiao, G.~C. (1973). {\em Bayesian Inference in Statistical Analysis}. Reading, MA: Addison-Wesley.

\bibitem[Clarke and Sun(1997)]{CS1997}
Clarke, B. and Sun, D. (1997). Reference priors under the chi-square distance. \emph{Sankhy\={a} A}, {\bf 59}, 215--231.

\bibitem[Clarke and Sun(1999)]{CS1999}
Clarke, B. and Sun, D. (1999). Asymptotics of the expected posterior. \emph{Annals of the Institute of Statistical Mathematics}, {\bf 51}, 163--185.

\bibitem[Clarke and Wasserman(1993)]{CW1993}
Clarke, B. and Wasserman, L. (1993). Noninformative priors and nuisance parameters. \emph{Journal of the American Statistical Association}, {\bf 88}, 1472--1432.

\bibitem[Clarke and Wasserman(1995)]{CW1995}
Clarke, B. and Wasserman, L. (1995). Information trade-off. \emph{TEST}, {\bf 4}, 19--38.

\bibitem[Consonni et al.(2013)]{CFR2013}
Consonni, G., Forster, J.~J. and La Rocca, L. (2013). The Whetstone and the  alum block: Balanced objective Bayesian comparison of nested models for discrete data.
{\em Statistical Science}, {\bf 28}, 398--423. 

\bibitem[Consonni et al.(2018)]{Consonni2018}
Consonni, G., Fouskakis, D., Liseo, B. and Ntzoufras, I. (2018). Prior distributions for objective Bayesian analysis. {\em Bayesian Analysis}, {\bf 13}, 627--679.

\bibitem[Dawid(1983)]{Dawid1983}
Dawid, A.~P. (1983). Invariant prior distributions. In \emph{Encyclopedia of Statistical Sciences}, eds. S. Kotz and N.~L. Johnson, New York: John Wiley.

\bibitem[Dey et al.(1993)]{Deyetal1993}
Dey, D.~K., Gelfand, A.~E. and Peng, F. (1993). Overdispersed generalized linear models. {\em Technical report}, University of Connecticut, Dept. of Statistics.

\bibitem[de Finetti(1974)]{Definetti1974}
de Finetti, B. (1974). \emph{Theory of Probability}, Wiley, New York.

\bibitem[Fonseca et al.(2008)]{Fonseca2008}
Fonseca, T.~C.~O. , Ferreira, M.~A.~R. and Migon, H.~S. (2008). Objective Bayesian analysis for the Student-T regression model. {\em Biometrika}, {\bf 95}, 325--333.

\bibitem[Ghosh(2011)]{Ghosh2011}
Ghosh, M. (2011). Objective priors: an introduction for frequentists. \emph{Statistical Science}, {\bf 26}, 187--202.

\bibitem[Ghosh et al.(2006)]{GDS2006}
Ghosh, J.~K., Delampady, M. and Samanta, T. (2006). \emph{An Introduction to Bayesian Analysis}. Springer, New York.

\bibitem[Ghosh et. al.(2011)]{Ghoshetal2011}
Ghosh, M., Mergel, V. and Liu, R. (2011). A general divergence criterion for prior selection. {\em Annals of the Institute of Statistical Mathematics}, {\bf 63}, 43--58.

\bibitem[Ghosh and Mukerjee(1992)]{GhoshMuke1992}
Ghosh, M. and Mukerjee, R. (1992). Non-informative priors. {\em Bayesian Statistics 4}, eds. J.~M. Bernardo, J.~O. Berger, A.~P. Dawid and A.~F.~M. Smith, Oxford, U.K.: Clarendon Press.

\bibitem[Goodrich and Lu(2013)]{GoodLu2013}
Goodrich, B. and Lu, Y. (2013). Poisson-Bayes: Bayesian Poisson regression
in Christine Choirat, Christopher Gandrud, James Honaker, Kosuke Imai, Gary King, and Olivia Lau, \emph{Zelig: Everyone's Statistical Software}, http://zeligproject.org/

\bibitem[Grazian and Robert(2018)]{Grazian2015}
Grazian, C. and Robert, C.~P. (2018). Jeffreys priors for mixture estimation: properties and alternatives. {\em Computational Statistics and Data Analysis}, {\bf 121}, 149--163.

\bibitem[Gronwall(1919)]{Gronwall1919}
Gronwall, T.~H. (1919). Note on the derivatives with respect to a parameter of the solutions of a system of differential equations. {\em Annals of Mathematics}, {\bf 20}, 292--296. 

\bibitem[Hartigan(1964)]{Harti1964}
Hartigan, J.~A. (1964). Invariant prior distributions. {\em Annals of Mathematical Statistics}, {\bf 35}, 836--845. 

\bibitem[Hyv\"{a}rinen(2005)]{Hyva2005}
Hyv\"{a}rinen, A. (2005). Estimation of non-normalized statistical models by score matching. {\em Journal of Machine Learning Research}, {\bf 6}, 695--709.

\bibitem[Ibrahim and Laud(1991)]{IbraLaud1991}
Ibrahim, J.~G. and Laud, P.~W. (1991). On Bayesian analysis of generalized linear models using Jeffreys's prior. {\em Journal of the American Statistical Association}, {\bf 86}, 981--986.

\bibitem[Kass(1990)]{Kass1990}
Kass, R.~E. (1990). Data-translate likelihood and and Jeffreys's rule. {\em Biometrika}, {\bf 77}, 107--114.

\bibitem[Kass and Wasserman(1996)]{KassWass1996}
Kass, R.~E. and Wasserman, L. (1996). The selection of prior distributions by formal rules. {\em Journal of the American Statistical Association}, {\bf 91}, 1343--1370.

\bibitem[Jaynes(1957)]{Jaynes1957}
Jaynes, E.~T. (1957). Information theory and statistical mechanics I, II. {\em Physical Review}, {\bf 106}, 620--630; {\bf 108}, 171--190.

\bibitem[Jaynes(1968)]{Jaynes1968}
Jaynes, E.~T. (1968). Prior probabilities. {\em IEEE Transactions on Systems Science and Cybernetics}, {\bf SSC-4}, 227--241.

\bibitem[Jeffreys(1946)]{Jeff1946}
Jeffreys, H. (1946). An invariant form for the prior probability in estimation problems. {\em Proceedings of the Royal Society of London}, Ser. A, {\bf 186}, 453--461.

\bibitem[Jeffreys(1961)]{Jeff1961}
Jeffreys, H. (1961). \emph{Theory of Probability and Inference}, 3rd ed. Cambridge University Press, London.

\bibitem[Laplace(1820)]{Lap1820}
Laplace, P.~S. (1820). {\em Essai Philosophique sue les Probbailit\'{e}s}.

\bibitem[Martin(1992)]{Martin1992}
Martin, L. (1992). \emph{Coercive Cooperation: Explaining Multilateral Economic Sanctions}. Princeton: Princeton University Press.

\bibitem[Merhav and Feder(1998)]{Merhav:Feder:1998}
Merhav, N. and Feder, M. (1998). Universal prediction. {\em IEEE Transactions on Information Theory}, {\bf 44}, 2124--2147.

\bibitem[Natarajan and McCulloch(1995)]{NataMcC1995}
Natarajan, R. and McCulloch, C.~E. (1995). A note on the existence of the posterior distribution for a class of mixed models for binomial responses. {\em Biometrika}, {\bf 82}, 639--643.

\bibitem[O'Hagan(1995)]{OHagan1995}
O'Hagan, A. (1995). Fractional Bayes factors for model comparison. {\em Journal of the Royal Statistical Society, Series B}, {\bf 57}, 99--138.

\bibitem[Parry et al.(2012)]{Parry2012}
Parry, M., Dawid, A. P. and Lauritzen S. (2012). Proper local scoring rules. {\em Annals of Statistics}, {\bf 40}, 561--592.

\bibitem[Rissanen(1983)]{Rissan1983}
Rissanen, J. (1983). A universal prior for integers and estimation by minimum description length. {\em Annals of Statistics}. {\bf 11}, 416--431.

\bibitem[Rubio and Liseo(2014)]{RB2014}
Rubio, F.~J. and Liseo, B. (2014). On the independence Jeffreys prior for skew-symmetric models. {\em Statistics \& Probability Letters}, {\bf 85}, 91--97.

\bibitem[Rubio and Steel(2018)]{RubioSteel2018}
Rubio, F.~J. and Steel, M.~F.~J. (2018). Flexible linear mixed models with improper priors for longitudinal and survival data. {\em Electronic Journal of Statistics}, {\bf 12}, 572--598.

\bibitem[Rustagi(1976)]{Rustagi1976}
Rustagi, J.S. (1976). \emph{Variational Methods in Statistics}. Academic Press.

\bibitem[Shannon(1948)]{Shannon1948}
Shannon, C.~E. (1948). A mathematical theory of communication. {\em Bell System Technical Journal}, {\bf 27}, 379--423.

\bibitem[Stone and Dawid(1972)]{StoneDavid72}
Stone, M. and Dawid, A. (1972). Un-Bayesian implications of improper Bayes inference in
routine statistical problems. {\em Biometrika}, {\bf 59}, 369--375.

\bibitem[Syversveen(1998)]{Syver1998}
Syversveen, A.~R. (1998). Noninformative Bayesian priors. Interpretation and problems with construction and applications. {\em Technical Report}.

\bibitem[Simpson et al.(2017)]{Simpson2017}
Simpson, D., Rue, H., Riebler, A., Martins, T.~G. and S\o rbye, S.~H. (2017). Penalising model component complexity: a principled, practical approach to constructing priors. {\em Statistical Science}, {\bf 32}, 1--28.

\bibitem[Sweeting(2008)]{Swee2008}
Sweeting, T.~J. (2008). On predictive probability matching priors. {\em IMS Collections: Pushing the Limits of Contemporary Statistics: Contributions in Honor of Jayanta K. Ghosh}, eds. B. Clarke and S. Ghosal. {\bf 3}, 46-59.

\bibitem[Sweeting(2011)]{Swee2011}
Sweeting, T.~J. (2011). Invited discussion of M. Ghosh : Objective priors: an introduction for frequentists. {\em Statistical Science} {\bf 26}, 206-209.

\bibitem[Sweeting et al.(2006)]{Sweetetal2006}
Sweeting, T.~J., Datta, G.~S. and Ghosh, M. (2006). Nonsubjective priors via predictive relative entropy loss. {\em Annals of Statistics}. {\bf 34}, 441-468.

\bibitem[Titterington et al.(1985)]{Titt1985}
Titterington, D., Smith, A. and Makov, U. (1985). {\em Statistical Analysis of Finite Mixture Distributions}. John Wiley, New York.

\bibitem[Vallejos and Steel(2013)]{ValleSteel2013}
Vallejos, C.~A. and Steel, M.~J.~F. (2013). {\em On posterior propriety for the Student-$t$ linear regression model under Jeffreys priors}. arxiv:1311.1454.

\bibitem[Villa and Walker(2015)]{VW2015}
Villa, C. and Walker, S.~G. (2015). An objective approach to prior mass functions for discrete parameter spaces. \emph{Journal of the American Statistical Association}, {\bf 110}, 1072--1082.

\bibitem[Walker(2016)]{W16}
Walker, S.~G. (2016). Bayesian information in an experiment and the Fisher information distance. \emph{Statistics and Probability Letters}, {\bf 112}, 5--9. 

\bibitem[Welch and Peers(1963)]{WP1963}
Welch, B.~L. and Peers, H.~W. (1963). On formulae for confidence points based on integrals of weighted likelihoods. \emph{Journal of the Royal Statistical Society, Series B}, {\bf 35}, 318--329.

\bibitem[Yand and Chen(1995)]{YC1995}
Yang, R. and Chen, M.~H. (1995). Bayesian analysis for random coefficient regression models using noninformative priors. {\em Journal of Multivariate Analysis}, {\bf 55}, 283--311

\bibitem[Zellner and Min(1993)]{ZellMin1993}
Zellner, A. and Min, C. (1993). Bayesian analysis model selection and prediction. In {\em Physics and Probability: Essays in honour of Edwin T. Jaynes}, eds. W.~T. Grandy, Jr. and P.~W. Milonni, Cambridge, U.K.: Cambridge University Press.

\end{thebibliography}

\newpage
\appendix
\section*{Appendix - MCMC algorithm}
The following algorithm produces a MCMC sample $\theta^{(0)}, \theta^{(1)}, ...,\theta^{(T)}$ from the posterior distribution. Suppose to fix the initial value $\theta^{(0)}$ and repeat the following procedure for $t=1,\dots,T$. 
\begin{center}
\begin{algorithm}{Metropolis-Hastings for the prior}
\label{alg1}
\par\vspace{5pt}\hrule
\bigskip
Suppose that $\theta^{(t)}$ is the current state of the chain:
\vspace{0.3cm}
\begin{enumerate}
\item Draw $\theta'$ from a proposal distribution $q(\cdot|\theta^{(t)})$
\item Evaluate $u(\theta')-u(\theta^{(t)})$ numerically, via 
\begin{equation*}
\begin{split}
u(\theta')&=u(\theta^{(t)})+(\theta'-\theta^{(t)})u'(\theta^{(t)})+\frac{(\theta'-\theta^{(t)})^2}{2}u''(\theta^{(t)})\\&+\frac{(\theta'-\theta^{(t)})^3}{6}u'''(\theta^{(t)})+R,
\end{split}
\end{equation*}
where
\begin{eqnarray}
u'(\theta^{(t)}) &=& \sqrt{ce^{u(\theta^{(t)})}-2(1+u(\theta^{(t)}))}\nonumber\\
u''(\theta^{(t)}) &=& \half ce^{u(\theta^{(t)})}-1\nonumber\\
u'''(\theta^{(t)}) &=& \half ce^{u(\theta^{(t)})}\sqrt{ce^{u(\theta^{(t)})}-2(1+u(\theta^{(t)}))}\nonumber
\end{eqnarray}
\item Compute $p(\theta')/p(\theta^{(t)})=\exp\{-[u(\theta')-u(\theta^{(t)})]\}$
\item Set $\theta^{(t+1)}=\theta'$ with probability
\begin{equation}
\alpha=\min\left\lbrace 1,\frac{l(\theta')}{l(\theta^{(t)})}\frac{p(\theta')}{p(\theta^{(t)})}\frac{q(\theta^{(t)}|\theta')}{q(\theta'|\theta^{(t)})}\right\rbrace,\nonumber
\end{equation}
and $\theta^{(t+1)}=\theta^{(t)}$ with probability $1-\alpha$
\end{enumerate}
\end{algorithm}
\hrule\vspace{5pt}
\end{center}
Depending on how far $\theta'$ is from $\theta$ we can either use the direct approximation in step 3 with higher order derivatives or otherwise get from $\theta$ to $\theta'$ using smaller step sizes.

\end{document}